\documentclass[12pt,a4paper]{amsart}
\usepackage{amssymb}
\usepackage{amsmath}
\usepackage{xspace}
\usepackage[german,english]{babel}
\usepackage[latin1]{inputenc}
\usepackage[all,2cell]{xy}
\usepackage{color}

\usepackage{dsfont}

\newcommand{\at}[1]%
            {\ensuremath{\protect\underline{\mathbf{#1}}}} 

\newcommand{\op}[1]{\ensuremath{\operatorname{#1}}}        
\newcommand{\farg}{\cdot}                                  

\newcommand{\h}[1][]                                       
 {\ifthenelse{\boolean{mmode}}%
  {$\mathrm{h}$}%
  {h\nobreakdash#1\hspace{0pt}}}

\DeclareMathOperator{\U}{\boldsymbol{\mathcal{U}}}   
\newcommand{\comp}{\circ}          
\newcommand{\adcomp}%
  {\overset{\operatorname{ad}}{\comp}} 
\newcommand{\funcomp}%
  {\overset{\operatorname{fn}}{\comp}}

\newcommand{\sccat}
{\mathbin{\kern-1pt\raisebox{6pt}{.}\kern-5pt
\downarrow\kern-5pt\raisebox{6pt}{.}\kern-1pt}}

\newcommand{\parrow}[1]
   {\underset{{\displaystyle \raisebox{5pt}%
   {$\longleftarrow$}}}{\op{#1}}{\,}}
\newcommand{\iarrow}[1]
   {\underset{{\displaystyle \raisebox{5pt}%
   {$\longrightarrow$}}}{\op{#1}}{\,}}

\DeclareMathOperator{\Sub}{Sub}    
\DeclareMathOperator{\inc}{in}     
\newcommand{\rest}%
{\mathnormal{\restriction}}        

\newcommand{\nin}{\not\in}         
\newcommand{\bprod}{\times}        
\newcommand{\bcoprod}{\amalg}      
\newcommand{\function}[4]{
            \begin{array}{@{\:}c@{\:}c@{\:}l}
                   #1 &\mor& #2 \\
                   #3 &\longmapsto& #4
            \end{array} }
\newcommand{\nfunction}[4]
    {\left\{
     \function{#1}{#2}{#3}{#4}
     \right. }

\newcommand{\fmon}[1]
 {\ensuremath{#1^{\star}}}

\newcommand{\bb}[1]{\ensuremath
 {\lvert #1 \rvert}}
\DeclareMathOperator{\Sg}{Sg}      
\DeclareMathOperator{\Cgr}{Cgr}    
\newcommand{\pr}{\mathrm{pr}}        

\DeclareMathOperator{\G}{G}        
\DeclareMathOperator{\bconcat}
            {\curlywedge}

\newcommand{\concat}
  {\ensuremath{\text
  {\Large $\curlywedge$}}}

\newcommand{\vacio}{\ensuremath{\varnothing}}

\newcommand{\brel}{\ensuremath{\xymatrix{{}\ar@{{*}{-}{*}}[r] & {}}}}

\newcommand{\nseq}[3]{\xymatrix@1@C=16pt{#1 \ar@{>}[r]_-{\scriptscriptstyle{#2}} & #3 }}

\NoCompileMatrices              
\xymatrixrowsep = {8ex}         
\xymatrixcolsep = {10ex}        

\newdir{<= }{:a(+25)\dir{-}*:a(-25)\dir{-}*!/:a(85)+1.8pt/:a(-25)\dir{-}}
\newdir{ = }{*@{=}}
\newdir{+>}{@{|}*@{-}!/-6.5pt/@{>}*!/-13pt/{}}
\newdir{ +}{{}*!/-5pt/@{+}}
\newdir{ >}{{}*!/-10pt/@{>}}
\newdir{.-}{{}*:a(-90)h!/2.1pt/{.}}
\newdir{.->}{:a(-90)h!/2.5pt/{.}*!/4pt/\dir{-}*!/0pt/\dir{-}*%
!/-2pt/\dir{-}*!/-8pt/\dir{>}}
\newdir{-.->}%
{:a(-90)h!/2pt/{.}*!/8pt/\dir{-}*!/4pt/\dir{-}*!/0pt/\dir{-}*%
!/-5pt/\dir{-}*!/-11pt/\dir{>}}
\newdir{~>}{!/4.6pt/\dir{~}*!/1pt/\dir{-}*!/-5pt/\dir{>}}
\newdir{~~>}{!/7pt/\dir{~}*!/0pt/\dir{~}*!/-5pt/\dir{>}}
\newdir{=~>}{*!/1pt/@2{~}*!/6pt/{}*!/-6pt/@2{>}}
\newdir{<~~>}{!/21pt/\dir{<}*!/21pt/\dir{-}!/12pt/\dir{~}!/5pt/\dir{~}*!/1pt/\dir{-}*!/-3pt/\dir{>}}
\newdir{=>}{!/5pt/\dir{=}!/2.5pt/\dir{=}*!/-5pt/\dir2{>}*!/7pt/\dir{ }}
\newdir{==>}{!/-2pt/\dir{=}!/-6pt/\dir{=}%
 !/-10pt/\dir{=}*!/-18pt/\dir2{>}}
\newdir{<==>}{!\dir2{<}!/-2pt/\dir{=}!/-6pt/\dir{=}%
 !/-10pt/\dir{=}*!/-17pt/\dir2{>}}
\newdir{< }{{}*!/12pt/\dir{<}}
\newdir{-o}{!/1.3pt/{}*!/0pt/{}@{o}*!/-5pt/{}}

\newsavebox{\xymor}  
\newsavebox{\xymon}  
\newsavebox{\xyepi}  
\newsavebox{\xytn}   
\newsavebox{\xyrel}  
\newsavebox{\xycel}  
\newsavebox{\xymdf}  
\newsavebox{\xyumor} 
\newsavebox{\xydmor} 
\newsavebox{\xyomor} 
\newsavebox{\xyemor} 

\newcommand{\xynode}{\makebox[0ex]{}}
\savebox{\xymor}{\ensuremath{%
\xymatrix@1@C=19pt{\xynode \ar@{>}[r] & \xynode }}}
\savebox{\xymon}{\ensuremath{%
\xymatrix@1@C=19pt{\xynode \ar@{{ +}{-}{>}}[r] & \xynode }}}
\savebox{\xyepi}{\ensuremath{%
\xymatrix@1@C=19pt{\xynode \ar@{{}{-}{+>}}[r] & \xynode }}}
\savebox{\xytn}{\ensuremath{%
\xymatrix@1@C=19pt{\xynode \ar[r]|(.44){\object@{.-}} & \xynode
}}}
\savebox{\xyrel}{\ensuremath{%
\xymatrix@1@C=19pt{\xynode \ar@{{}{-}{-o}}[r] & \xynode }}}
\savebox{\xycel}{\ensuremath{%
\xymatrix@1@C=19pt{\xynode \ar@{=>}[r] & \xynode }}}
\savebox{\xymdf}{\ensuremath{%
\xymatrix@1@C=16pt{\xynode \ar@{}[r]|{\dir{~>}} & \xynode}}}
\savebox{\xyumor}{\ensuremath{%
\xymatrix@1@C=19pt{\xynode \ar@{{}{-}^{>}}[r] & \xynode }}}
\savebox{\xydmor}{\ensuremath{%
\xymatrix@1@C=19pt{\xynode \ar@{{}{-}_{>}}[r] & \xynode }}}
\savebox{\xyomor}{\ensuremath{%
\xymatrix@1@C=19pt{\xynode \ar@{{}{-}^{< }}[r] & \xynode }}}
\savebox{\xyemor}{\ensuremath{%
\xymatrix@1@C=19pt{\xynode \ar@{{ >}{-}{>}}[r] & \xynode }}}

\newcommand{\mor}{\usebox{\xymor}}    

\newcommand{\functor}[9]{
 \xymatrix{
    #4 \save[]+<0ex,5ex>*+{#1}="1"  \restore
      \ar[d]_{#6}  \ar@{}[rd]|{\longmapsto}
  & #5 \save[]+<0ex,5ex>*+{#3}="3"  \restore
      \ar[d]^{#7}
  \\
   #8 & #9 \ar "1";"3"^-{#2} } }
\newcommand{\functornd}[9]{
 \xymatrix{
    #4 \save[]+<0ex,5ex>*+{#1}="1"  \restore
      \ar[d]_{#6}  \ar@{}[rd]|{\longmapsto}
  & #5 \save[]+<0ex,5ex>*+{#3}="3"  \restore
  \\
   #8 & #9 \ar[u]_{#7} \ar "1";"3"^-{#2} } }
\newcommand{\functordn}[9]{
 \xymatrix{
    #4 \save[]+<0ex,5ex>*+{#1}="1"  \restore
       \ar@{}[rd]|{\longmapsto}
  & #5 \save[]+<0ex,5ex>*+{#3}="3"  \restore
      \ar[d]^{#7}
  \\
   #8  \ar[u]^{#6}  & #9 \ar "1";"3"^-{#2} } }

\newcommand{\larr}{->}
\newcommand{\rarr}{->}

\newcommand{\xfunctor}[9]{
 \xymatrix{
    #4 \save[]+<0ex,5ex>*+{#1}="1"  \restore
      \ifthenelse{\equal{\larr}{->}}{\ar[d]_{#6}}{}
      \ifthenelse{\equal{\larr}{<-}}{\ar[d];[]^{#6}}{}
      \ifthenelse{\equal{\larr}{-<}}{\ar@{< }[d]_{#6}}{}
      \ar@{}[rd]|{\longmapsto}
  & #5 \save[]+<0ex,5ex>*+{#3}="3"  \restore
      \ifthenelse{\equal{\rarr}{->}}{\ar[d]^{#7}}{}
      \ifthenelse{\equal{\rarr}{<-}}{\ar[d];[]_{#7}}{}
      \ifthenelse{\equal{\rarr}{-<}}{\ar@{< }[d]^{#7}}{}
  \\
   #8 & #9 \ar "1";"3"^-{#2} } }

\UseAllTwocells

\theoremstyle{plain}
\newtheorem{theorem}{\indent\bf Theorem}[section]
\newtheorem{proposition}[theorem]{\indent\bf Proposition}
\newtheorem{corollary}[theorem]{\indent\bf Corollary}
\newtheorem{lemma}[theorem]{\indent\bf Lemma}

\theoremstyle{definition}
\newtheorem{definition}[theorem]{\indent\bf Definition}

\newtheorem*{remark}{\indent\bf Remark}

\theoremstyle{remark}

\numberwithin{equation}{section}

\begin{document}
\title[Eilenberg theorems for many-sorted formations]{Eilenberg theorems for many-sorted formations}

\author[Climent]{J. Climent Vidal}
\address{Universitat de Val\`{e}ncia\\
         Departament de L\`{o}gica i Filosofia de la Ci\`{e}ncia\\
         Av. Blasco Ib\'{a}\~{n}ez, 30-$7^{\mathrm{a}}$, 46010 Val\`{e}ncia, Spain}
\email{Juan.B.Climent@uv.es}
\author[Cosme]{E. Cosme Ll\'{o}pez}
\address{Universitat de Val\`{e}ncia\\
         Departament d'\`{A}lgebra\\
         Dr. Moliner, 50, 46100 Burjassot, Val\`{e}ncia, Spain}
\email{enric.cosme@uv.es}
\thanks{The research of the second author has been supported by the grant MTM2014-54707-C3-1-P from the \emph{Ministerio de Econom\'{\i}a y Competitividad} (Spanish Government) and FEDER (European Union).}

\subjclass[2010]{Primary: 08A68; Secondary: 08A70, 68Q70.} \keywords{Many-sorted algebra, support, many-sorted congruence, saturation, cogenerated congruence, many-sorted (finite) algebra formation, many-sorted (finite index) congruence formation, many-sorted regular language formation.}
\date{January 6th, 2016}

\begin{abstract}
A theorem of Eilenberg establishes that there exists a bijection between the set of all varieties of regular languages and the set of all varieties of finite monoids. In this article after defining, for a fixed set of sorts $S$ and a fixed  $S$-sorted signature $\Sigma$, the concepts of formation of congruences with respect to $\Sigma$ and of formation of $\Sigma$-algebras, we prove that
the algebraic lattices of all $\Sigma$-congruence formations and of all $\Sigma$-algebra formations are isomorphic, which is an Eilenberg's type theorem. Moreover, under a suitable condition on the free $\Sigma$-algebras and after defining the concepts of formation of congruences of finite index with respect to $\Sigma$, of formation of finite $\Sigma$-algebras, and of formation of regular languages with respect to $\Sigma$, we prove that the algebraic lattices of all $\Sigma$-finite index congruence formations, of all $\Sigma$-finite algebra formations, and of all $\Sigma$-regular language formations are isomorphic, which is also an Eilenberg's type theorem.
\end{abstract}
\maketitle



\section{Introduction}

In the development of the theory of regular languages the definition
and characterization of the \emph{varieties} ($\ast$-\emph{varieties}) of regular languages by Samuel Eilenberg (see~\cite{se76}, pp.~193--194), was crucial. Such a variety is a set of languages closed under the Boolean operations, left and right quotients by words, and inverse homomorphic images.  Eilenberg's main result is that varieties of regular languages are characterized by their \emph{syntactic semigroups}, and that the corresponding sets of finite semigroups are those that are closed under subsemigroups, quotient semigroups, and finite products of semigroups. Eilenberg called these sets \emph{varieties} of finite semigroups (see~\cite{se76},  p.~109). As it is well-known, one of the most important theorems in the study of formal languages and automata is the variety theorem of Eilenberg (see~\cite{se76}, p.~194), which states that there exists a bijection between the set of all varieties of regular languages and the set of all varieties of finite monoids. Eilenberg's work had as one of its consequences that of putting scattered results on diverse classes of languages into a general setting, and most of the subsequent work on regular languages can be properly viewed as taking place in this theoretical framework.

Several extensions of Eilenberg's theorem, obtained by replacing monoids by other algebraic constructs or by modifying the definition of variety of regular languages, have been considered in recent times. A further step in this research program has been to replace the varieties of finite monoids by the more general concept of formation of finite monoids (see~\cite{bps14}, p.~1740),  that is, a set of finite monoids closed under isomorphisms, homomorphic images, and finite subdirect products. The just mentioned replacement is founded, in the end, on the great significance that the (saturated) formations of finite groups---introduced by Wolfgang Gasch\"{u}tz in~\cite{gas62}---have, in particular, for a better understanding of the structure of the finite groups.
Perhaps it is appropriate at this point to recall that Gasch\"{u}tz, in~\cite{gas62} on p.~300, defines a formation as follows:
\begin{quotation}
Eine \emph{Menge} [\emph{we emphasize}] $F$ von Gruppen mit den Eigenschaften
\begin{enumerate}
\item[(2.1)] $\mathfrak{G}\in F$, $\mathfrak{G}^{\varphi}$ homomorphes Bild von $\mathfrak{G}$  $ \Rightarrow  $ $\mathfrak{G}^{\varphi}\in F$,
\item[(2.2)] $\mathfrak{N}_{1}$, $\mathfrak{N}_{2}$ Normalteiler von $\mathfrak{G}$, $\mathfrak{G}/\mathfrak{N}_{1}\in F$, $\mathfrak{G}/\mathfrak{N}_{2}\in F$ $ \Rightarrow $ $\mathfrak{G}/\mathfrak{N}_{1}\cap\mathfrak{N}_{2}\in F$
\end{enumerate}
hei{\ss}e Formation.
\end{quotation}

For the purposes of the present introduction, the following terminology is used: By $\boldsymbol{\mathcal{U}}$ we mean a fixed Grothendieck universe; by $\boldsymbol{\mathcal{U}}^{S}$, for a set of sorts $S\in \boldsymbol{\mathcal{U}}$, the set of all $S$-sorted sets, i.e., mappings $A$ from $S$ to $\boldsymbol{\mathcal{U}}$; by $\mathbf{T}_{\Sigma}(A)$ the free $\Sigma$-algebra on the $S$-sorted set $A$; by $\mathbf{Cgr}(\mathbf{T}_{\Sigma}(A))$ the algebraic lattice of all congruences on $\mathbf{T}_{\Sigma}(A)$; by $\mathrm{Filt}(\mathbf{Cgr}(\mathbf{T}_{\Sigma}(A)))$ the set of all filters of $\mathbf{Cgr}(\mathbf{T}_{\Sigma}(A))$; for a congruence $\Theta$ on $\mathbf{T}_{\Sigma}(B)$, by $\mathrm{pr}^{\Theta}$ the canonical projection from $\mathbf{T}_{\Sigma}(B)$ to $\mathbf{T}_{\Sigma}(B)/\Theta$; for an $L\subseteq \mathrm{T}_{\Sigma}(A)$,  by $\Omega^{\mathbf{T}_{\Sigma}(A)}(L)$ the greatest congruence which saturates $L$; and, under a condition on every free $\Sigma$-algebra $\mathbf{T}_{\Sigma}(A)$, specified below, by $\mathrm{Lang_{r}}(\mathbf{T}_{\Sigma}(A))$ the set of all $L\subseteq \mathrm{T}_{\Sigma}(A)$ such that $\Omega^{\mathbf{T}_{\Sigma}(A)}(L)$ is a congruence of finite index on $\mathbf{T}_{\Sigma}(A)$.

In this article, for a fixed set of sorts $S$ in $\boldsymbol{\mathcal{U}}$ and a fixed $S$-sorted signature $\Sigma$, we firstly consider the following types of many-sorted formations. ($\mathrm{I}$) Formations of $\Sigma$-algebras. That is, sets of $\Sigma$-algebras $\mathcal{F}$ closed under isomorphisms, homomorphic images, and finite subdirect products. And ($\mathrm{II}$) formations of congruences with respect to $\Sigma$. That is, choice functions $\mathfrak{F}$ for the family  $(\mathrm{Filt}(\mathbf{Cgr}(\mathbf{T}_{\Sigma}(A))))_{A\in \boldsymbol{\mathcal{U}}^{S}}$ such that, for every $S$-sorted sets $A$, $B$, every homomorphism $f$ from $\mathbf{T}_{\Sigma}(A)$ to $\mathbf{T}_{\Sigma}(B)$, and every $\Theta\in \mathfrak{F}(B)$, if $\mathrm{pr}^{\Theta}\circ f\colon \mathbf{T}_{\Sigma}(A)\mor \mathbf{T}_{\Sigma}(B)/\Theta$ is surjective, then $\mathrm{Ker}(\mathrm{pr}^{\Theta}\circ f)\in \mathfrak{F}(A)$. Let us point out that the notion of formation of congruences with respect to $\Sigma$ is a generalization to the many-sorted case of the definition presented in~\cite{bcrr15}, on $\mathrm{p}.\, 186$. And our first main result concerning the aforementioned formations is the proof that there exists an isomorphism between the algebraic lattice of all $\Sigma$-algebra formations and the algebraic lattice of all $\Sigma$-congruence formations.

Before proceeding any further we remark that with regard to the congruence approach for many-sorted algebras adopted by us in this article, it was explored for monoids in other papers (e.g., in~\cite{bcrr15} and~\cite{cll15}). Actually, one of the most significant efforts known to us in this last direction was made by Denis Th\'{e}rien in~\cite{th80}. There, Th\'{e}rien considers the problem of providing an algebraic classification of regular languages. Actually, one of his most interesting contributions is the proof that the $\ast$-varieties of congruences are in one-to-one correspondence with the varieties of regular languages and with the pseudovarieties of monoids---which is an extension of Eilenberg's variety theorem. For the case of monoids, the main difference between an $\ast$-variety of congruences and a formation of congruences is that Th\'{e}rien only considers finite index congruences. Moreover, he does nor require that the composition of the corresponding homomorphisms be surjective. The congruence approach is very helpful because it is fundamentally constructive and one can systematically generate $\ast$-varieties of congruences of increasing complexity.


After the above remark, we further note that in this article, for a fixed set of sorts $S$, a fixed $S$-sorted signature $\Sigma$, and under the hypothesis that, for every $S$-sorted set $A$ in $\boldsymbol{\mathcal{U}}^{S}$, the support of $\mathbf{T}_{\Sigma}(A)$ is finite, we also secondly and finally consider the following types of many-sorted formations.  ($\mathrm{III}$) Formations of finite index congruences with respect to $\Sigma$. That is, formations of congruences $\mathfrak{F}$ with respect to $\Sigma$ such that, for every $S$-sorted set $A$, $\mathfrak{F}(A)\subseteq \mathrm{Cgr}_{\mathrm{fi}}(\mathbf{T}_{\Sigma}(A))$, where $\mathrm{Cgr}_{\mathrm{fi}}(\mathbf{T}_{\Sigma}(A))$ is the filter of the algebraic lattice $\mathbf{Cgr}(\mathbf{T}_{\Sigma}(A))$ formed by those congruences that are of finite index. ($\mathrm{IV}$) Formations of finite $\Sigma$-algebras. And ($\mathrm{V}$)  formations of regular languages with respect to $\Sigma$. That is, choice functions $\mathcal{L}$ for $(\mathrm{Sub}(\mathrm{Lang_{r}}(\mathbf{T}_{\Sigma}(A))))_{A\in \boldsymbol{\mathcal{U}}^{S}}$, satisfying the following conditions: (1) for every $A\in \boldsymbol{\mathcal{U}}^{S}$, $\mathcal{L}(A)$ contains all languages of $\mathrm{T}_{\Sigma}(A)$ saturated by the greatest congruence on $\mathbf{T}_{\Sigma}(A)$, (2) for every $A\in \boldsymbol{\mathcal{U}}^{S}$ and every $L$, $L' \in \mathcal{L}(A)$, $\mathcal{L}(A)$ contains all languages of $\mathrm{T}_{\Sigma}(A)$ saturated by $\Omega^{\mathbf{T}_{\Sigma}(A)}(L)\cap \Omega^{\mathbf{T}_{\Sigma}(A)}(L')$, the intersection of the cogenerated congruences by $L$ and $L'$, respectively, and (3) for every $S$-sorted sets $A$, $B$, every $M\in \mathcal{L}(B)$, and every homomorphism $f$ from $\mathbf{T}_{\Sigma}(A)$ to $\mathbf{T}_{\Sigma}(B)$, if $\mathrm{pr}^{\Omega^{\mathbf{T}_{\Sigma}(B)}(M)}\circ f$ is an epimorphism, then $\mathcal{L}(A)$ contains all languages of $\mathrm{T}_{\Sigma}(A)$ saturated by  $\mathrm{Ker}(\mathrm{pr}^{\Omega^{\mathbf{T}_{\Sigma}(B)}(M)}\circ f)$. This last definition is, ultimately, based on that  presented, for the monoid case, in~\cite{bcrr15} on $\mathrm{p}.\, 187$. However, in this article, in contrast with~\cite{bcrr15}, no appeal to coalgebras is needed since all relevant notions can be stated by using saturations under congruences. And our second main result concerning the aforementioned formations is the proof that the algebraic lattices of all $\Sigma$-finite index congruence formations, of all $\Sigma$-finite algebra formations, and of all $\Sigma$-regular language formations are isomorphic.

Let us point out that the use of formations in the field of many-sorted algebra, as we do in this article, seems, to the best of our knowledge, to be new. The generality we have achieved in this work by using the many-sorted algebras encompasses not only the automata case and their generalizations, but also every type of action of an algebraic construct on another. Moreover, in the light of the results obtained, we think that the original Eilenberg's variety theorem can now be considered as a theorem of many-sorted universal algebra.

We next proceed to succinctly summarize the contents of the subsequent sections of this article. The reader will find a more detailed explanation at the beginning of the succeeding sections.

In Section 2, for the convenience of the reader, we recall, mostly without proofs, for a set of sorts $S$ and an $S$-sorted signature $\Sigma$, those notions and constructions of the theories of $S$-sorted sets and of $\Sigma$-algebras which are indispensable to define in the subsequent sections those others which will allow us to achieve the above  mentioned  main results, thus making, so we hope, our exposition self-contained.

In Section 3 we define, for an $S$-sorted signature $\Sigma$, the concepts of formation of algebras with respect to $\Sigma$, of formation of congruences with respect to $\Sigma$, and of Shemetkov\!\!$\And$\!\!Skiba-formation of algebras with respect to $\Sigma$, which is a generalization to the many-sorted case of that proposed in~\cite{shsk89}, and of which we prove that is equivalent to that of formation of algebras with respect to $\Sigma$. Besides, we investigate the properties of the aforementioned formations and prove an Eilenberg type theorem which states an isomorphism between the algebraic lattice of all $\Sigma$-algebra formations and the algebraic lattice of all $\Sigma$-congruence formations.

In Section 4 we define, for a $\Sigma$-algebra, the concepts of elementary translation and of translation with respect to it, and provide, by using the just mentioned notions, two characterizations of the congruences on a $\Sigma$-algebra. Moreover, we investigate the relationships between the translations and the homomorphisms between $\Sigma$-algebras.

In Section 5, for a $\Sigma$-algebra $\mathbf{A}$, we define, by making use of the translations, a mapping $\Omega^{\mathbf{A}}$ from $\mathrm{Sub}(A)$, the set of all subsets of the underlying $S$-sorted set $A$ of $\mathbf{A}$, to $\mathrm{Cgr}(\mathbf{A})$, the set of all $S$-sorted congruences on $\mathbf{A}$, which assigns to a subset $L$ of $A$ the, so-called, congruence cogenerated by $L$, and investigate its properties.

In Section 6 we define, for an $S$-sorted signature $\Sigma$ and under a suitable condition on the free $\Sigma$-algebras, the concepts of formation of finite index congruences with respect to $\Sigma$, of formation of finite $\Sigma$-algebras, of formation of regular languages with respect to $\Sigma$, and of Ballester\!$\And$\!Pin\!$\And$\!Soler-formation of regular languages  with respect to $\Sigma$, which is a generalization to the many-sorted case of that proposed in~\cite{bps14}, and of which we prove that is equivalent to that of formation of regular languages with respect to $\Sigma$. Moreover, we investigate the properties of the aforementioned formations and prove that the algebraic lattices of all $\Sigma$-finite index congruence formations, of all $\Sigma$-finite algebra formations, and of all $\Sigma$-regular language formations are isomorphic.

Our underlying set theory is $\mathbf{ZFSK}$, Zermelo-Fraenkel-Skolem set theory (also known as Zermelo-Fraenkel set theory), plus the existence of a Grothendieck(-Sonner-Tarski) universe $\ensuremath{\boldsymbol{\mathcal{U}}}$, fixed once and for all (see~\cite{sM98}, pp.~21--24). We recall that the elements of $\ensuremath{\boldsymbol{\mathcal{U}}}$ are called $\ensuremath{\boldsymbol{\mathcal{U}}}$-small sets and the subsets of $\ensuremath{\boldsymbol{\mathcal{U}}}$ are called $\ensuremath{\boldsymbol{\mathcal{U}}}$-large sets or classes. Moreover, from now on $\mathbf{Set}$ stands for the category of sets, i.e., the category whose object set, $\mathrm{Ob}(\mathbf{Set})$, is $\ensuremath{\boldsymbol{\mathcal{U}}}$, and whose morphism set, $\mathrm{Mor}(\mathbf{Set})$, is the set of all mappings between $\ensuremath{\boldsymbol{\mathcal{U}}}$-small sets (notice that $\mathrm{Mor}(\mathbf{Set})\subseteq \ensuremath{\boldsymbol{\mathcal{U}}}$ and that, for every $A$, $B\in \ensuremath{\boldsymbol{\mathcal{U}}}$, the hom-set $\mathrm{Hom}_{\mathbf{Set}}(A,B) = \mathrm{Hom}(A,B)\in \ensuremath{\boldsymbol{\mathcal{U}}}$).

In all that follows we use standard concepts and constructions from category theory, see e.g., \cite{sM98}; classical universal algebra, see e.g., \cite{gb15}, \cite{pC81}, \cite{gG08}, and \cite{mmt87}; many-sorted universal algebra, see e.g., \cite{jB68}, \cite{gm85}, \cite{h63}, and \cite{m76}; lattice theory, see e.g., \cite{bir79} and \cite{gG11}, and set theory, see e.g., \cite{nB70}, \cite{hE77}, and \cite{dM69}. Nevertheless, regarding set theory, we have adopted the following conventions. Between ordinals the binary relation ``$<$'' is identified with ``$\in$''; thus for a \emph{(von Neumann) ordinal} $\alpha$ we have that $\alpha = \{\,\beta\mid \beta\in \alpha\,\}$, and $\mathbb{N}$, the first transfinite ordinal, is the set of all \emph{natural numbers}. For a mapping $f\colon A\mor B$, a subset $X$ of $A$, and a subset $Y$ of $B$, we denote by $f^{-1}[Y]$ the \emph{inverse image of} $Y$ \emph{under} $f$, and by $f[X]$ the \emph{direct image of} $X$ \emph{under} $f$. For a set $B$, a set $I$, a family of sets $(A_{i})_{i\in I}$, and a family of mappings $(f_{i})_{i\in I}$ in $\prod_{i\in I}\mathrm{Hom}(B,A_{i})$, we denote by $\left<f_{i}\right>_{i\in I}$ the unique mapping from $B$ to $\prod_{i\in I}A_{i}$ such that, for every $i\in I$, $f_{i} = \mathrm{pr}_{i}\comp \left<f_{i}\right>_{i\in I}$, where, for every $i\in I$, $\mathrm{pr}_{i}$ is the canonical projection from $\prod_{i\in I}A_{i}$ to $A_{i}$.

More specific assumptions, conditions, and conventions will be included and explained in the successive sections.

\section{Preliminaries.}

In this section we introduce those basic notions and constructions which we shall need to define in the subsequent sections those others which will allow us to achieve the aforementioned main results of this article. Specifically, for a set (of sorts) $S$ in $\ensuremath{\boldsymbol{\mathcal{U}}}$, we begin by recalling the concept of free monoid on $S$, which will be fundamental for defining the concept of $S$-sorted signature. Following this we define the concepts of $S$-sorted set, $S$-sorted mapping from an $S$-sorted set to another, and the corresponding category. Moreover, we define the subset relation between $S$-sorted sets, the notion of finiteness as applied to $S$-sorted sets, some special objects of the category of $S$-sorted sets---in particular, the deltas of Kronecker---, the concept of support of an $S$-sorted set, and its properties, the notion of $S$-sorted equivalence on an $S$-sorted set, the saturation of an $S$-sorted set with respect to an $S$-sorted equivalence on an $S$-sorted set, and its properties, the quotient $S$-sorted set of an $S$-sorted set by an $S$-sorted equivalence on it, and the usual set-theoretic operations on the $S$-sorted sets.

Afterwards, for a set (of sorts) $S$ in $\ensuremath{\boldsymbol{\mathcal{U}}}$, we define the notion of $S$-sorted signature. Next, for an $S$-sorted signature $\Sigma$, we define the concepts of $\Sigma$-algebra, $\Sigma$-homomorphism (or, to abbreviate, homomorphism) from a $\Sigma$-algebra to another, and the corresponding category. Moreover, we define the notions of support of a $\Sigma$-algebra, of finite $\Sigma$-algebra, and of subalgebra of a $\Sigma$-algebra, the construction of the product of a family of $\Sigma$-algebras, the concepts of subfinal $\Sigma$-algebra and of congruence on a $\Sigma$-algebra, the constructions of the quotient $\Sigma$-algebra of a $\Sigma$-algebra by a congruence on it and of the free $\Sigma$-algebra on an $S$-sorted set, and the concept of subdirect product of a family of $\Sigma$-algebras.

From now on we make the following assumption: $S$ is a set of sorts in $\ensuremath{\boldsymbol{\mathcal{U}}}$, fixed once and for all.

\begin{definition}
Let $S$ be a set of sorts. The \emph{free monoid on} $S$, denoted by $\mathbf{S}^{\star}$, is $(S^{\star},\curlywedge,\lambda)$, where $S^{\star}$, the set of all \emph{words on} $S$, is $\bigcup_{n\in\mathbb{N}}\mathrm{Hom}(n,S)$, the set of all mappings $w\colon n\mor S$ from some $n\in \mathbb{N}$ to $S$, $\curlywedge$, the \emph{concatenation} of words on $S$, is the binary operation on $S^{\star}$ which sends a pair of words $(w,v)$ on $S$ to the mapping $w\curlywedge v$ from $\bb{w}+\bb{v}$ to $S$, where $\bb{w}$ and $\bb{v}$ are the lengths ($\equiv$ domains) of the mappings $w$ and $v$, respectively, defined as follows:
$$
w\bconcat v
\nfunction
{\bb{w}+\bb{v}}{S}
{i}{
\begin{cases}
w_{i}, & \text{if $0\leq i < \bb{w}$;}\\
v_{i-\bb{w}}, & \text{if $\bb{w}\leq i < \bb{w}+\bb{v}$,}
\end{cases}
}
$$
and $\lambda$, the \emph{empty word on} $S$, is the unique mapping from $0 = \varnothing$ to $S$.
\end{definition}

\begin{definition}
Let $S$ be a set of sorts. An $S$-\emph{sorted set} is a function $A = (A_{s})_{s\in S}$ from $S$ to $\ensuremath{\boldsymbol{\mathcal{U}}}$. If $A$ and $B$ are $S$-sorted sets, an $S$-\emph{sorted mapping from} $A$ \emph{to} $B$ is an $S$-indexed family $f = (f_{s})_{s\in S}$, where, for every $s$ in $S$, $f_{s}$ is a mapping from $A_{s}$ to  $B_{s}$. Thus, an $S$-sorted mapping from $A$ to $B$ is an element of $\prod_{s\in S}\mathrm{Hom}(A_{s}, B_{s})$, where, for every $s\in S$, $\mathrm{Hom}(A_{s}, B_{s})$ is the set of all mappings from $A_{s}$ to $B_{s}$. We denote by $\mathrm{Hom}(A,B)$ the set of all $S$-sorted mappings from $A$ to $B$. From now on, $\mathbf{Set}^{S}$ stands for the category of $S$-sorted sets and $S$-sorted mappings.
\end{definition}

\begin{definition}
Let $S$ be a set of sorts, $I$ a set in $\ensuremath{\boldsymbol{\mathcal{U}}}$, and $(A^{i})_{i\in I}$ an $I$-indexed family of $S$-sorted sets. Then the \emph{product} of $(A^{i})_{i\in I}$, denoted by $\prod_{i\in I}A^{i}$, is the $S$-sorted set defined, for every $s\in S$, as $\left(\prod\nolimits_{i\in I}A^{i}\right)_{s} = \prod\nolimits_{i\in I}A^{i}_{s}$. Moreover, for every $i\in I$, the \emph{ith canonical projection}, $\mathrm{pr}^{i} = (\mathrm{pr}^{i}_{s})_{s\in S}$, is the $S$-sorted mapping from  $\prod_{i\in I}A^{i}$ to $A^{i}$ defined, for every $s\in S$, as follows:
$$
\mathrm{pr}^{i}_{s}
\nfunction
{\prod_{i\in I}A^{i}_{s}}
{A^{i}_{s}}
{(a_{i})_{i\in I}}
{a_{i}}
$$
On the other hand, if $B$ is an $S$-sorted set, $I$ a set of indexes, and $(f^{i})_{i\in I}$ an $I$-indexed family of $S$-sorted mappings, where, for every $i\in I$, $f^{i}$ is an $S$-sorted mapping from $B$ to $A^{i}$, then we denote by $\left<f^{i}\right>$ the unique $S$-sorted mapping $f$ from $B$ to $\prod_{i\in I}A^{i}$ such that, for every $i\in I$, $\mathrm{pr}^{i}\circ f = f^{i}$.

The remaining set-theoretic operations on $S$-sorted sets: $\times$ (binary product), $\coprod$ (coproduct), $\amalg$ (binary coproduct), $\bigcup$ (union), $\cup$ (binary union), $\bigcap$ (intersection), $\cap$ (binary intersection), $\complement_{A}$ (complement of an $S$-set with respect to a given $S$-sorted $A$), and $-$ (difference), are defined in a similar way, i.e., componentwise.
\end{definition}

\begin{definition}
Let $S$ be a set of sorts. An $S$-sorted set $A$ is \emph{subfinal} if, for every $s\in S$, $\mathrm{card}(A_{s})\leq 1$. We denote by $1^{S}$, or, to abbreviate, by $1$, the (standard) final $S$-sorted set of $\mathbf{Set}^{S}$, which is $1^{S} = (1)_{s\in S}$, and by $\varnothing^{S}$ the initial $S$-sorted set, which is $\varnothing^{S} = (\varnothing)_{s\in S}$.
\end{definition}

\begin{definition}
Let $S$ be a set of sorts. If $A$ and $B$ are $S$-sorted sets, then we will say that $A$ is a \emph{subset} of $B$, denoted by $A\subseteq B$, if, for every $s\in S$, $A_{s}\subseteq B_{s}$. We denote by $\mathrm{Sub}(A)$ the set of all $S$-sorted sets $X$ such that $X\subseteq A$.
\end{definition}

\begin{definition}
Given a sort $t\in S$ we call \emph{delta of Kronecker in} $t$, the $S$-sorted set $\delta^{t} = (\delta^{t}_{s})_{s\in S}$ defined, for every $s\in S$, as follows:
\begin{equation*}
\delta^{t}_{s} =
\begin{cases}
1, &\text{if $s = t$;}\\
\vacio, & \text{otherwise.}
\end{cases}
\end{equation*}
Let $t$ be a sort in $S$ and $X$ a set, then we denote by $\delta^{t,X}$ the $S$-sorted set defined, for every $s\in S$, as follows:
\begin{equation*}
\delta^{t,X}_{s} =
\begin{cases}
X, &\text{if $s = t$;}\\
\vacio, & \text{otherwise.}
\end{cases}
\end{equation*}
Let us notice that $\delta^{t}$ is $\delta^{t,1}$, i.e., the deltas of Kronecker, are particular cases of the $S$-sorted sets $\delta^{t,X}$ (however, see the remark immediately below). Therefore we will use $\delta^{t}$ or $\delta^{t,1}$.
\end{definition}

\begin{remark}
For a sort $t\in S$ and a set $X$, the $S$-sorted set $\delta^{t,X}$ is isomorphic to the $S$-sorted set $\coprod_{x\in X}\delta^{t}$, i.e., to the coproduct of the family $(\delta^{t})_{x\in X}$.

For every sort $t\in S$ we have a functor $\delta^{t,\cdot}$ from $\mathbf{Set}$ to $\mathbf{Set}^{S}$. In fact, for every set $X$, $\delta^{t,\cdot}(X) = \delta^{t,X}$, and, for every mapping $f\colon X\mor Y$, $\delta^{t,\cdot}(f) = \delta^{t,f}$, where, for $s\in S$, $\delta^{t,f}_{s} = \mathrm{id}_{\varnothing}$, if $s\neq t$, and $\delta^{t,f}_{t} = f$. Moreover, for every $t\in S$, the object mapping of the functor $\delta^{t,\cdot}$ is injective and $\delta^{t,\cdot}$ is full and faithful. Hence, for every $s\in S$, $\delta^{t,\cdot}$ is a full embedding from $\mathbf{Set}$ to $\mathbf{Set}^{S}$.

The final object $1^{S}$ does not generates ($\equiv$ separates)
$\mathbf{Set}^{S}$, but the set $\{\,\delta^{s}\mid s\in S\,\}$, of the deltas of Kronecker, is a generating ($\equiv$ separating) set for the category $\mathbf{Set}^{S}$.
Therefore, every $S$-sorted set $A$ can be represented as a coproduct of copowers of deltas of Kronecker, i.e., $A$ is
naturally isomorphic to $\coprod_{s\in S}\mathrm{card}(A_{s})\cdot\delta^{s}$, where, for every $s\in S$,
$\mathrm{card}(A_{s})\cdot\delta^{s}$ is the copower of the family $(\delta^{s})_{\alpha\in \mathrm{card}(A_{s})}$, i.e., the coproduct of $(\delta^{s})_{\alpha\in \mathrm{card}(A_{s})}$.

To this we add the following. (1) That $\{\,\delta^{s}\mid s\in S\,\}$ is the set of atoms of the Boolean algebra $\mathbf{Sub}(1^{S})$, of subobjects of $1^{S}$. (2) That the Boolean algebras $\mathbf{Sub}(1^{S})$ and $\mathbf{Sub}(S)$ are isomorphic. (3) That, for every $s\in S$, $\delta^{s}$ is a projective object. And (4) that, for every $s\in S$, every $S$-sorted mapping from $\delta^{s}$ to another $S$-sorted set is a monomorphism.

In view of the above, it must  be concluded that the deltas of Kronecker are of crucial importance for many-sorted sets and associated fields.
\end{remark}

Before proceeding any further, let us point out that it is no longer unusual to find in the literature devoted to many-sorted algebra the following. (1) That an $S$-sorted set $A$ is defined in such a way that $\mathrm{Hom}(1^{S},A)\neq \varnothing$, or, what is equivalent, requiring that, for every $s\in S$, $A_{s}\neq \varnothing$. This has as an immediate consequence that the corresponding category is not even finite cocomplete. Since cocompleteness (and completeness) are desirable properties for a category, we exclude such a convention in our work (the admission of $\varnothing^{S}$ is crucial in many applications). And (2) that an $S$-sorted set $A$ must be such that, for every $s$, $t\in S$, if $s\neq t$, then $A_{s}\cap A_{t} = \varnothing$. We also exclude such a requirement (the possibility of a common underlying set for the different sorts is very important in many applications). The above conventions are possibly based on the untrue widespread belief that many-sorted equational logic is an inessential variation of single-sorted equational logic (one can find definitive refutations to the just mentioned belief, e.g., in~\cite{mb72}, \cite{cs05}, \cite{fs90}, \cite{gm85}, \cite{jh85}, and \cite{m76}).

We next define for an $S$-sorted mapping the associated mappings of direct and inverse image formation.

\begin{definition}
Let $f\colon A\mor B$ be an $S$-sorted mapping. Then
\begin{enumerate}
\item The mapping $f[\farg]$ of $f$-\emph{direct image formation} is the mapping defined as follows:
$$
f[\farg] \nfunction
{\mathrm{Sub}(A)}
{\mathrm{Sub}(B)}
{X}
{ f[X] = (f_{s}[X_{s}])_{s\in S}  }
$$

\item The mapping $f^{-1}[\farg]$ of $f$-\emph{inverse image formation} is the mapping defined as follows:
$$
f^{-1}[\farg] \nfunction
{\mathrm{Sub}(B)}
{\mathrm{Sub}(A)}
{Y}
{ f^{-1}[Y] = (f_{s}^{-1}[Y_{s}] )_{s\in S} }
$$
\end{enumerate}
\end{definition}

\begin{definition}
Let $S$ be a set of sorts. An $S$-sorted set $A$ is \emph{finite} if $\coprod A = \bigcup_{s\in S}(A_{s}\times \{s\})$ is finite. We say that $A$ is a \emph{finite} subset of $B$ if $A$ is finite and $A\subseteq B$.
\end{definition}

\begin{remark}
An $S$-sorted set $A$ is finite if, and only if, the covariant hom-functor $\mathrm{H}(A,\cdot)\colon \mathbf{Set}^{S}\mor \mathbf{Set}$ is finitary, i.e., if, and only if, for every $\ensuremath{\boldsymbol{\mathcal{U}}}$-small upward-directed preordered set $\mathbf{I}$ and every functor $D$ from $\mathbf{\mathbf{I}}$, the category canonically associated to $\mathbf{I}$, to $\mathbf{Set}^{S}$, if $((f^{i})_{i\in \mathrm{Ob}(\mathbf{\mathbf{I}})},L)$ is an inductive limit of $D$, then $((F(f^{i}))_{i\in \mathrm{Ob}(\mathbf{I})},F(L))$ is an epi-sink. Moreover, the $S$-sorted set $A$ is finite if, and only if, the functor $\mathrm{H}(A,\cdot)$ from $\mathbf{Set}^{S}$ to $\mathbf{Set}$ is strongly finitary, i.e., if, and only if, for every $\ensuremath{\boldsymbol{\mathcal{U}}}$-small upward-directed preordered set $\mathbf{I}$ and every functor $D$ from $\mathbf{I}$ to $\mathbf{Set}^{S}$, if $((f^{i})_{i\in \mathrm{Ob}(\mathbf{I})},L)$ is an inductive limit of $D$, then $((F(f^{i}))_{i\in \mathrm{Ob}(\mathbf{I})},F(L))$ is an inductive limit.
\end{remark}

\begin{definition}
Let $S$ be a set of sorts. Then the \emph{support of} $A$, denoted by $\mathrm{supp}_{S}(A)$, is the set $\{\,s\in S\mid A_{s}\neq \varnothing\,\}$.
\end{definition}

\begin{remark}
An $S$-sorted set $A$ is finite if, and only if, $\mathrm{supp}_{S}(A)$ is finite and, for every $s\in \mathrm{supp}_{S}(A)$, $A_{s}$ is finite.
\end{remark}

In the following proposition, for a set of sorts $S$, we gather together the most interesting properties of the mapping   $\mathrm{supp}_{S}\colon\U^{S}\mor \Sub(S)$, the support mapping for $S$, which sends an $S$-sorted set $A$ to $\mathrm{supp}_{S}(S)$.

\begin{proposition}\label{propssupport}
Let $S$ be a set of sorts, $A$, $B$ two $S$-sorted sets, $I$ a set in $\ensuremath{\boldsymbol{\mathcal{U}}}$, and $(A^{i})_{i\in I}$ an $I$-indexed family of $S$-sorted sets. Then the following properties hold:
\begin{enumerate}
\item $\mathrm{Hom}(A,B)\neq \vacio$ if, and only if,
      $\mathrm{supp}_{S}(A)\subseteq\mathrm{supp}_{S}(B)$. Therefore, if $A\subseteq B$, then $\mathrm{supp}_{S}(A)\subseteq\mathrm{supp}_{S}(B)$. Moreover, if $f$ is an $S$-sorted mapping from $A$ to $B$ and $X\subseteq A$, then $\mathrm{supp}_{S}(X) = \mathrm{supp}_{S}(f[X])$.

\item If from $A$ to $B$ there exists a surjective $S$-sorted mapping $f$, then
      $\mathrm{supp}_{S}(A) = \mathrm{supp}_{S}(B)$. Moreover, if $Y\subseteq B$, then $\mathrm{supp}_{S}(Y) = \mathrm{supp}_{S}(f^{-1}[Y])$.

\item $\mathrm{supp}_{S}(\varnothing^{S}) = \varnothing$; $\mathrm{supp}_{S}(1) = S$; $\mathrm{supp}_{S}(\bigcup_{i\in I} A^{i}) = \bigcup_{i\in I}\mathrm{supp}_{S}(A^{i})$; $\mathrm{supp}_{S}(\coprod_{i\in I} A^{i}) = \bigcup_{i\in I}\mathrm{supp}_{S}(A^{i})$; if $I\neq \varnothing$, $\mathrm{supp}_{S}(\bigcap_{i\in I} A^{i}) = \bigcap_{i\in I}\mathrm{supp}_{S}(A^{i})$; $\mathrm{supp}_{S}(\prod_{i\in I}A^{i}) = \bigcap\nolimits_{i\in I}\mathrm{supp}_{S}(A_{i})$; and $\mathrm{supp}_{S}(A)-\mathrm{supp}_{S}(B)\subseteq\mathrm{supp}_{S}(A-B)$.
\end{enumerate}
\end{proposition}

\begin{definition}
Let $S$ be a set of sorts. An $S$-\emph{sorted equivalence relation on} (or, to abbreviate, an $S$-\emph{sorted equivalence on}) an $S$-sorted set $A$ is an $S$-sorted relation $\Phi$ on $A$, i.e., a subset $\Phi = (\Phi_{s})_{s\in S}$ of the cartesian product $A\times A = (A_{s}\times A_{s})_{s\in S}$ such that, for every $s\in S$, $\Phi_{s}$ is an equivalence relation on $A_{s}$. We denote by $\mathrm{Eqv}(A)$ the set of all $S$-sorted equivalences on $A$ (which is an algebraic closure system on $A\times A$), by $\mathbf{Eqv}(A)$ the algebraic lattice  $(\mathrm{Eqv}(A),\subseteq)$, by $\nabla^{A}$ the greatest element of $\mathbf{Eqv}(A)$, and by $\Delta^{A}$ the least element of $\mathbf{Eqv}(A)$.

For an $S$-sorted equivalence relation $\Phi$ on $A$, $A/\Phi$, the $S$-\emph{sorted quotient set of} $A$ \emph{by} $\Phi$, is $(A_{s}/\Phi_{s})_{s\in S}$, and $\mathrm{pr}^{\Phi}\colon A\mor A/\Phi$, the \emph{canonical projection from} $A$ \emph{to} $A/\Phi$, is the $S$-sorted mapping $(\mathrm{pr}^{\Phi_{s}})_{s\in S}$, where, for every $s\in S$, $\mathrm{pr}^{\Phi_{s}}$ is the canonical projection from $A_{s}$ to $A_{s}/\Phi_{s}$ (which sends $x$ in $A_{s}$ to $\mathrm{pr}^{\Phi_{s}}(x) = [x]_{\Phi_{s}}$, the $\Phi_{s}$-equivalence class of $x$, in $A_{s}/\Phi_{s}$).

Let $\Phi$ and $\Psi$ be $S$-sorted equivalence on $A$ such that $\Phi\subseteq\Psi$. Then the \emph{quotient of} $\Psi$ \emph{by} $\Phi$, denoted by $\Psi/\Phi$, is the $S$-sorted equivalence $(\Psi_{s}/\Phi_{s})_{s\in S}$ on $A/\Phi$ defined, for every $s\in S$, as follows:
$$
\Psi_{s}/\Phi_{s}=\{([a]_{\Phi_{s}},
[b]_{\Phi_{s}})\in (A_{s}/\Phi_{s})\mid (a,b)\in\Psi_{s}\}.
$$

Let $X$ be a subset of $A$ and $\Phi\in\mathrm{Eqv}(A)$. Then the $\Phi$-\emph{saturation of} $X$ (or, the \emph{saturation of} $X$ \emph{with respect to} $\Phi$), denoted by $[X]^{\Phi}$, is the $S$-sorted set defined, for every $s\in S$, as follows:
$$
\textstyle[X]^{\Phi}_{s} = \{a\in A_{s}\mid
X_{s}\cap [a]_{\Phi_{s}}\neq\vacio\}= \bigcup_{x\in X_{s}}[x]_{\Phi_{s}} = [X_{s}]^{\Phi_{s}}.
$$
Let $X$ be a subset of $A$ and $\Phi\in\mathrm{Eqv}(A)$. Then we say that $X$ is $\Phi$-\emph{saturated} if, and only if, $X = [X]^{\Phi}$. We will denote by $\Phi\!-\!\mathrm{Sat}(A)$ the subset of $\mathrm{Sub}(A)$ defined as $\Phi\!-\!\mathrm{Sat}(A) = \{X\in \mathrm{Sub}(A)\mid X = [X]^{\Phi}\}$.
\end{definition}

\begin{remark}
Let $S$ be a set of sorts, $A$ an $S$-sorted set, and $\Phi$ an $S$-sorted equivalence on $A$. Then, by Proposition~\ref{propssupport}, $\mathrm{supp}_{S}(A) = \mathrm{supp}_{S}(A/\Phi)$.
\end{remark}

\begin{remark}
Let $S$ be a set of sorts, $A$ an $S$-sorted set, and $\Phi\in\mathrm{Eqv}(A)$. Then, for an $S$-sorted subset $X$ of
$A$, we have that the $\Phi$-saturation of $X$ is $(\mathrm{pr}^{\Phi})^{-1}[\mathrm{pr}^{\Phi}[X]]$. Therefore, $X$ is $\Phi$-saturated if, and only if, $X \supseteq [X]^{\Phi}$. Besides, $X$ is $\Phi$-saturated if, and only if, there exists a $\mathcal{Y}\subseteq A/\Phi$ such that $X = (\mathrm{pr}^{\Phi})^{-1}[\mathcal{Y}]$.
\end{remark}

\begin{proposition}\label{PropIncSat}
Let $A$ be an $S$-sorted set and $\Phi$, $\Psi\in\mathrm{Eqv}(A)$, then
$$
\Phi\subseteq \Psi\, \text{if, and only if, }\, \forall X\subseteq A\;([[X]^{\Psi}]^{\Phi}=[X]^{\Psi}).
$$
\end{proposition}

\begin{proof}
Let us assume that $\Phi\subseteq \Psi$ and let $X$ be a subset of $A$. In order to prove that $[[X]^{\Psi}]^{\Phi}=[X]^{\Psi}$ it suffices to verify that $[[X]^{\Psi}]^{\Phi}\subseteq [X]^{\Psi}$. Let $s$ be an element of $S$. Then, by definition, $a\in [[X]^{\Psi}]^{\Phi}_{s}$ if, and only if, there exists some $b\in [X]^{\Psi}_{s}$ such that $a\in [b]_{\Phi_{s}}$. Since $\Phi\subseteq \Psi$, we have that $a\in [b]_{\Psi_{s}}$, therefore $a\in [X]^{\Psi}_{s}$.

To prove the converse, let us assume that $\Phi\not\subseteq \Psi$. Then there exists some sort $s\in S$ and elements $a$, $b$ in $A_{s}$ such that $(a,b)\in\Phi_{s}$ and $(a,b)\not\in \Psi_{s}$. Hence $b$ does not belong to $[\delta^{s,[a]_{\Psi_{s}}}]^{\Psi}_s$, whereas it does belong to $[[\delta^{s,[a]_{\Psi_{s}}}]^{\Psi}]^{\Phi}_{s}$. It follows that $[\delta^{s,[a]_{\Psi_{s}}}]^{\Psi}\neq [[\delta^{s,[a]_{\Psi_{s}}}]^{\Psi}]^{\Phi}$.
\end{proof}

\begin{corollary}\label{IncSat}
Let $A$ be an $S$-sorted set and $\Phi$, $\Psi\in\mathrm{Eqv}(A)$. If $\Phi\subseteq \Psi$, then  $\Psi\!-\!\mathrm{Sat}(A)\subseteq\Phi\!-\!\mathrm{Sat}(A)$.
\end{corollary}

\begin{remark}
If, for an $S$-sorted set $A$, we denote by $(\cdot)\!-\!\mathrm{Sat}(A)$ the mapping from $\mathrm{Eqv}(A)$ to $\mathrm{Sub}(\mathrm{Sub}(A))$ which send $\Phi$ to $\Phi\!-\!\mathrm{Sat}(A)$, then the above corollary means that $(\cdot)\!-\!\mathrm{Sat}(A)$ is an antitone ($\equiv$ order-reversing) mapping from the ordered set $(\mathrm{Eqv}(A),\subseteq)$ to the ordered set $(\mathrm{Sub}(\mathrm{Sub}(A)),\subseteq)$.
\end{remark}

\begin{proposition}\label{NablaSat}
Let $A$ be an $S$-sorted set and $X\subseteq A$. Then $X\in \nabla^{A}\!-\!\mathrm{Sat}(A)$ if, and only if, for every $s\in S$, if $s\in \mathrm{supp}_{S}(X)$, then $X_{s} = A_{s}$.
\end{proposition}

\begin{proof}
Let us suppose that there exists a $t\in S$ such that $X_{t}\neq \varnothing$ and $X_{t}\neq A_{t}$. Then, since $[X]_{t}^{\nabla^{A}} = \bigcup_{x\in X_{t}}[x]_{\nabla_{A_{t}}}$ and $X_{t}\neq \varnothing$, we have that, for some $y\in X_{t}$, $[y]_{\nabla_{A_{t}}} = A_{t}$. But $X_{t}\subset A_{t}$. Hence $[X]_{t}^{\nabla^{A}}\neq X_{t}$. Therefore $X\nin \nabla^{A}\!-\!\mathrm{Sat}(A)$.

The converse implication is straightforward.
\end{proof}

\begin{remark}
Let $A$ be an $S$-sorted set and $X\subseteq A$. Then, from the above proposition, it follows that $\varnothing^{S}$, $A\in \nabla^{A}\!-\!\mathrm{Sat}(A)$. Moreover, for every $T\subseteq S$, $\bigcup_{t\in T}\delta^{t,A_{t}}\in \nabla^{A}\!-\!\mathrm{Sat}(A)$.
\end{remark}

\begin{proposition}
Let $A$ be an $S$-sorted set, $X\subseteq A$, and $\Phi$, $\Psi\in\mathrm{Eqv}(A)$, then
$[X]^{\Phi\cap\Psi}\subseteq[X]^{\Phi}\cap[X]^{\Psi}$.
\end{proposition}

\begin{proof}
Let $s$ be a sort in $S$ and $b\in[X]^{\Phi\cap \Psi}_{s}$. Then, by definition, there exists an $a\in X_{s}$ such that $(a,b)\in (\Phi\cap\Psi)_{s} = \Phi_{s}\cap\Psi_{s}$. Hence, $(a,b)\in \Phi_{s}$ and $(a,b)\in \Psi_{s}$. Therefore $b\in [X]^{\Phi}_{s}$ and $b\in[X]^{\Psi}_{s}$. Consequently, $b\in ([X]^{\Phi}\cap [X]^{\Psi})_{s}$. Thus $[X]^{\Phi\cap\Psi}\subseteq[X]^{\Phi}\cap[X]^{\Psi}$.
\end{proof}

\begin{corollary}
Let $A$ be an $S$-sorted set and $\Phi$, $\Psi\in\mathrm{Eqv}(A)$. Then $\Phi\!-\!\mathrm{Sat}(A)\cap\Psi\!-\!\mathrm{Sat}(A)\subseteq(\Phi\cap\Psi)\!-\!\mathrm{Sat}(A)$.
\end{corollary}

We next state that, for a set of sorts $S$, an $S$-sorted set $A$, and an $S$\nobreakdash-sorted equivalence $\Phi$ on $A$, the set $\Phi\!-\!\mathrm{Sat}(A)$ is the set of all fixed points of a suitable operator on $A$, i.e., of a mapping of $\mathrm{Sub}(A)$ into itself.

\begin{proposition}\label{SatOperator}
Let $S$ be a set of sorts, $A$ an $S$-sorted set, and $\Phi\in\mathrm{Eqv}(A)$. Then the mapping $[\cdot]^{\Phi}$ from $\mathrm{Sub}(A)$ to $\mathrm{Sub}(A)$ defined as follows:
$$
[\cdot]^{\Phi} \nfunction
{\mathrm{Sub}(A)}
{\mathrm{Sub}(A)}
{X}
{ [\cdot]^{\Phi}(X) = [X]^{\Phi}  }
$$
is a completely additive closure operator on $A$. Moreover, $[\cdot]^{\Phi}$ is completely  multiplicative, i.e., for every nonempty set $I$ in $\ensuremath{\boldsymbol{\mathcal{U}}}$ and every $I$-indexed family $(X^{i})_{i\in I}$ in $\mathrm{Sub}(A)$, $[\bigcap_{i\in I}X^{i}]^{\Phi} = \bigcap_{i\in I}[X^{i}]^{\Phi}$ (and, obviously, $[A]^{\Phi} = A$), and, for every $X\subseteq A$, if $X = [X]^{\Phi}$, then  $\complement_{A}X = [\complement_{A}X]^{\Phi}$. Besides, $[\cdot]^{\Phi}$ is uniform, i.e., is such that, for every $X$, $Y\subseteq A$, if  $\mathrm{supp}_{S}(X) = \mathrm{supp}_{S}(Y)$, then $\mathrm{supp}_{S}([X]^{\Phi}) = \mathrm{supp}_{S}([Y]^{\Phi})$---hence, in particular, $[\cdot]^{\Phi}$ is a uniform algebraic closure operator on $A$. And $\Phi\!-\!\mathrm{Sat}(A) = \mathrm{Fix}([\cdot]^{\Phi})$, where $\mathrm{Fix}([\cdot]^{\Phi})$ is the set of all fixed point of the operator $[\cdot]^{\Phi}$.
\end{proposition}


\begin{proposition}\label{CABA Saturades} Let $A$ be an $S$-sorted set and $\Phi\in\mathrm{Eqv}(A)$. Then the ordered pair
$\Phi\!-\!\mathbf{Sat}(A) = (\Phi\!-\!\mathrm{Sat}(A),\subseteq)$ is a complete atomic Boolean algebra (for brevity a CABA).
\end{proposition}

\begin{proof}
The proof is straightforward and we leave it to the reader. We only point out that the atoms of $\Phi\!-\!\mathbf{Sat}(A)$ are precisely the deltas of Kronecker $\delta^{t, [x]_{\Phi_t}}$, for some $t\in S$ and some $x\in A_{t}$, and that, obviously, every $\Phi$-saturated subset $X$ of $A$ is the join ($\equiv$ union) of all atoms smaller than $X$.
\end{proof}

We next recall the concept of kernel of an $S$-sorted mapping and the universal property of the $S$-sorted quotient set of an $S$-sorted set by an $S$-sorted equivalence on it (which are at the basis of those of kernel of a $\Sigma$-homomorphism and of the universal property of the quotient $\Sigma$-algebra of a $\Sigma$-algebra by a congruence on it, respectively).

\begin{definition}
Let $f\colon A\mor B$ be an $S$-sorted mapping. Then the \emph{kernel} of $f$, denoted by  $\mathrm{Ker}(f)$, is the $S$-sorted relation defined, for every $s\in S$, as $\mathrm{Ker}(f)_{s} = \mathrm{Ker}(f_{s})$ (i.e., as the kernel pair of $f_{s}$).
\end{definition}

\begin{proposition}
Let $f$ be an $S$-sorted mapping from $A$ to $B$. Then $\mathrm{Ker}(f)$ is an $S$-sorted equivalence on $A$.
\end{proposition}

\begin{proposition}
Let $A$ be an $S$-sorted set and $\Phi\in \mathrm{Eqv}(A)$. Then the pair $(\mathrm{pr}^{\Phi},A/\Phi)$ is such that:
\begin{enumerate}
\item $\mathrm{Ker}(\mathrm{pr}^{\Phi}) = \Phi$.
\item \emph{(Universal property)} For every $S$-sorted mapping $f\colon A\mor B$, if $\Phi\subseteq \mathrm{Ker}(f)$, then there exists a unique $S$-sorted mapping $\mathrm{p}^{\Phi,\mathrm{Ker}(f)}$ from $A/\Phi$ to $B$ such that $f = \mathrm{p}^{\Phi,\mathrm{Ker}(f)}\circ \mathrm{pr}^{\Phi}$.
\end{enumerate}
\end{proposition}

Following this we define, for a set of sorts $S$, the category of $S$-sorted signatures.

\begin{definition}\label{$S$-sorted signature}
Let $S$ be a set of sorts. Then an $S$-\emph{sorted signature} is a function $\Sigma$ from $\fmon{S}\bprod S$ to $\ensuremath{\boldsymbol{\mathcal{U}}}$ which sends a pair $(w,s)\in \fmon{S}\bprod S$ to the set $\Sigma_{w,s}$ of the
\emph{formal operations} of \emph{arity} $w$, \emph{sort} (or \emph{coarity}) $s$, and \emph{rank} (or \emph{biarity}) $(w,s)$.  Sometimes we will write $\sigma\colon w\mor s$ to indicate that the formal operation $\sigma$ belongs to
$\Sigma_{w,s}$.
\end{definition}

From now on we make the following assumption: $\Sigma$ stands for an $S$-sorted signature, fixed once and for all.

We next define, for an $S$-sorted signature $\Sigma$ the category of $\Sigma$-algebras.

\begin{definition}
Let $\Sigma$ be an $S$-sorted signature. The $S^{\star}\times S$-sorted set of the \emph{finitary operations on} an $S$-sorted set $A$ is $(\mathrm{Hom}(A_{w},A_{s}))_{(w,s)\in S^{\star}\times S}$, where, for every $w\in S^{\star}$, $A_{w} = \prod_{i\in \lvert w\rvert}A_{w_{i}}$, with $\lvert w\rvert$ denoting the length of the word $w$.
A \emph{structure of} $\Sigma$-\emph{algebra on} an $S$-sorted set  $A$ is a family $(F_{w,s})_{(w,s)\in S^{\star}\times S}$, denoted by $F$, where, for $(w,s)\in S^{\star}\times S$, $F_{w,s}$ is a mapping from $\Sigma_{w,s}$ to $\mathrm{Hom}(A_{w},A_{s})$.  For a pair $(w,s)\in S^{\star}\times S$ and a formal operation $\sigma\in \Sigma_{w,s}$, in order to simplify the notation, the operation from $A_{w}$ to $A_{s}$ corresponding to $\sigma$ under $F_{w,s}$ will be written as $F_{\sigma}$ instead of $F_{w,s}(\sigma)$. A $\Sigma$-\emph{algebra} is a pair $(A,F)$, abbreviated to $\mathbf{A}$, where $A$ is an $S$-sorted set and $F$ a structure of $\Sigma$-algebra on $A$. A $\Sigma$-\emph{homomorphism} from $\mathbf{A}$ to $\mathbf{B}$, where $\mathbf{B} = (B,G)$, is a triple $(\mathbf{A},f,\mathbf{B})$, abbreviated to $f\colon \mathbf{A}\mor \mathbf{B}$, where $f$ is an $S$-sorted mapping from $A$ to $B$ such that, for every $(w,s)\in S^{\star}\times S$, $\sigma\in \Sigma_{w,s}$, and $(a_{i})_{i\in \lvert w\rvert}\in A_{w}$ we have that
$$
f_{s}(F_{\sigma}((a_{i})_{i\in \lvert w\rvert})) = G_{\sigma}(f_{w}((a_{i})_{i\in \lvert w\rvert})),
$$
where $f_{w}$ is the mapping $\prod_{i\in \lvert w\rvert}f_{w_{i}}$ from $A_{w}$ to $B_{w}$ which sends $(a_{i})_{i\in \lvert w\rvert}$ in $A_{w}$ to $(f_{w_{i}}(a_{i}))_{i\in \lvert w\rvert}$ in $B_{w}$. We denote by $\mathbf{Alg}(\Sigma)$ the category of $\Sigma$-algebras and $\Sigma$-homomorphisms (or, to abbreviate, homomorphisms) and by $\mathrm{Alg}(\Sigma)$ the set of objects of $\mathbf{Alg}(\Sigma)$.
\end{definition}

\begin{remark}
With regard to the category $\mathbf{Alg}(\Sigma)$ let us point out the following. (1) That $\mathrm{Alg}(\Sigma)\subseteq\ensuremath{\boldsymbol{\mathcal{U}}}$. And (2) that, for every $\mathbf{A}$, $\mathbf{B}\in \mathrm{Alg}(\Sigma)$, $\mathrm{Hom}_{\mathbf{Alg}(\Sigma)}(\mathbf{A},\mathbf{B})\in\ensuremath{\boldsymbol{\mathcal{U}}}$. Thus $\mathbf{Alg}(\Sigma)$ is a $\ensuremath{\boldsymbol{\mathcal{U}}}$-category.
\end{remark}

\begin{definition}
Let $\mathbf{A}$ be a $\Sigma$-algebra. Then the \emph{support of} $\mathbf{A}$, denoted by $\mathrm{supp}_{S}(\mathbf{A})$, is $\mathrm{supp}_{S}(A)$, i.e., the support of the underlying $S$-sorted set $A$ of $\mathbf{A}$.
\end{definition}

\begin{remark}
The set $\{\mathrm{supp}_{S}(\mathbf{A})\mid \mathbf{A}\in \mathrm{Alg}(\Sigma)\}$ is a closure system on $S$.
\end{remark}

\begin{definition}
Let $\mathbf{A}$ be a $\Sigma$-algebra. We say that $\mathbf{A}$ is \emph{finite} if $A$, the underlying $S$-sorted set of $\mathbf{A}$, is finite.
\end{definition}

We next define when a subset $X$ of the underlying $S$-sorted set $A$ of a $\Sigma$-algebra $\mathbf{A} = (A,F)$ is closed under an operation $F_{\sigma}$ of $\mathbf{A}$, as well as when $X$ is a subalgebra of $\mathbf{A}$.

\begin{definition}\label{Subalg}
Let $\mathbf{A}$ be a $\Sigma$-algebra and $X\subseteq A$. Let $\sigma$ be such that $\sigma\colon w\mor s$, i.e., a formal operation in $\Sigma_{w,s}$. We say that $X$ is \emph{closed under the operation} $F_{\sigma}\colon A_{w}\mor A_{s}$ if, for every $a\in X_{w}$, $F_{\sigma}(a)\in X_{s}$. We say that $X$ is a \emph{subalgebra} of $\mathbf{A}$ if $X$ is closed under the operations of $\mathbf{A}$. We denote by $\mathrm{Sub}(\mathbf{A})$ the set of all subalgebras of $\mathbf{A}$ (which is an algebraic closure system on $A$) and by  $\mathbf{Sub}(\mathbf{A})$ the algebraic lattice $(\mathrm{Sub}(\mathbf{A}),\subseteq)$. We also say, equivalently, that a $\Sigma$-algebra $\mathbf{B}$ is a \emph{subalgebra} of $\mathbf{A}$ if $B\subseteq A$ and the canonical embedding of $B$ into $A$ determines an embedding of $\mathbf{B}$ into $\mathbf{A}$.
\end{definition}

\begin{definition}
Let $\mathbf{A}$ be a $\Sigma$-algebra. Then we denote by $\mathrm{Sg}_{\mathbf{A}}$ the algebraic closure operator canonically
associated to the algebraic closure system $\mathrm{Sub}(\mathbf{A})$ on $A$ and we call it the \emph{subalgebra  generating operator} for $\mathbf{A}$. Moreover, if $X\subseteq A$, then we call $\mathrm{Sg}_{\mathbf{A}}(X)$ the \emph{subalgebra of} $\mathbf{A}$ \emph{generated by} $X$, and if $X$ is such that $\Sg_{\mathbf{A}}(X) = A$, then we say that $X$ is a \emph{generating} subset of $\mathbf{A}$.
\end{definition}

\begin{remark}
Let $\mathbf{A}$ be a $\Sigma$-algebra. Then the algebraic closure operator $\mathrm{Sg}_{\mathbf{A}}$ is uniform, i.e., for every  $X$, $Y\subseteq A$, if $\mathrm{supp}_{S}(X) = \mathrm{supp}_{S}(Y)$, then we have that $\mathrm{supp}_{S}(\mathrm{Sg}_{\mathbf{A}}(X)) = \mathrm{supp}_{S}(\mathrm{Sg}_{\mathbf{A}}(Y))$. To appreciate the significance  of the just mentioned property of $\mathrm{Sg}_{\mathbf{A}}$, see~\cite{cs04}.
\end{remark}

We now recall the concept of product of a family of $\Sigma$-algebras.

\begin{definition}
Let $I$ be a set in $\ensuremath{\boldsymbol{\mathcal{U}}}$ and $(\mathbf{A}^{i})_{i\in I}$ an $I$-indexed family of $\Sigma$-algebras, where, for every $i\in I$, $\mathbf{A}^{i} = (A^{i},F^{i})$. The \emph{product} of $(\mathbf{A}^{i})_{i\in I}$, denoted by $\prod_{i\in I}\mathbf{A}^{i}$, is the $\Sigma$-algebra $(\prod_{i\in I}A^{i},F)$ where, for every $\sigma\colon w\mor s$ in $\Sigma$, $F_{\sigma}$ is defined as follows:%
$$
F_{\sigma}
\nfunction
{(\prod_{i\in I}A^{i})_{w}}
{\prod_{i\in I}A^{i}_{s}}
{(a_{\alpha})_{\alpha\in \lvert w\rvert}}
{(F^{i}_{\sigma}((a_{\alpha}(i))_{\alpha\in \lvert w\rvert}))_{i\in I}}
$$
For every $i\in I$, the \emph{$i$th canonical projection}, $\mathrm{pr}^{i} = (\mathrm{pr}^{i}_{s})_{s\in S}$, is the homomorphism from $\prod_{i\in I}\mathbf{A}^{i}$ to $\mathbf{A}^{i}$ defined, for every $s\in S$, as follows:
$$
\mathrm{pr}^{i}_{s}
\nfunction
{\prod_{i\in I}A^{i}_{s}}
{A^{i}_{s}}
{(a_{i})_{i\in I}}
{a_{i}}
$$
On the other hand, if $\mathbf{B}$ is a $\Sigma$-algebra and $(f^{i})_{i\in I}$ an $I$-indexed family of homomorphisms, where, for every $i\in I$, $f^{i}$ is a homomorphism from $\mathbf{B}$ to $\mathbf{A}^{i}$, then we denote by $\left<f^{i}\right>_{i\in I}$ the unique homomorphism $f$ from $\mathbf{B}$ to $\prod_{i\in I}\mathbf{A}^{i}$ such that, for every $i\in I$, $\mathrm{pr}^{i}\circ f = f^{i}$.
\end{definition}

We next  define the concept of subfinal $\Sigma$-algebra. But before defining the just mentioned concept, we recall that $\mathbf{1}$, the final $\Sigma$-algebra in $\mathbf{Alg}(\Sigma)$, has as underlying $S$-sorted set $1 = (1)_{s\in S}$, the family constantly $1$, and, for every $(w,s)\in S^{\star}\times S$ and every formal operation $\sigma\in\Sigma_{w,s}$, as operation $F_{\sigma}$ from $1_{w} =  \prod_{i\in \lvert w\rvert}1_{w_{i}} = \{(\overbrace{0,\ldots,0}^{\lvert w\rvert})\}$ to $1_{s} = 1$ the unique mapping from $1_{w}$ to $1$.

\begin{definition}
A $\Sigma$-algebra $\mathbf{A}$ is \emph{subfinal} if $\mathbf{A}$ is isomorphic to a subalgebra of $\mathbf{1}$, the final $\Sigma$-algebra in $\mathbf{Alg}(\Sigma)$. We denote by $\mathrm{Sf}(\mathbf{1})$ the set of all subfinal $\Sigma$-algebras of $\mathbf{1}$.
\end{definition}

\begin{remark}
If $\mathbf{A}$ is a subfinal $\Sigma$-algebra, then $A$ is a subfinal $S$-sorted set. In fact, since there exists an $\mathbf{X}\in \mathrm{Sub}(\mathbf{1})$ such that $\mathbf{A}\cong \mathbf{X}$, then, by Proposition~\ref{propssupport}, $\mathrm{supp}_{S}(\mathbf{A}) = \mathrm{supp}_{S}(\mathbf{X})$, hence, for every $s\in S$, $\mathrm{card}(A_{s})\leq 1$, i.e., $A$ is a subfinal $S$-sorted set. On the other hand, if the $\Sigma$-algebra $\mathbf{A}$ is such that $A$ is a subfinal $S$-sorted set, then $\mathbf{A}$ is a subobject of $\mathbf{1}$, i.e., $\mathbf{A}$ is isomorphic to a subalgebra of $\mathbf{1}$. In fact, the $S$-sorted mapping $f$ from $A$ to $1 = (1)_{s\in S}$ which, for every $s\in S$, is the unique mapping from $A_{s}$ to $1$, determines an embedding from $\mathbf{A}$ to $\mathbf{1}$. Therefore, given a $\Sigma$-algebra $\mathbf{A}$, we have that $\mathbf{A}$ is a subfinal $\Sigma$-algebra if, and only if, $A$ is a subfinal $S$-sorted set. Furthermore, if $\mathbf{A}$ is a subfinal $\Sigma$-algebra, then, for every $\Sigma$-algebra $\mathbf{B}$, there exists at most a homomorphism from $\mathbf{B}$ to $\mathbf{A}$.
\end{remark}

Our next goal is to define the concepts of congruence on a $\Sigma$-algebra and of quotient of a $\Sigma$-algebra by a congruence on it. Moreover, we recall the notion of kernel of a homomorphism between $\Sigma$-algebras and the universal property of the quotient of a $\Sigma$-algebra by a congruence on it.

\begin{definition}
Let $\mathbf{A}$ be a $\Sigma$-algebra and $\Phi$ an $S$-sorted equivalence on $A$. We say that $\Phi$ is an
$S$-\emph{sorted congruence on} (or, to abbreviate, a \emph{congruence on}) $\mathbf{A}$ if, for every $(w,s)\in (S^{\star}-\{\lambda\})\times S$, $\sigma\colon w\mor s$,
and $a,b\in A_{w}$ we have that
$$
\frac
{\forall i\in \lvert w\rvert\mathrm{, }\,\, (a_{i}, b_{i})\in\Phi_{w_{i}} }
{(F_{\sigma}(a), F_{\sigma}(b))\in \Phi_{s}}\cdot
$$

We denote by $\mathrm{Cgr}(\mathbf{A})$ the set of all $S$-sorted congruences on $\mathbf{A}$ (which is an algebraic closure system on $A\times A$), by $\mathbf{Cgr}(\mathbf{A})$ the algebraic lattice  $(\Cgr(\mathbf{A}),\subseteq)$, by $\nabla^{\mathbf{A}}$ the greatest element of $\mathbf{Cgr}(\mathbf{A})$, and by $\Delta^{\mathbf{A}}$ the least element of $\mathbf{Cgr}(\mathbf{A})$.
\end{definition}

\begin{definition}
Let $\mathbf{A}$ be a $\Sigma$-algebra and $\Phi\in\mathrm{Cgr}(\mathbf{A})$. Then $\mathbf{A}/\Phi$, the \emph{quotient $\Sigma$-algebra} of $\mathbf{A}$ \emph{by} $\Phi$, is the $\Sigma$-algebra $(A/\Phi,F^{\mathbf{A}/\Phi})$, where, for every $\sigma\colon w\mor s$, the operation $F_{\sigma}^{\mathbf{A}/\Phi}\colon (A/\Phi)_{w}\mor A_{s}/\Phi_{s}$, also denoted, to simplify, by $F_{\sigma}$, is defined, for every $([a_{i}]_{\Phi_{w_{i}}})_{i\in\lvert w\rvert}\in (A/\Phi)_{w}$, as follows: %
$$
F_{\sigma}
\nfunction
{(A/\Phi)_{w}} {A_{s}/\Phi_{s}}
{([a_{i}]_{\Phi_{w_{i}}})_{i\in\lvert w\rvert}}
{[F_{\sigma}((a_{i})_{i\in \lvert w\rvert})]_{\Phi_{s}}}
$$
And $\mathrm{pr}^{\Phi}\colon \mathbf{A}\mor \mathbf{A}/\Phi$, the \emph{canonical projection from} $\mathbf{A}$ \emph{to} $\mathbf{A}/\Phi$, is the homomorphism determined by the $S$-sorted mapping $\mathrm{pr}^{\Phi}$ from $A$ to $A/\Phi$.
\end{definition}

\begin{proposition}
Let $f$ be a homomorphism from $\mathbf{A}$ to $\mathbf{B}$. Then $\mathrm{Ker}(f)$ is a congruence on $\mathbf{A}$.
\end{proposition}

\begin{proposition}
Let $\mathbf{A}$ be a $\Sigma$-algebra and $\Phi\in \mathrm{Cgr}(\mathbf{A})$. Then the pair $(\mathrm{pr}^{\Phi},\mathbf{A}/\Phi)$ is such that:
\begin{enumerate}
\item $\mathrm{Ker}(\mathrm{pr}^{\Phi}) = \Phi$.
\item \emph{(Universal property)} For every homomorphism $f\colon\mathbf{A}\mor \mathbf{B}$, if $\Phi\subseteq \mathrm{Ker}(f)$, then there exists a unique homomorphism $\mathrm{p}^{\Phi,\mathrm{Ker}(f)}$ from $\mathbf{A}/\Phi$ to $\mathbf{B}$ such that $f = \mathrm{p}^{\Phi,\mathrm{Ker}(f)}\circ \mathrm{pr}^{\Phi}$.
\end{enumerate}
\end{proposition}

Given a $\Sigma$-algebra $\mathbf{A}$ and two congruences $\Phi$ and $\Psi$ on $\mathbf{A}$, if $\Phi\subseteq \Psi$, then, in the sequel, unless otherwise stated, $\mathrm{p}^{\Phi,\Psi}$ stands for the unique homomorphism from $\mathbf{A}/\Phi$ to $\mathbf{A}/\Psi$ such that $\mathrm{p}^{\Phi,\Psi}\circ \mathrm{pr}^{\Phi} = \mathrm{pr}^{\Psi}$.

\begin{remark}
Let $\mathbf{A}$ be a $\Sigma$-algebra. Then, for the congruence $\nabla^{\mathbf{A}}$, we have that $\mathbf{A}/\nabla^{\mathbf{A}}$ is isomorphic to a subalgebra of $\mathbf{1}$. In the single-sorted case, if $\pmb{\varnothing}$ is an algebra, then $\pmb{\varnothing}/\nabla^{\pmb{\varnothing}}$ is $\pmb{\varnothing}$, which is a subalgebra of $\mathbf{1}$, and $\mathrm{Sub}(\mathbf{1}) = \{\pmb{\varnothing},\mathbf{1}\}$. In the many-sorted case, for a set of sorts $S$ such that $\mathrm{card}(S)\geq 2$ and an $S$-sorted signature $\Sigma$ such that, for every $s\in S$, $\Sigma_{\lambda,s} = \varnothing$, we have that $\pmb{\varnothing}^{S}$ is a $\Sigma$-algebra and that $\pmb{\varnothing}^{S}/\nabla^{\pmb{\varnothing}^{S}}$ is $\pmb{\varnothing}^{S}$. However, in contrast with what happens in the single-sorted case, in the many-sorted case, for a set of sorts $S$ and an $S$-sorted signature $\Sigma$ such as those above, there may be $\Sigma$-algebras $\mathbf{A}$ such that $\varnothing\subset\mathrm{supp}_{S}(\mathbf{A})\subset S$. Hence, for such a type of $\Sigma$-algebras, the quotient $\Sigma$-algebra $\mathbf{A}/\nabla^{\mathbf{A}}$ will be isomorphic to a subalgebra in $\mathrm{Sub}(\mathbf{1}) - \{\pmb{\varnothing}^{S},\mathbf{1}\}$.
\end{remark}

Following this we state that the forgetful functor $\mathrm{G}_{\Sigma}$ from $\mathbf{Alg}(\Sigma)$ to
$\mathbf{Set}^{S}$ has a left adjoint $\mathbf{T}_{\Sigma}$ which assigns to an $S$-sorted set $X$ the free $\Sigma$-algebra $\mathbf{T}_{\Sigma}(X)$ on $X$.

\begin{definition}
Let $\Sigma$ be an $S$-sorted signature and $X$ an $S$-sorted set. The \emph{algebra of} $\Sigma$-\emph{rows} in $X$, denoted by $\mathbf{W}_{\Sigma}(X)$, is defined as follows:
\begin{enumerate}
\item For every $s\in S$, ${\mathrm{W}_{\Sigma}(X)}_{s}=
      \fmon{(\coprod\Sigma \bcoprod \coprod X)}$, i.e., the underlying
      $S$\nobreakdash-sorted set of $\mathbf{W}_{\Sigma}(X)$ is, for every $s\in S$,
      the set of all words on the alphabet $\coprod\Sigma \bcoprod \coprod X$.

\item For every $(w,s)\in\fmon{S}\times S$, and every $\sigma\in\Sigma_{w,s}$, the structural operation $F_{\sigma}$ associated to $\sigma$ is the mapping from $\mathrm{W}_{\Sigma}(X)_{w}$ to ${\mathrm{W}_{\Sigma}(X)}_{s}$ defined as follows:
      $$
      F_{\sigma}
      \nfunction
      {\mathrm{W}_{\Sigma}(X)_{w}}
      {{\mathrm{W}_{\Sigma}(X)}_{s}}
      {(P_{i})_{i\in\bb{w}}}
      {(\sigma)\bconcat\concat_{i\in\bb{w}}P_{i}}
      $$
      where $(\sigma)$ abbreviates $(((\sigma,(w,s)),0))$, which,
      in its turn, is obtained as indicated in the following figure
      $$\xymatrix@C=2.50pc@R=1pc{
      \Sigma_{w,s} \ar[r]^-{\inc_{\Sigma_{w,s}}} & \coprod \Sigma
      \ar[r]^-{\inc_{\coprod\Sigma}} &
      \coprod\Sigma\bcoprod\coprod X
      \ar[r]^-{\eta_{\coprod\Sigma\bcoprod\coprod X}} &
      \fmon{(\coprod\Sigma\bcoprod\coprod X)} \\
      \sigma \ar@{|->}[r]  & (\sigma,(w,s)) \ar@{|->}[r] &
      ((\sigma,(w,s)),0) \ar@{|->}[r]  &
      (((\sigma,(w,s)),0))\equiv(\sigma) }
      $$
\end{enumerate}
\end{definition}

\begin{definition}
The \emph{free} $\Sigma$-\emph{algebra on} an $S$-sorted set $X$, denoted by $\mathbf{T}_{\Sigma}(X)$, is the $\Sigma$-algebra determined by $\mathrm{Sg}_{\mathbf{W}_{\Sigma}(X)}((\{(x)\mid x\in X_{s}\})_{s\in S})$, the subalgebra of $\mathbf{W}_{\Sigma}(X)$ generated by $(\{(x)\mid x\in X_{s}\})_{s\in S}$, where, for every $s\in S$ and every $x\in X_{s}$, $(x)$ abbreviates $(((x,s),1))$, which, in its turn, is obtained as indicated in the following figure
$$
\xymatrix@C=3pc@R=1pc{
X_{s} 
      \ar[r]^-{\inc_{X_{s}}} &
\coprod X \ar[r]^-{\inc_{\coprod X}} &
\coprod\Sigma\bcoprod\coprod X
\ar[r]^-{\eta_{\coprod\Sigma\bcoprod\coprod X}} &
\fmon{(\coprod\Sigma\bcoprod\coprod X)} \\
x \ar@{|->}[r]  & (x,s) \ar@{|->}[r]  & ((x,s),1) \ar@{|->}[r]  &
(((x,s),1))\equiv(x) }
$$
\end{definition}

\begin{proposition}\label{rut}
Let $\Sigma$ be an $S$-sorted signature and $X$ an $S$-sorted set. Then, for every $s\in S$ and every $P\in
\mathrm{W}_{\Sigma}(X)_{s}$, we have that $P\in \mathrm{T}_{\Sigma}(X)_{s}$, if and only if
\begin{enumerate}
\item $P = (x)$, for a unique $x\in
      X_{s}$, or

\item $P = (\sigma)$, for a unique $\sigma\in\Sigma_{\lambda,s}$,
      or

\item $P = (\sigma)\bconcat\concat(P_{i})_{i\in\bb{w}}$, for
      a unique $w\in\fmon{S}-\{\lambda\}$, a unique
      $\sigma\in\Sigma_{w,s}$, and a unique family
      $(P_{i})_{i\in\bb{w}}$ in $\in\mathrm{T}_{\Sigma}(X)_{w}$.
\end{enumerate}
Moreover, the three possibilities are mutually exclusive.
\end{proposition}

From the above proposition it follows, immediately, the universal property of the free $\Sigma$-algebra on an $S$-sorted set $X$, as stated in the subsequent proposition.

\begin{proposition}
For every $S$-sorted set $X$, the pair $(\eta^{X},\mathbf{T}_{\Sigma}(X))$, where $\eta^{X}$, the
\emph{insertion (of the generators)} $X$ \emph{into} $\mathrm{T}_{\Sigma}(X)$, is the co-restric\-tion to
$\mathrm{T}_{\Sigma}(X)$ of the canonical embedding of $X$ into $\mathrm{W}_{\Sigma}(X)$, is a universal morphism from $X$ to $\mathbf{T}_{\Sigma}$, i.e., for every $\Sigma$-algebra $\mathbf{A}$ and every $S$-sorted mapping $f\colon X\mor A$, there exists a unique homomorphism $f^{\sharp}\colon\mathbf{T}_{\Sigma}(X)\mor\mathbf{A}$ such that $f^{\sharp}\circ \eta^{X} = f$.
\end{proposition}

\begin{corollary} \label{FladjG}
The functor $\mathbf{T}_{\Sigma}$ (which assigns to an $S$-sorted set $A$, $\mathbf{T}_{\Sigma}(A)$, and to an $S$-sorted mapping  $f\colon A\mor B$, $(\eta^{B}\circ f)^{\sharp}$) is left adjoint for the forgetful functor $\G_{\Sigma}$ from $\mathbf{Alg}(\Sigma)$ to $\mathbf{Set}^{S}$.
\end{corollary}

We next recall a lemma which, together with the universal property of the free $\Sigma$-algebra on an $S$-sorted set, allows one to prove that every free $\Sigma$-algebra on an $S$-sorted set is projective.

\begin{lemma}
Let $X$ be an $S$-sorted set, $\mathbf{A}$ a $\Sigma$-algebra, and $f$, $g$ two homomorphisms from $\mathbf{T}_{\Sigma}(X)$ to $\mathbf{A}$. If $f\circ\eta^{X} = g\circ\eta^{X}$, then $f = g$.
\end{lemma}

\begin{proposition}\label{FreeProj}
Let $X$ be an $S$-sorted set, then $\mathbf{T}_{\Sigma}(X)$ is projective, i.e., for every epimorphism $f\colon \mathbf{A}\mor \mathbf{B}$ and every homomorphism $g\colon \mathbf{T}_{\Sigma}(X)\mor \mathbf{B}$, there exists a homomorphism $h\colon \mathbf{T}_{\Sigma}(X)\mor \mathbf{A}$ such that $f\circ h = g$.
\end{proposition}

We next recall that every $\Sigma$-algebra is a homomorphic image of a free $\Sigma$-algebra on an $S$-sorted set.

\begin{proposition}\label{AlgIsoQuotFree}
Let $\mathbf{A}$ be a $\Sigma$-algebra. Then $\mathbf{A}$ is isomorphic to a quotient of a free $\Sigma$-algebra on an $S$-sorted set.
\end{proposition}

We next define the concept of subdirect product of a family of $\Sigma$-algebras. But before doing that, for two $\Sigma$-algebras $\mathbf{A}$ and $\mathbf{B}$, from now on $\mathrm{Mon}(\mathbf{A},\mathbf{B})$ stands for the set of all monomorphisms from $\mathbf{A}$ to $\mathbf{B}$ and $\mathrm{Epi}(\mathbf{A},\mathbf{B})$ stands for the set of all epimorphisms from $\mathbf{A}$ to $\mathbf{B}$.

\begin{definition}\label{Sdprod}
Let $I$ be a set in $\boldsymbol{\mathcal{U}}$. A $\Sigma$-algebra $\mathbf{A}$ is a \emph{subdirect product} of a family of $\Sigma$-algebras $(\mathbf{A}^{i})_{i\in I}$ if it satisfies the following conditions:
\begin{enumerate}
\item $\mathbf{A}$ is a subalgebra of $\prod_{i\in I}\mathbf{A}^{i}$.

\item For every $i\in I$, $\pr^{i}\rest \mathbf{A}$ is surjective, where $\pr^{i}\rest \mathbf{A}$ is the restriction to $\mathbf{A}$ of $\pr^{i}\colon \prod_{i\in I}\mathbf{A}^{i}\mor \mathbf{A}^{i}$.
\end{enumerate}
On the other hand, we will say that an embedding ($\equiv$ injective homomorphism) $f\colon\mathbf{A}\mor \prod_{i\in I}\mathbf{A}^{i}$ is \emph{subdirect} if $f[\mathbf{A}]$, the $\Sigma$-algebra canonically associated to the subalgebra $f[A]$ of $\prod_{i\in I}\mathbf{A}^{i}$, is a subdirect product of $(\mathbf{A}^{i})_{i\in I}$. We will denote by $\mathrm{Em}_{\mathrm{sd}}(\mathbf{A},\prod_{i\in I}\mathbf{A}^{i})$ the set of all subdirect embeddings of $\mathbf{A}$ in  $\prod_{i\in I}\mathbf{A}^{i}$, i.e., the subset of $\mathrm{Mon}(\mathbf{A},\prod_{i\in I}\mathbf{A}^{i})$ defined as follows:
$$
\textstyle
\mathrm{Em}_{\mathrm{sd}}(\mathbf{A},\prod_{i\in I}\mathbf{A}^{i}) = \{f\in \mathrm{Mon}(\mathbf{A},\prod_{i\in I}\mathbf{A}^{i})\mid \forall i\in I\, (\mathrm{pr}^{i}\circ f\in \mathrm{Epi}(\mathbf{A},\mathbf{A}^{i}))\}.
$$
Moreover, we will say that two subdirect embeddings $f\colon\mathbf{A}\mor \prod_{i\in I}\mathbf{A}^{i}$ and $g\colon\mathbf{A}\mor \prod_{i\in I}\mathbf{B}^{i}$ are \emph{isomorphic} if, and only if, there exists a family $(h^{i})_{i\in I}\in\prod_{i\in I}\mathrm{Iso}(\mathbf{A}^{i},\mathbf{B}^{i})$, of isomorphisms, such that, for every $i\in I$, $h^{i}\circ \mathrm{pr}^{\mathbf{A}^{i}}\circ f = \mathrm{pr}^{\mathbf{B}^{i}}\circ g$.
\end{definition}

\section{$\Sigma$-congruence formations, $\Sigma$-algebra formations, and an Eilenberg type theorem for them.}

In this section we define, for a fixed set of sorts $S$ and a fixed $S$-sorted signature $\Sigma$, the concepts of formation of congruences with respect to $\Sigma$, of formation of algebras with respect to $\Sigma$, and of Shemetkov\!$\And$\!Skiba-formation of algebras with respect to $\Sigma$, which is a generalization to the many-sorted case of that proposed by the mentioned authors  in~\cite{shsk89} for the single-sorted case, and of which we prove that is equivalent to that of formation of algebras with respect to $\Sigma$. Moreover, we investigate the properties of the aforementioned formations and prove that there exists an isomorphism between the algebraic lattice of all $\Sigma$-algebra formations and the algebraic lattice of all $\Sigma$-congruence formations, which can be considered as an Eilenberg type theorem.

Before defining, for an $S$-sorted signature $\Sigma$, the concept of $\Sigma$-congruence formation, we next recall the concept of filter of a lattice since it will be necessary to state the definition of the just mentioned concept.

\begin{definition}
Let $\mathbf{L} = (L,\vee,\wedge)$ be a lattice. We say that $F\subseteq L$ is a \emph{filter} of $\mathbf{L}$ if it satisfies the following conditions:
\begin{enumerate}
\item $F\neq \varnothing$.
\item For every $x$, $y\in F$ we have that $x\wedge y\in F$.
\item For every $x\in F$ and $y\in L$, if $x\leq y$, then $y\in F$.
\end{enumerate}
We denote by $\mathrm{Filt}(\mathbf{L})$ the set of all filters of $\mathbf{L}$.
\end{definition}

We next define, for a many-sorted signature $\Sigma$, the notion of formation of $\Sigma$-congruences which will be used through this article. This notion was defined, for monoids, by Cosme in~\cite{cll15} on p.~53.

\begin{definition}\label{DefFormCgr}
A \emph{formation of congruences with respect to} $\Sigma$ is a function $\mathfrak{F}$ from $\boldsymbol{\mathcal{U}}^{S}$ such that the following conditions are satisfied:
\begin{enumerate}
\item For every $A\in \boldsymbol{\mathcal{U}}^{S}$, $\mathfrak{F}(A)$ is a filter of the algebraic lattice $\mathbf{Cgr}(\mathbf{T}_{\Sigma}(A))$, i.e., $\mathfrak{F}(A)$ is a nonempty subset of $\mathrm{Cgr}(\mathbf{T}_{\Sigma}(A))$, for every $\Phi$, $\Psi\in \mathfrak{F}(A)$ we have that $\Phi\cap \Psi\in \mathfrak{F}(A)$, and, for every $\Phi\in \mathfrak{F}(A)$ and every $\Psi\in \mathrm{Cgr}(\mathbf{T}_{\Sigma}(A))$, if $\Phi\subseteq \Psi$, then $\Psi\in \mathfrak{F}(A)$.

\item For every $A$, $B\in \boldsymbol{\mathcal{U}}^{S}$, every congruence $\Theta\in \mathfrak{F}(B)$, and every homomorphism $f\colon \mathbf{T}_{\Sigma}(A)\mor \mathbf{T}_{\Sigma}(B)$, if $\mathrm{pr}^{\Theta}\circ f\colon \mathbf{T}_{\Sigma}(A)\mor \mathbf{T}_{\Sigma}(B)/\Theta$ is an epimorphism, then $\mathrm{Ker}(\mathrm{pr}^{\Theta}\circ f)\in \mathfrak{F}(A)$.
\end{enumerate}

We denote by $\mathrm{Form}_{\mathrm{Cgr}}(\Sigma)$ the set of all formations of congruences with respect to $\Sigma$. Let us notice that $\mathrm{Form}_{\mathrm{Cgr}}(\Sigma)\subseteq \prod_{A\in \boldsymbol{\mathcal{U}}^{S}} \mathrm{Filt}(\mathbf{Cgr}(\mathbf{T}_{\Sigma}(A)))$, where, for every $A\in \boldsymbol{\mathcal{U}}^{S}$, $\mathrm{Filt}(\mathbf{Cgr}(\mathbf{T}_{\Sigma}(A)))$ is the set of all filters of the algebraic lattice $\mathbf{Cgr}(\mathbf{T}_{\Sigma}(A))$. Therefore a formation of congruences with respect to $\Sigma$ is a special type of choice function for $(\mathrm{Filt}(\mathbf{Cgr}(\mathbf{T}_{\Sigma}(A))))_{A\in \boldsymbol{\mathcal{U}}^{S}}$.
\end{definition}

\begin{remark}
If $\mathcal{V}$ is a variety of $\Sigma$-algebras (see~\cite{m76}, p.~57), a finitary variety of $\Sigma$-algebras (see~\cite{m76}, p.~56), or an Eilenberg's variety of $\Sigma$-algebras, then the function $\mathfrak{F}_{\mathcal{V}}$ from $\boldsymbol{\mathcal{U}}^{S}$ which assigns to $A\in \boldsymbol{\mathcal{U}}^{S}$ the set
$$
\mathfrak{F}_{\mathcal{V}}(A) = \{\Phi\in \mathrm{Cgr}(\mathbf{T}_{\Sigma}(A))\mid \mathbf{T}_{\Sigma}(A)/\Phi\in \mathcal{V}\}
$$
is a formation of congruences with respect to $\Sigma$.
\end{remark}

Since two formations of congruences $\mathfrak{F}$ and $\mathfrak{G}$ with respect to $\Sigma$ can be compared in a natural way, e.g., by stating that $\mathfrak{F}\leq \mathfrak{G}$ if, and only if, for every $A\in \boldsymbol{\mathcal{U}}^{S}$, $\mathfrak{F}(A)\subseteq \mathfrak{G}(A)$, we next proceed to investigate the properties of $\mathbf{Form}_{\mathrm{Cgr}}(\Sigma) = (\mathrm{Form}_{\mathrm{Cgr}}(\Sigma),\leq)$.

\begin{proposition}
$\mathbf{Form}_{\mathrm{Cgr}}(\Sigma)$ is a complete lattice.
\end{proposition}

\begin{proof}
It is obvious that $\mathbf{Form}_{\mathrm{Cgr}}(\Sigma)$ is an ordered set. On the other hand, if we take as choice function for the family $(\mathrm{Filt}(\mathbf{Cgr}(\mathbf{T}_{\Sigma}(A))))_{A\in \boldsymbol{\mathcal{U}}^{S}}$ the function $\mathfrak{F}$ defined, for every $A\in \boldsymbol{\mathcal{U}}^{S}$, as $\mathfrak{F}(A) = \mathrm{Cgr}(\mathbf{T}_{\Sigma}(A))$, then $\mathfrak{F}$ is a formation of congruences with respect to $\Sigma$ and, actually, the greatest one. Let us, finally, prove that, for every nonempty set $J$ in $\boldsymbol{\mathcal{U}}$ and every family $(\mathfrak{F}_{j})_{j\in J}$ in $\mathrm{Form}_{\mathrm{Cgr}}(\Sigma)$, there exists $\bigwedge_{j\in J}\mathfrak{F}_{j}$, the greatest lower bound of $(\mathfrak{F}_{j})_{i\in J}$ in $\mathbf{Form}_{\mathrm{Cgr}}(\Sigma)$. Let $\bigwedge_{j\in J}\mathfrak{F}_{j}$ be the function defined, for every $A\in \boldsymbol{\mathcal{U}}^{S}$, as $(\bigwedge_{j\in J}\mathfrak{F}_{j})(A) = \bigcap_{j\in J}\mathfrak{F}_{j}(A)$. It is straightforward to prove that, thus defined, $\bigwedge_{j\in J}\mathfrak{F}_{j}$ is a choice function for the family $(\mathrm{Filt}(\mathbf{Cgr}(\mathbf{T}_{\Sigma}(A))))_{A\in \boldsymbol{\mathcal{U}}^{S}}$ and that it satisfies the second condition in the definition of formation of congruences with respect to $\Sigma$. Moreover, for every $j\in J$, we have that $\bigwedge_{j\in J}\mathfrak{F}_{j}\leq \mathfrak{F}_{j}$ and, for every formation of congruences with respect to $\Sigma$, $\mathfrak{F}$, if, for every $j\in J$, we have that $\mathfrak{F}\leq \mathfrak{F}_{j}$, then $\mathfrak{F}\leq \bigwedge_{j\in J}\mathfrak{F}_{j}$.
 From this we can assert that the ordered set $\mathbf{Form}_{\mathrm{Cgr}}(\Sigma)$ is a complete lattice.

Let us recall that, for every nonempty set $J$ in $\boldsymbol{\mathcal{U}}$ and every family $(\mathfrak{F}_{j})_{j\in J}$ in $\mathrm{Form}_{\mathrm{Cgr}}(\Sigma)$, $\bigvee_{j\in J}\mathfrak{F}_{j}$, the least upper bound of $(\mathfrak{F}_{j})_{i\in J}$ in $\mathbf{Form}_{\mathrm{Cgr}}(\Sigma)$, is obtained as:
$$
\textstyle
\bigvee_{j\in J}\mathfrak{F}_{j} = \bigwedge\{\mathfrak{F}\in \mathrm{Form}_{\mathrm{Cgr}}(\Sigma)\mid \forall j\in J\, (\mathfrak{F}_{j}\leq \mathfrak{F})\}.
$$
Moreover, if we take as choice function for the family $(\mathrm{Filt}(\mathbf{Cgr}(\mathbf{T}_{\Sigma}(A))))_{A\in \boldsymbol{\mathcal{U}}^{S}}$ the function $\mathfrak{F}$ defined, for every $A\in \boldsymbol{\mathcal{U}}^{S}$, as $\mathfrak{F}(A) = \{\nabla^{\mathbf{T}_{\Sigma}(A)}\}$, where $\nabla^{\mathbf{T}_{\Sigma}(A)}$ is the largest congruence on $\mathbf{T}_{\Sigma}(A)$, then $\mathfrak{F}$ is a formation of congruences with respect to $\Sigma$ and, actually, the smallest one.
\end{proof}
%
%

\begin{remark}
Afterwards we will improve the above lattice-theoretic results about $\mathbf{Form}_{\mathrm{Cgr}}(\Sigma)$ by proving that it is, in fact, an algebraic lattice.
\end{remark}

We next define two operators, $\mathrm{H}$ and $\mathrm{P}_{\mathrm{fsd}}$, on $\mathrm{Alg}(\Sigma)$, i.e., two mappings of $\mathrm{Sub}(\mathrm{Alg}(\Sigma))$ into itself, which will be used afterwards.

\begin{definition}
Let $\mathcal{F}$ be a set of $\Sigma$-algebras, i.e., a subset of the $\boldsymbol{\mathcal{U}}$-large set $\mathrm{Alg}(\Sigma)$. Then
\begin{enumerate}
\item $\mathrm{H}(\mathcal{F})$ stands for the set of all homomorphic images of members of $\mathcal{F}$, i.e., for the set defined as:
    $$
    \mathrm{H}(\mathcal{F}) = \{\mathbf{A}\in\mathrm{Alg}(\Sigma)\mid \exists\, \mathbf{B}\in \mathcal{F}\,(\mathrm{Epi}(\mathbf{B},\mathbf{A})\neq\varnothing) \},\, \text{and}
    $$
\item $\mathrm{P}_{\mathrm{fsd}}(\mathcal{F})$ stands for the subset of $\mathrm{Alg}(\Sigma)$ defined as follows. For every $\Sigma$-algebra $\mathbf{A}$, we have that $\mathbf{A}\in \mathrm{P}_{\mathrm{fsd}}(\mathcal{F})$ if, and only if, for some $n\in \mathbb{N}$ and some family $(\mathbf{C}^{\alpha})_{\alpha\in n} \in \mathcal{F}^{n}$,  $\mathrm{Em}_{\mathrm{sd}}(\mathbf{A},\prod_{\alpha\in n}\mathbf{C}^{\alpha})\neq\varnothing$.
\end{enumerate}
\end{definition}

\begin{proposition}
Let $\mathfrak{F}$ be a formation of congruences with respect to $\Sigma$. Then the subset $\mathcal{F}_{\mathfrak{F}}$ of $\mathrm{Alg}(\Sigma)$ defined as follows:
$$
\mathcal{F}_{\mathfrak{F}} = \biggl\{ \mathbf{C}\in \mathrm{Alg}(\Sigma)\biggm|
\begin{gathered}
\exists\, A\in \boldsymbol{\mathcal{U}}^{S}\,\, \exists\, \Phi\in \mathfrak{F}(A)
\\[-3pt]
(\mathbf{C}\cong \mathbf{T}_{\Sigma}(A)/\Phi)
\end{gathered}
\biggr\},
$$
has the following properties:
\begin{enumerate}
\item $\mathcal{F}_{\mathfrak{F}}\neq \varnothing$.
\item If $\mathbf{C}\in \mathcal{F}_{\mathfrak{F}}$ and $\mathbf{D}$ is a $\Sigma$-algebra such that $\mathbf{D}\cong \mathbf{C}$, then $\mathbf{D}\in \mathcal{F}_{\mathfrak{F}}$, i.e., $\mathcal{F}_{\mathfrak{F}}$ is abstract.
\item $\mathrm{H}(\mathcal{F}_{\mathfrak{F}})\subseteq \mathcal{F}_{\mathfrak{F}}$, i.e., $\mathcal{F}_{\mathfrak{F}}$ is closed under the formation of homomorphic images of members of $\mathcal{F}_{\mathfrak{F}}$.
\item $\mathrm{P}_{\mathrm{fsd}}(\mathcal{F}_{\mathfrak{F}})\subseteq \mathcal{F}_{\mathfrak{F}}$, i.e., for every $\Sigma$-algebra $\mathbf{A}$, if, for some $n\in \mathbb{N}$ and some family $(\mathbf{C}^{\alpha})_{\alpha\in n}\in \mathcal{F}_{\mathfrak{F}}^{n}$, $\mathrm{Em}_{\mathrm{sd}}(\mathbf{A},\prod_{\alpha\in n}\mathbf{C}^{\alpha})\neq\varnothing$, then $\mathbf{A}\in \mathcal{F}_{\mathfrak{F}}$.
\end{enumerate}
\end{proposition}

\begin{proof}
The first property is evident (it suffices to verify that $\mathbf{1}\in \mathcal{F}_{\mathfrak{F}}$).

The second property is also obvious, since the composition of isomorphisms is an isomorphism.

To verify the third property let $\mathbf{C}$ be an element of $\mathcal{F}_{\mathfrak{F}}$ and $f\colon \mathbf{C}\mor \mathbf{D}$ an epimorphism. Since $\mathbf{C}\in \mathcal{F}_{\mathfrak{F}}$ there exists an $A\in \boldsymbol{\mathcal{U}}^{S}$ and a $\Phi\in \mathfrak{F}(A)$ such that $\mathbf{C}\cong \mathbf{T}_{\Sigma}(A)/\Phi$. Let $g$ be a fixed isomorphism from $\mathbf{T}_{\Sigma}(A)/\Phi$ to $\mathbf{C}$. Then the homomorphism $f\circ g\circ \mathrm{pr}^{\Phi}$ from $\mathbf{T}_{\Sigma}(A)$ to $\mathbf{D}$ is an epimorphism. Hence $\mathbf{T}_{\Sigma}(A)/\mathrm{Ker}(f\circ g\circ \mathrm{pr}^{\Phi})$ is isomorphic to $\mathbf{D}$. But $\Phi\subseteq \mathrm{Ker}(f\circ g\circ \mathrm{pr}^{\Phi})$. Thus $\mathrm{Ker}(f\circ g\circ \mathrm{pr}^{\Phi})\in \mathfrak{F}(A)$. Therefore, $\mathbf{D}\in \mathcal{F}_{\mathfrak{F}}$.

To verify the fourth property let $\mathbf{C}$ be a $\Sigma$-algebra such that $\mathbf{C}\in \mathrm{P}_{\mathrm{fsd}}(\mathcal{F}_{\mathfrak{F}})$. Then, by definition of $\mathrm{P}_{\mathrm{fsd}}(\mathcal{F}_{\mathfrak{F}})$, for some $n\in \mathbb{N}$ and some family $(\mathbf{C}^{\alpha})_{\alpha\in n}\in \mathcal{F}_{\mathfrak{F}}^{n}$, we have that $\mathrm{Em}_{\mathrm{sd}}(\mathbf{C},\prod_{\alpha\in n}\mathbf{C}^{\alpha})\neq\varnothing$. Hence there exists a family $(A^{\alpha})_{\alpha\in n}\in (\boldsymbol{\mathcal{U}}^{S})^{n}$ and a family $(\Phi^{\alpha})_{\alpha\in n}\in \prod_{\alpha\in n}\mathfrak{F}(A^{\alpha})$ such that, for every $\alpha \in n$, $\mathbf{C}^{\alpha}\cong \mathbf{T}_{\Sigma}(A^{\alpha})/\Phi^{\alpha}$.

Let $f\colon \mathbf{C}\mor\prod_{\alpha\in n}\mathbf{T}_{\Sigma}(A^{\alpha})/\Phi^{\alpha}$ be a subdirect embedding (here we use the notion of isomorphism between subdirect embeddings as stated in Definition~\ref{Sdprod}), $B\in\boldsymbol{\mathcal{U}}^{S}$, and $g$ an epimorphism from $\mathbf{T}_{\Sigma}(B)$ to $\mathbf{C}$ (recall that, by Proposition~\ref{AlgIsoQuotFree}, every $\Sigma$-algebra is isomorphic to a quotient of a free $\Sigma$-algebra on an $S$-sorted set). Then, since, by Proposition~\ref{FreeProj}, $\mathbf{T}_{\Sigma}(B)$ is projective, for every $\alpha\in n$, there exists a homomorphism $h^{\alpha}$ from $\mathbf{T}_{\Sigma}(B)$ to $\mathbf{T}_{\Sigma}(A^{\alpha})$ such that the following diagram
$$
\xymatrix@C=15pt@R=15pt{
\mathbf{T}_{\Sigma}(B)\ar[ddd]_-{h^{\alpha}}\ar@{+>}[rd]^*{g}& {} & {} & {}\\
{} & \mathbf{C}\ar@{{ +}{-}{>}}[rd]^*{f}\ar@{+>} @/_-3pc/[ddrr]^*[l]{\!\!\!\mathrm{pr}^{\alpha}\circ f} & {} & {} \\
{} & {} & \prod_{\alpha\in n}\mathbf{T}_{\Sigma}(A^{\alpha})/\Phi^{\alpha}\ar@{+>}[rd]^-{\mathrm{pr}^{\alpha}} & {} \\
\mathbf{T}_{\Sigma}(A^{\alpha})\ar@{+>}[rrr]_-{\mathrm{pr}^{\Phi^{\alpha}}}& {} & {} & \mathbf{T}_{\Sigma}(A^{\alpha})/\Phi^{\alpha}
}
$$
commutes. Besides, since, for every $\alpha\in n$, $\mathrm{pr}^{\alpha}\circ f\circ g$ is an epimorphism and the above diagram commutes, we have that, for every $\alpha\in n$, $\mathrm{pr}^{\Phi^{\alpha}}\circ h^{\alpha}$ is an epimorphism. Therefore, because $\mathfrak{F}$ is a formation of congruences with respect to $\Sigma$, we have that, for every $\alpha\in n$, $\mathrm{Ker}(\mathrm{pr}^{\Phi^{\alpha}}\circ h^{\alpha})\in \mathfrak{F}(B)$. Hence $\bigcap_{\alpha\in n}\mathrm{Ker}(\mathrm{pr}^{\Phi^{\alpha}}\circ h^{\alpha})\in \mathfrak{F}(B)$. We next proceed to show that the congruence $\bigcap_{\alpha\in n}\mathrm{Ker}(\mathrm{pr}^{\Phi^{\alpha}}\circ h^{\alpha})$ is included in $\mathrm{Ker}(g)$. Let $s$ be a sort in $S$ and $(P,Q)\in (\bigcap_{\alpha\in n}\mathrm{Ker}(\mathrm{pr}^{\Phi^{\alpha}}\circ h^{\alpha}))_{s} = \bigcap_{\alpha\in n}\mathrm{Ker}(\mathrm{pr}^{\Phi^{\alpha}}_{s}\circ h^{\alpha}_{s})$. Then, for every $\alpha\in n$, $\mathrm{pr}^{\Phi^{\alpha}}_{s}(h^{\alpha}_{s}(P)) = \mathrm{pr}^{\Phi^{\alpha}}_{s}(h^{\alpha}_{s}(Q))$. Thus, for every $\alpha\in n$, $\mathrm{pr}^{\alpha}_{s}(f_{s}(g_{s}(P))) = \mathrm{pr}^{\alpha}_{s}(f_{s}(g_{s}(Q)))$. So, because projections, acting conjointly, act monomorphically, $f_{s}(g_{s}(P)) = f_{s}(g_{s}(Q))$. But $f_{s}$ is a monomorphism, hence $g_{s}(P) = g_{s}(Q)$, i.e., $(P,Q)\in \mathrm{Ker}(g)_{s} = \mathrm{Ker}(g_{s})$. This proves that $\bigcap_{\alpha\in n}\mathrm{Ker}(\mathrm{pr}^{\Phi^{\alpha}}\circ h^{\alpha})\subseteq \mathrm{Ker}(g)$. Therefore $\mathrm{Ker}(g)\in \mathfrak{F}(B)$. Consequently, because $\mathbf{T}_{\Sigma}(B)/\mathrm{Ker}(g)\cong \mathbf{C}$, it follows that $\mathbf{C}\in \mathcal{F}_{\mathfrak{F}}$.
\end{proof}

\begin{remark}
If the $\Sigma$-algebra $\mathbf{B}$ is such that, for some $A\in \boldsymbol{\mathcal{U}}^{S}$ and some $\Phi\in \mathfrak{F}(A)$, $\mathbf{C}\cong \mathbf{T}_{\Sigma}(A)/\Phi$, then, by Proposition~\ref{propssupport}, $\mathrm{supp}_{S}(\mathbf{B}) = \mathrm{supp}_{S}(\mathbf{T}_{\Sigma}(A))$.
\end{remark}

We next define, for a many-sorted signature $\Sigma$, the notion of formation of $\Sigma$-algebras which will be used through this article and afterwards we will prove that it is, in fact, equivalent to that of Shemetkov and Skiba in~\cite{shsk89} (after generalizing their definition from the single-sorted case to the many-sorted case).

\begin{definition}\label{DefFormAlg}
A \emph{formation of $\Sigma$-algebras} is a set of $\Sigma$-algebras $\mathcal{F}$ such that the following conditions are satisfied:
\begin{enumerate}
\item $\mathcal{F}\neq \varnothing$.
\item For every $\mathbf{A}\in \mathcal{F}$ and every $\mathbf{B}\in \mathrm{Alg}(\Sigma)$, if $\mathbf{B}\cong \mathbf{A}$, the $\mathbf{B}\in \mathcal{F}$, i.e., $\mathcal{F}$ is abstract.

\item $\mathrm{H}(\mathcal{F})\subseteq \mathcal{F}$, i.e., $\mathcal{F}$ is closed under the formation of homomorphic images of members of $\mathcal{F}$.

\item $\mathrm{P}_{\mathrm{fsd}}(\mathcal{F})\subseteq \mathcal{F}$, i.e., for every $\Sigma$-algebra $\mathbf{A}$, if, for some $n\in \mathbb{N}$ and some family $(\mathbf{C}^{\alpha})_{\alpha\in n}\in \mathcal{F}^{n}$,
    $\mathrm{Em}_{\mathrm{sd}}(\mathbf{A},\prod_{\alpha\in n}\mathbf{C}^{\alpha})\neq\varnothing$, then $\mathbf{A}\in \mathcal{F}$.
\end{enumerate}

We denote by $\mathrm{Form}_{\mathrm{Alg}}(\Sigma)$ the set of all formations of $\Sigma$-algebras.
\end{definition}

\begin{remark}
If $\mathcal{F}$ is a formation of $\Sigma$-algebras, then, since $\mathrm{Alg}(\Sigma)\subseteq \boldsymbol{\mathcal{U}}$, we have that $\mathcal{F}\subseteq\boldsymbol{\mathcal{U}}$, i.e., $\mathcal{F}$ is a $\boldsymbol{\mathcal{U}}$-large set. Therefore  $\mathrm{Form}_{\mathrm{Alg}}(\Sigma)\subseteq\mathrm{Sub}(\mathrm{Alg}(\Sigma))\subseteq \mathrm{Sub}(\boldsymbol{\mathcal{U}})$. Hence $\mathrm{Form}_{\mathrm{Alg}}(\Sigma)$ is a legitimate set in our underlying set theory.
\end{remark}

\begin{remark}
For an $S$-sorted signature $\Sigma$, all varieties, finitary varieties, and Eilenberg's pseudovarieties of $\Sigma$-algebras are examples of formations of $\Sigma$-algebras. Moreover, $\mathrm{Sf}(\mathbf{1})$, the set of subfinal $\Sigma$-algebras, i.e., the set of all $\Sigma$-algebras which are isomorphic to a subalgebra of $\mathbf{1}$, is a formation of $\Sigma$-algebras. In fact, $\mathrm{Sf}(\mathbf{1})\neq \varnothing$ since, obviously, $\mathbf{1}\in \mathrm{Sf}(\mathbf{1})$. Let $\mathbf{A}$ be an element of $\mathrm{Sf}(\mathbf{1})$ and  $\mathbf{B}$ a $\Sigma$-algebra such that $\mathbf{B}\cong \mathbf{A}$, then, clearly, $\mathbf{B}\in \mathrm{Sf}(\mathbf{1})$. Let $\mathbf{A}$ be an element of $\mathrm{Sf}(\mathbf{1})$ and $f$ an epimorphism from $\mathbf{A}$ to a $\Sigma$-algebra $\mathbf{B}$. Then, by Proposition~\ref{propssupport}, $\mathrm{supp}_{S}(A) = \mathrm{supp}_{S}(B)$ and $f$ is an isomorphism from $\mathbf{A}$ to $\mathbf{B}$, thus $\mathbf{B}\in\mathrm{Sf}(\mathbf{1})$. Finally, let $n$ be an element of $\mathbb{N}$,  $(\mathbf{C}^{\alpha})_{\alpha\in n}$ an $n$-indexed family in $\mathrm{Sf}(\mathbf{1})$, and $\mathbf{A}$ a $\Sigma$-algebra such that there exists a subdirect embedding $f$ of $\mathbf{A}$ in $\prod_{\alpha\in n}\mathbf{C}^{\alpha}$. Then $\prod_{\alpha\in n}\mathbf{C}^{\alpha}$ is a subalgebra of $\mathbf{1}$, $f[A]$ is a subalgebra of $\mathbf{1}$, and $\mathbf{A}$ is isomorphic to $f[A]$, hence $\mathbf{A}\in \mathrm{Sf}(\mathbf{1})$.
\end{remark}

\begin{remark}
For $n = 0 = \varnothing$, we have the empty family $(\mathbf{C}^{i})_{i\in \varnothing}\in \mathcal{F}^{\varnothing}$ and $\prod_{i\in \varnothing}\mathbf{C}^{i}$ is $\mathbf{1}$, the final $\Sigma$-algebra. Therefore, if $\mathcal{F}$ is a formation of $\Sigma$-algebras, then $\mathrm{Sf}(\mathbf{1})\subseteq \mathcal{F}$.
\end{remark}

According to Shemetkov and Skiba (see~\cite{shsk89}), for a single sorted signature $\Sigma$, a set of $\Sigma$-algebras $\mathcal{F}$ is a formation of $\Sigma$-algebras if the following conditions are satisfied:
\begin{enumerate}
\item For every $\mathbf{A}\in \mathcal{F}$ and every $\mathbf{B}\in \mathrm{Alg}(\Sigma)$, if $\mathbf{B}\cong \mathbf{A}$, the $\mathbf{B}\in \mathcal{F}$, i.e., $\mathcal{F}$ is abstract.
\item $\mathrm{H}(\mathcal{F})\subseteq \mathcal{F}$, i.e., $\mathcal{F}$ is closed under the formation of homomorphic images of members of $\mathcal{F}$.
\item For every $\Sigma$-algebra $\mathbf{A}$ and every $\Phi$, $\Psi\in \mathrm{Cgr}(\mathbf{A})$, if $\mathbf{A}/\Phi$ and $\mathbf{A}/\Psi\in \mathcal{F}$, then we have that  $\mathbf{A}/\Phi\cap \Psi\in \mathcal{F}$.
\end{enumerate}

Let us point out that for the aforementioned authors a single-sorted $\Sigma$-algebra is a \emph{nonempty} set together with an arbitrary system of algebraic operations.

We next define, for the many-sorted case, the notion of Shemetkov\!$\And$\!Skiba-formation of $\Sigma$-algebras, abbreviated to ShSk-formation of $\Sigma$-algebras, and prove that they are equivalent to those stated in Definition~\ref{DefFormAlg}.

\begin{definition}
An \emph{ShSk-formation of} $\Sigma$-\emph{algebras} is a set of $\Sigma$-algebras $\mathcal{F}$ such that the following conditions are satisfied:
\begin{enumerate}
\item  $\mathcal{F}\neq\varnothing$.
\item For every $\mathbf{A}\in \mathcal{F}$ and every $\mathbf{B}\in \mathbf{Alg}(\Sigma)$, if $\mathbf{B}\cong \mathbf{A}$, the $\mathbf{B}\in \mathcal{F}$, i.e., $\mathcal{F}$ is abstract.

\item $\mathrm{H}(\mathcal{F})\subseteq \mathcal{F}$, i.e., $\mathcal{F}$ is closed under the formation of homomorphic images of members of $\mathcal{F}$.

\item For every $\Sigma$-algebra $\mathbf{A}$ and every $\Phi$, $\Psi\in \mathrm{Cgr}(\mathbf{A})$, if $\mathbf{A}/\Phi$ and $\mathbf{A}/\Psi\in \mathcal{F}$, then $\mathbf{A}/\Phi\cap \Psi\in \mathcal{F}$.
\end{enumerate}
\end{definition}

\begin{proposition}
Let $\mathcal{F}$ be a formation of $\Sigma$-algebras. Then, for every $\Sigma$-algebra $\mathbf{A}$ and every $\Phi$, $\Psi\in \mathrm{Cgr}(\mathbf{A})$, if $\mathbf{A}/\Phi$ and $\mathbf{A}/\Psi\in \mathcal{F}$, then we have that  $\mathbf{A}/\Phi\cap \Psi\in \mathcal{F}$.
\end{proposition}

\begin{proof}
Let $\mathbf{A}$ be a $\Sigma$-algebra and let $\Phi$ and $\Psi$ be congruences on $\mathbf{A}$ such that $\mathbf{A}/\Phi$ and $\mathbf{A}/\Psi\in \mathcal{F}$. Then there exists a unique homomorphism $\mathrm{p}^{\Phi\cap\Psi,\Phi}$ from $\mathbf{A}/\Phi\cap\Psi$ to $\mathbf{A}/\Phi$ such that $\mathrm{p}^{\Phi\cap\Psi,\Phi}\circ
\mathrm{pr}^{\Phi\cap\Psi} = \mathrm{pr}^{\Phi}$. In the same way, there exists a unique homomorphism $\mathrm{p}^{\Phi\cap\Psi,\Psi}$ from $\mathbf{A}/\Phi\cap\Psi$ to $\mathbf{A}/\Psi$ such that $\mathrm{p}^{\Phi\cap\Psi,\Psi}\circ \mathrm{pr}^{\Phi\cap\Psi} = \mathrm{pr}^{\Psi}$. Then, by the universal property of the product, there exists a unique homomorphism $\langle\mathrm{p}^{\Phi\cap\Psi,\Phi},\mathrm{p}^{\Phi\cap\Psi,\Psi}\rangle$ from $\mathbf{A}/\Phi\cap\Psi$ to $\mathbf{A}/\Phi\times \mathbf{A}/\Psi$ such that
$$
\mathrm{pr}^{\Phi}\circ \langle\mathrm{p}^{\Phi\cap\Psi,\Phi},\mathrm{p}^{\Phi\cap\Psi,\Psi}\rangle = \mathrm{p}^{\Phi\cap\Psi,\Phi} \text{ and } \mathrm{pr}^{\Psi}\circ \langle\mathrm{p}^{\Phi\cap\Psi,\Phi},\mathrm{p}^{\Phi\cap\Psi,\Psi}\rangle = \mathrm{p}^{\Phi\cap\Psi,\Psi}.
$$
Moreover, $\langle\mathrm{p}^{\Phi\cap\Psi,\Phi},\mathrm{p}^{\Phi\cap\Psi,\Psi}\rangle$ is an embedding and the homomorphisms $\mathrm{p}^{\Phi\cap\Psi,\Phi}$ and $\mathrm{p}^{\Phi\cap\Psi,\Psi}$ are surjective. Therefore $\langle\mathrm{p}^{\Phi\cap\Psi,\Phi},\mathrm{p}^{\Phi\cap\Psi,\Psi}\rangle$ is a subdirect embedding of $\mathbf{A}/\Phi\cap\Psi$ in $\mathbf{A}/\Phi\times \mathbf{A}/\Psi$ and, consequently, $\mathbf{A}/\Phi\cap \Psi\in \mathcal{F}$.
\end{proof}

\begin{proposition}
Let $\mathcal{F}$ be an ShSk-formation of $\Sigma$-algebras. Then, for every $\Sigma$-algebra $\mathbf{A}$, every $\mathbf{B}$, $\mathbf{C}\in \mathcal{F}$, and every subdirect embedding $f$ of $\mathbf{A}$ in $\mathbf{B}\times \mathbf{C}$, we have that $\mathbf{A}\in \mathcal{F}$.
\end{proposition}

\begin{proof}
Let $I$ be a set in $\boldsymbol{\mathcal{U}}$. We know that if $f\colon\mathbf{A}\mor \prod_{i\in I}\mathbf{A}^{i}$ is a subdirect embedding, then if, for every $i\in I$, we denote by $\Phi^{i}$ the congruence $\mathrm{Ker}(\mathrm{pr}^{\mathbf{A}^{i}}\circ f)$ on $\mathbf{A}$, and by $g$ the homomorphism from $\mathbf{A}$ to $\prod_{i\in I}\mathbf{A}/\Phi^{i}$ defined, for every $i\in I$, every $s\in S$, and every $a\in A_{s}$, as $g_{s}(a) = ([a]_{\Phi^{i}_{s}})_{i\in I}$, we have that $g$ is a subdirect embedding which, in addition, is isomorphic to the subdirect embedding $f$. Therefore, given the subdirect embedding $f$ of $\mathbf{A}$ in $\mathbf{B}\times \mathbf{C}$, we have that it is isomorphic to the subdirect embedding $g\colon \mathbf{A}\mor \mathbf{A}/\Phi\times \mathbf{A}/\Psi$, where $\Phi$ is the congruence $\mathrm{Ker}(\mathrm{pr}^{\mathbf{B}}\circ f)$ on $\mathbf{A}$, $\Psi$ the congruence $\mathrm{Ker}(\mathrm{pr}^{\mathbf{C}}\circ f)$ on $\mathbf{A}$, and $g$ the homomorphism from $\mathbf{A}$ to $\mathbf{A}/\Phi\times \mathbf{A}/\Psi$ defined, for every $i\in I$, every $s\in S$, and every $a\in A_{s}$, as $g_{s}(a) = ([a]_{\Phi_{s}},[a]_{\Psi_{s}})$. Since $\Phi\cap \Psi = \Delta_{\mathbf{A}}$ and $\mathbf{A}\cong \mathbf{A}/\Phi\cap \Psi$, we have that $\mathbf{A}\in \mathcal{F}$.
\end{proof}

\begin{corollary}
The notions of formation of $\Sigma$-algebras and of ShSk-formation of $\Sigma$-algebras are equivalent.
\end{corollary}

Since $\mathrm{Form}_{\mathrm{Alg}}(\Sigma)\subseteq\mathrm{Sub}(\mathrm{Alg}(\Sigma))$, two formations $\mathcal{F}$ and $\mathcal{G}$ of $\Sigma$-algebras can be compared in a natural way by stating that $\mathcal{F}\leq \mathcal{G}$ if, and only if, $\mathcal{F}\subseteq \mathcal{G}$. Therefore $\mathbf{Form}_{\mathrm{Alg}}(\Sigma) = (\mathrm{Form}_{\mathrm{Alg}}(\Sigma),\leq)$ is an ordered set.

We next proceed to investigate the properties of $\mathrm{Form}_{\mathrm{Alg}}(\Sigma)$.

\begin{proposition}\label{FormAlgAlgLat}
The subset $\mathrm{Form}_{\mathrm{Alg}}(\Sigma)$ of $\mathrm{Sub}(\mathrm{Alg}(\Sigma))$ is an algebraic closure system.
\end{proposition}

\begin{proof}
It is obvious that $\mathrm{Alg}(\Sigma)$, the set of all $\Sigma$-algebras, is a formation of $\Sigma$-algebras.

Let $J$ be a nonempty set in $\boldsymbol{\mathcal{U}}$ and $(\mathcal{F}_{j})_{j\in J}$ a $J$-indexed family in $\mathrm{Form}_{\mathrm{Alg}}(\Sigma)$. Then the set $\mathcal{F}$ defined as $\mathcal{F} = \bigcap_{j\in J}\mathcal{F}_{j}\in \mathrm{Form}_{\mathrm{Alg}}(\Sigma)$. We have that $\mathcal{F}\neq\varnothing$, since, for every $j\in J$, $\mathrm{Sf}(\mathbf{1})\subseteq \mathcal{F}_{j}$. Let $\mathbf{A}$ be a $\Sigma$-algebra in $\mathcal{F}$ and let $\mathbf{B}$ be a $\Sigma$-algebra such that $\mathbf{B}\cong \mathbf{A}$. Then, for every $j\in J$, $\mathbf{A}\in \mathcal{F}_{j}$, hence, for every $j\in J$, $\mathbf{B}\in \mathcal{F}_{j}$. Therefore $\mathbf{B}\in \mathcal{F}$. It is obvious that $\mathrm{H}(\mathcal{F})\subseteq \mathcal{F}$. Finally, let us prove that $\mathrm{P}_{\mathrm{fsd}}(\mathcal{F})\subseteq \mathcal{F}$. Let $n$ be a natural number, $(\mathbf{C}^{\alpha})_{\alpha\in n}$ an $n$-family of $\Sigma$-algebras in $\mathcal{F}$, $\mathbf{A}$ a $\Sigma$-algebra, and let us suppose that there exists a subdirect embedding of $\mathbf{A}$ in $\prod_{\alpha\in n}\mathbf{C}^{\alpha}$. From the definition of $\mathcal{F}$ it follows that, for every $j\in J$ and for every $\alpha\in n$, $\mathbf{C}^{\alpha}$ belongs to $\mathcal{F}_{j}$. Hence, for every $j\in J$, $\mathbf{A}\in \mathcal{F}_{j}$. Therefore $\mathbf{A}\in \mathcal{F}$. This proves that $\mathcal{F}$ is a formation of $\Sigma$-algebras.

Let $J$ be a nonempty set in $\boldsymbol{\mathcal{U}}$ and $(\mathcal{F}_{j})_{j\in J}$ an upward directed family in $\mathrm{Form}_{\mathrm{Alg}}(\Sigma)$. Then, obviously, the set $\mathcal{F}$ defined as $\mathcal{F} = \bigcup_{j\in J}\mathcal{F}_{j}\in \mathrm{Form}_{\mathrm{Alg}}(\Sigma)$.
\end{proof}

\begin{definition}
We denote by $\mathrm{Fmg}_{\Sigma}$ the algebraic closure operator on $\mathrm{Alg}(\Sigma)$ canonically associated to the algebraic closure system $\mathrm{Form}_{\mathrm{Alg}}(\Sigma)$ and we call it the \emph{formation generating operator} for $\mathrm{Alg}(\Sigma)$.
\end{definition}

\begin{remark}
We recall that, for every $\mathcal{M}\subseteq \mathrm{Alg}(\Sigma)$, $\mathrm{Fmg}_{\Sigma}(\mathcal{M})$ is defined as follows:
$$
\textstyle
\mathrm{Fmg}_{\Sigma}(\mathcal{M}) = \bigcap\{\mathcal{F}\in \mathrm{Form}_{\mathrm{Alg}}(\Sigma)\mid \mathcal{M}\subseteq \mathcal{F}\}.
$$
Moreover, $\mathrm{Form}_{\mathrm{Alg}}(\Sigma) = \mathrm{Fix}(\mathrm{Fmg}_{\Sigma})$, the set of all fixed points of $\mathrm{Fmg}_{\Sigma}$.
\end{remark}

\begin{corollary}
$\mathbf{Form}_{\mathrm{Alg}}(\Sigma)$ is an algebraic lattice (and, for every $\mathcal{F}\in \mathrm{Form}_{\mathrm{Alg}}(\Sigma)$, $\mathcal{F}$ is compact if, and only if, there exists a finite subset $\mathcal{M}$ of $\mathrm{Alg}(\Sigma)$ such that $\mathcal{F} = \mathrm{Fmg}_{\Sigma}(\mathcal{M})$).
\end{corollary}

\begin{remark}
We recall that, for every $J\in \boldsymbol{\mathcal{U}}$ and every $J$-indexed family $(\mathcal{F}_{j})_{j\in J}$ in $\mathrm{Form}_{\mathrm{Alg}}(\Sigma)$, $\bigvee_{j\in J}\mathcal{F}_{j} = \mathrm{Fmg}_{\Sigma}(\bigcup_{j\in J}\mathcal{F}_{j})$.
\end{remark}

\begin{proposition}
Let $\mathcal{F}$ be a formation of $\Sigma$-algebras. Then the function $\mathfrak{F}_{\mathcal{F}}$ from   $\boldsymbol{\mathcal{U}}^{S}$ which assigns to $A\in \boldsymbol{\mathcal{U}}^{S}$ the subset
$$
\mathfrak{F}_{\mathcal{F}}(A)= \{\Phi\in \mathrm{Cgr}(\mathbf{T}_{\Sigma}(A))\mid \mathbf{T}_{\Sigma}(A)/\Phi\in \mathcal{F}\}
$$
of $\mathrm{Cgr}(\mathbf{T}_{\Sigma}(A))$, is a formation of congruences with respect to $\Sigma$.
\end{proposition}

\begin{proof}
Let us first prove that, for every $A\in \boldsymbol{\mathcal{U}}^{S}$, $\mathfrak{F}_{\mathcal{F}}(A)$ is a filter of the algebraic lattice $\mathbf{Cgr}(\mathbf{T}_{\Sigma}(A))$. $\mathfrak{F}_{\mathcal{F}}(A)\neq \varnothing$. In fact, $\nabla^{\mathbf{T}_{\Sigma}(A)}\in \mathfrak{F}_{\mathcal{F}}(A)$ since $\mathbf{T}_{\Sigma}(A)/\nabla^{\mathbf{T}_{\Sigma}(A)}\cong \mathbf{1}$, $\mathrm{Sf}(\mathbf{1})\subseteq \mathcal{F}$, and $\mathcal{F}$ is abstract. Let $\Phi$ and $\Psi$ be elements of $\mathfrak{F}_{\mathcal{F}}(A)$. Then, by definition of $\mathfrak{F}_{\mathcal{F}}(A)$,  $\mathbf{T}_{\Sigma}(A)/\Phi$ and $\mathbf{T}_{\Sigma}(A)/\Psi$ belong to $\mathcal{F}$. Hence $\mathbf{T}_{\Sigma}(A)/\Phi\cap\Psi\in \mathcal{F}$. Therefore $\Phi\cap\Psi\in \mathfrak{F}_{\mathcal{F}}(A)$. Let $\Phi$ be an element of $\mathfrak{F}_{\mathcal{F}}(A)$ and $\Psi$ a congruence on $\mathbf{T}_{\Sigma}(A)$ such that $\Phi\subseteq \Psi$. Then $\mathbf{T}_{\Sigma}(A)/\Psi$ is a quotient of $\mathbf{T}_{\Sigma}(A)/\Phi$. Hence $\mathbf{T}_{\Sigma}(A)/\Psi\in \mathcal{F}$. Therefore $\Psi\in \mathfrak{F}_{\mathcal{F}}(A)$. This proves that $\mathfrak{F}_{\mathcal{F}}(A)$ is a filter of the algebraic lattice $\mathbf{Cgr}(\mathbf{T}_{\Sigma}(A))$.

Let $A$ and $B$ two elements of $\boldsymbol{\mathcal{U}}^{S}$, $\Theta\in \mathfrak{F}_{\mathcal{F}}(B)$, and $f$ a homomorphism from $\mathbf{T}_{\Sigma}(A)$ to $\mathbf{T}_{\Sigma}(B)$ such that $\mathrm{pr}^{\Theta}\circ f$ is an epimorphism from $\mathbf{T}_{\Sigma}(A)$ to $\mathbf{T}_{\Sigma}(B)/\Theta$. Then $\mathbf{T}_{\Sigma}(A)/\mathrm{Ker}(\mathrm{pr}^{\Theta}\circ f)$ is isomorphic to $\mathbf{T}_{\Sigma}(B)/\Theta$. Moreover, since by hypothesis $\Theta\in \mathfrak{F}_{\mathcal{F}}(B)$, we have that $\mathbf{T}_{\Sigma}(B)/\Theta\in \mathcal{F}$. But  $\mathcal{F}$ is abstract, hence $\mathbf{T}_{\Sigma}(A)/\mathrm{Ker}(\mathrm{pr}^{\Theta}\circ f)\in \mathcal{F}$. Therefore, by definition of $\mathfrak{F}_{\mathcal{F}}(A)$, we have that $\mathrm{Ker}(\mathrm{pr}^{\Theta}\circ f)\in \mathfrak{F}_{\mathcal{F}}(A)$.
\end{proof}

Finally, we show that there exists an isomorphism between the complete lattices $\mathbf{Form}_{\mathrm{Alg}}(\Sigma)$ and $\mathbf{Form}_{\mathrm{Cgr}}(\Sigma)$, from which it follows that $\mathbf{Form}_{\mathrm{Cgr}}(\Sigma)$ is also an algebraic lattice.

\begin{proposition}\label{FormAlgFormCgrIso}
The complete lattices $\mathbf{Form}_{\mathrm{Alg}}(\Sigma)$ and $\mathbf{Form}_{\mathrm{Cgr}}(\Sigma)$ are isomorphic.
\end{proposition}

\begin{proof}
Let us first prove that, for every $\mathcal{F}\in \mathrm{Form}_{\mathrm{Alg}}(\Sigma)$, $\mathcal{F} = \mathcal{F}_{\mathfrak{F}_{\mathcal{F}}}$. By definition, $\mathfrak{F}_{\mathcal{F}}$ is such that, for every  $A\in \boldsymbol{\mathcal{U}}^{S}$,  $\mathfrak{F}_{\mathcal{F}}(A)$ is
$$
 \mathfrak{F}_{\mathcal{F}}(A)= \{\Phi\in \mathrm{Cgr}(\mathbf{T}_{\Sigma}(A))\mid \mathbf{T}_{\Sigma}(A)/\Phi\in \mathcal{F}\}.
$$
On the other hand, by definition, we have that
$$
\mathcal{F}_{\mathfrak{F}_{\mathcal{F}}} = \biggl\{ \mathbf{C}\in \mathrm{Alg}(\Sigma)\biggm|
\begin{gathered}
\exists\, A\in \boldsymbol{\mathcal{U}}^{S}\,\, \exists\, \Phi\in \mathfrak{F}_{\mathcal{F}}(A)
\\[-3pt]
(\mathbf{C}\cong \mathbf{T}_{\Sigma}(A)/\Phi)
\end{gathered}
\biggr\}.
$$
Let us prove that $\mathcal{F}\subseteq \mathcal{F}_{\mathfrak{F}_{\mathcal{F}}}$. Let $\mathbf{C}$ be a $\Sigma$-algebra in $\mathcal{F}$. Then, since every $\Sigma$-algebra is isomorphic to a quotient of a free  $\Sigma$-algebra, there exists an $A\in \boldsymbol{\mathcal{U}}^{S}$ and a congruence $\Phi$ on $\mathbf{T}_{\Sigma}(A)$ such that $\mathbf{C}\cong \mathbf{T}_{\Sigma}(A)/\Phi$. But $\mathcal{F}$ is abstract, hence $\mathbf{T}_{\Sigma}(A)/\Phi\in \mathcal{F}$. Therefore $\Phi\in \mathfrak{F}_{\mathcal{F}}(A)$ and, consequently, $\mathbf{C}\in \mathcal{F}_{\mathfrak{F}_{\mathcal{F}}}$. The proof of the converse inclusion is straightforward and the details are left to the reader. Thus we have that $\mathcal{F} = \mathcal{F}_{\mathfrak{F}_{\mathcal{F}}}$.

We next prove that, for every $\mathfrak{F}\in \mathrm{Form}_{\mathrm{Cgr}}(\Sigma)$, $\mathfrak{F} = \mathfrak{F}_{\mathcal{F}_{\mathfrak{F}}}$. By definition $\mathcal{F}_{\mathfrak{F}}$ is
$$
\mathcal{F}_{\mathfrak{F}} = \biggl\{ \mathbf{C}\in \mathrm{Alg}(\Sigma)\biggm|
\begin{gathered}
\exists\, A\in \boldsymbol{\mathcal{U}}^{S}\,\, \exists\, \Phi\in \mathfrak{F}(A)
\\[-3pt]
(\mathbf{C}\cong \mathbf{T}_{\Sigma}(A)/\Phi)
\end{gathered}
\biggr\}.
$$
On the other hand, by definition, we have that, for every  $A\in \boldsymbol{\mathcal{U}}^{S}$,  $\mathfrak{F}_{\mathcal{F}_{\mathfrak{F}}}(A)$ is
$$
 \mathfrak{F}_{\mathcal{F}_{\mathfrak{F}}}(A) = \{\Phi\in \mathrm{Cgr}(\mathbf{T}_{\Sigma}(A))\mid \mathbf{T}_{\Sigma}(A)/\Phi\in \mathcal{F}_{\mathfrak{F}}\}.
$$

Let us prove that $\mathfrak{F} \leq \mathfrak{F}_{\mathcal{F}_{\mathfrak{F}}}$. Let $A$ be an element of  $\boldsymbol{\mathcal{U}}^{S}$ and let $\Phi$ be a congruence in $\mathfrak{F}(A)$. Then, by definition of $\mathcal{F}_{\mathfrak{F}}$,  $\mathbf{T}_{\Sigma}(A)/\Phi\in \mathcal{F}_{\mathfrak{F}}$. Hence $\Phi\in \mathfrak{F}_{\mathcal{F}_{\mathfrak{F}}}(A)$. Now let us prove that $\mathfrak{F}_{\mathcal{F}_{\mathfrak{F}}}\leq \mathfrak{F}$. Let $A$ be an element of  $\boldsymbol{\mathcal{U}}^{S}$ and let $\Phi$ be a congruence in $\mathfrak{F}_{\mathcal{F}_{\mathfrak{F}}}(A)$. Then, by definition of $\mathfrak{F}_{\mathcal{F}_{\mathfrak{F}}}(A)$, $\mathbf{T}_{\Sigma}(A)/\Phi\in \mathcal{F}_{\mathfrak{F}}$. Hence, by definition of $\mathcal{F}_{\mathfrak{F}}$, there exists a $B\in \boldsymbol{\mathcal{U}}^{S}$ and a $\Psi\in \mathfrak{F}(A)$ such that $\mathbf{T}_{\Sigma}(A)/\Phi\cong \mathbf{T}_{\Sigma}(B)/\Psi$. Let $f$ be a fixed isomorphism from $\mathbf{T}_{\Sigma}(A)/\Phi$ to $\mathbf{T}_{\Sigma}(B)/\Psi$ and let us consider the following diagram:
$$
\xymatrix@C=60pt{
\mathbf{T}_{\Sigma}(A)\ar@{+>}[r]^-{\mathrm{pr}^{\Phi}} 
\ar@{+>}[rd]^*{\,\, f\circ \mathrm{pr}^{\Phi}}&
\mathbf{T}_{\Sigma}(A)/\Phi\ar @{+>}[d]^-{f}\\
\mathbf{T}_{\Sigma}(B)\ar @{+>}[r]_-{\mathrm{pr}^{\Psi}}&
\mathbf{T}_{\Sigma}(B)/\Psi
}
$$
Then, since every free $\Sigma$-algebra is projective, there exists a homomorphism $g$ from $\mathbf{T}_{\Sigma}(A)$ to $\mathbf{T}_{\Sigma}(B)$ such that $\mathrm{pr}^{\Psi}\circ g = f\circ \mathrm{pr}^{\Phi}$. Since $\mathfrak{F}$ is a $\Sigma$-congruence formation, $\mathrm{Ker}(\mathrm{pr}^{\Psi}\circ g)\in \mathfrak{F}(A)$. But $\mathrm{Ker}(\mathrm{pr}^{\Psi}\circ g) = \mathrm{Ker}(f\circ \mathrm{pr}^{\Phi})$ and $\mathrm{Ker}(f\circ \mathrm{pr}^{\Phi}) = \Phi$, consequently $\Phi\in \mathfrak{F}(A)$. Thus we have that $\mathfrak{F} = \mathfrak{F}_{\mathcal{F}_{\mathfrak{F}}}$.

Since in the category $\mathbf{Poset}$, of partially ordered sets, an isomorphism preserves all existing infima and suprema and, in addition, in the category $\mathbf{CLat}$, of complete lattices, isomorphisms coincide with order isomorphisms, to prove that the complete lattices  $\mathbf{Form}_{\mathrm{Alg}}(\Sigma)$ and $\mathbf{Form}_{\mathrm{Cgr}}(\Sigma)$ are isomorphic it suffices to verify that the bijection $\theta_{\Sigma}$ from $\mathrm{Form}_{\mathrm{Alg}}(\Sigma)$ to $\mathrm{Form}_{\mathrm{Cgr}}(\Sigma)$ which sends $\mathcal{F}$ to $\theta_{\Sigma}(\mathcal{F}) = \mathfrak{F}_{\mathcal{F}}$---with inverse the mapping $\theta_{\Sigma}^{-1}$ from $\mathrm{Form}_{\mathrm{Cgr}}(\Sigma)$ to $\mathrm{Form}_{\mathrm{Alg}}(\Sigma)$ which sends $\mathfrak{F}$ to $\theta_{\Sigma}^{-1}(\mathfrak{F}) = \mathcal{F}_{\mathfrak{F}}$---is such that both $\theta_{\Sigma}$ and $\theta_{\Sigma}^{-1}$ are order-preserving.
But this is straightforward.
\end{proof}

Taking into account that $\mathbf{ALat}$, the category of algebraic lattices, is the full subcategory of $\mathbf{CLat}$ determined by the algebraic lattices and that $\mathbf{ALat}$ is isomorphism-closed, we obtain immediately the following corollary.

\begin{corollary}
$\mathbf{Form}_{\mathrm{Cgr}}(\Sigma)$ is an algebraic lattice.
\end{corollary}


\begin{remark}
Let $J$ be a nonempty set in $\boldsymbol{\mathcal{U}}$ and $(\mathfrak{F}_{j})_{j\in J}$ an upward directed family in $\mathrm{Form}_{\mathrm{Cgr}}(\Sigma)$. Then the function $\mathfrak{F}$ defined, for every $A\in \boldsymbol{\mathcal{U}}^{S}$, as $\mathfrak{F}(A) = \bigcup_{j\in J}\mathfrak{F}_{j}(A)$ is the least upper bound of $(\mathfrak{F}_{j})_{i\in J}$ in $\mathbf{Form}_{\mathrm{Cgr}}(\Sigma)$. Moreover, since $\mathbf{Form}_{\mathrm{Cgr}}(\Sigma)$ is an algebraic lattice, it is meet-continuous, i.e., for every $\mathfrak{F}$ in $\mathrm{Form}_{\mathrm{Cgr}}(\Sigma)$, every nonempty set $J$ in $\boldsymbol{\mathcal{U}}$, and every upward directed family $(\mathfrak{F}_{j})_{j\in J}$ in $\mathrm{Form}_{\mathrm{Cgr}}(\Sigma)$ we have that
$$
\textstyle
\mathfrak{F}\wedge\bigvee_{j\in J}\mathfrak{F}_{j} = \bigvee_{j\in J} (\mathfrak{F}\wedge \mathfrak{F}_{j}).
$$
\end{remark}

\section{Elementary translations and translations.} \hfill

In this section we define, for a $\Sigma$-algebra the concepts of elementary translation and of translation with respect to it, which, in turn, will allow us to define, in the following section, the concept of congruence cogenerated by an $S$-sorted subset of the underlying $S$-sorted set of a $\Sigma$-algebra. Moreover, we investigate the relationships between the translations and the homomorphisms between $\Sigma$-algebras. To the best of our knowledge, the elementary translations and the translations were defined, for the many-sorted case, by Matthiessen in~\cite{m76} on $\mathrm{p}.\, 10$, and in \cite{m78} on $\mathrm{p}.\, 198$,

\begin{definition}
Let $\mathbf{A}$ be a $\Sigma$-algebra and $t\in S$. Then we denote by $\mathrm{Etl}_{t}(\mathbf{A})$ the subset $(\mathrm{Etl}_{t}(\mathbf{A})_{s})_{s\in S}$ of $(\mathrm{Hom}(A_{t},A_{s}))_{s\in S}$ defined, for every $s\in S$, as follows: For every mapping $T\in \mathrm{Hom}(A_{t},A_{s})$, $T\in \mathrm{Etl}_{t}(\mathbf{A})_{s}$ if and only if there is a word $w\in \fmon{S}-\{\lambda\}$, an $i\in \bb{w}$, a $\sigma\in \Sigma_{w,s}$, a family $(a_{j})_{j\in i}\in\prod_{j\in i}A_{w_{j}}$, and a family $(a_{k})_{k\in \bb{w}-(i+1)} \in\prod_{k\in \bb{w}-(i+1)}A_{w_{k}}$ such that $w_{i} = t$ and, for every $x\in A_{t}$, $T(x) =
F_{\sigma}(a_{0},\ldots,a_{i-1},x,a_{i+1},\ldots,a_{\bb{w}-1})$. We call the elements of $\mathrm{Etl}_{t}(\mathbf{A})_{s}$ the $t$-\emph{elementary translations of sort} $s$ for $\mathbf{A}$.
\end{definition}

\begin{definition}
Let $\mathbf{A}$ be a $\Sigma$-algebra and $t\in S$. Then we denote by $\mathrm{Tl}_{t}(\mathbf{A})$ the subset
$(\mathrm{Tl}_{t}(\mathbf{A})_{s})_{s\in S}$ of $(\mathrm{Hom}(A_{t},A_{s}))_{s\in S}$ defined, for every $s\in S$, as follows: For every mapping $T\in \mathrm{Hom}(A_{t},A_{s})$, $T\in \mathrm{Tl}_{t}(\mathbf{A})_{s}$ if, and only if, there is an $n\in \mathbb{N}-1$, a word $(s_{j})_{j\in n+1}\in S^{n+1}$, and a family $(T_{j})_{j\in n}$ such that $s_{0} = t$, $s_{n} = s$, $T_{0}\in \mathrm{Etl}_{t}(\mathbf{A})_{s_{1}}$, $T_{1}\in \mathrm{Etl}_{s_{1}}(\mathbf{A})_{s_{2}}$, \ldots, $T_{n-1}\in \mathrm{Etl}_{s_{n-1}}(\mathbf{A})_{s}$ and $T = T_{n-1}\comp\cdots\comp T_{0}$. We call the elements of $\mathrm{Tl}_{t}(\mathbf{A})_{s}$ the $t$-\emph{translations of sort} $s$ for $\mathbf{A}$. Besides, for every $s\in S$, the mapping  $\mathrm{id}_{A_{s}}$ will be viewed as an element of $\mathrm{Tl}_{s}(\mathbf{A})_{s}$.
\end{definition}

\begin{remark}
The $S\times S$-sorted set $(\mathrm{Tl}_{t}(\mathbf{A})_{s})_{(t,s)\in S\times S}$ determines a category $\mathbf{Tl}(\mathbf{A})$ whose objects are the sorts $s\in S$ and in which, for every $(t,s)\in S\times S$, $\mathrm{Hom}_{\mathbf{Tl}(\mathbf{A})}(t,s)$, the hom-set from $t$ to $s$, is $\mathrm{Tl}_{t}(\mathbf{A})_{s}$.
\end{remark}

Given a $\Sigma$-algebra $\mathbf{A}$ and a translation $T\in \mathrm{Tl}_{t}(\mathbf{A})_{s}$ we next define the action of $T[\cdot]$ and $T^{-1}[\cdot]$ on a subset $L\subseteq A$, as well as the actions of $T[\cdot]$ on a subset $X\subseteq A_{t}$ and of $T^{-1}[\cdot]$ on a subset $Y\subseteq A_{s}$.

\begin{definition}\label{DAntiTrans}
Let $\mathbf{A}$ be a $\Sigma$-algebra, $L\subseteq A$, $s$, $t\in S$, $X\subseteq A_{t}$, $Y\subseteq A_{s}$, and $T\in \mathrm{Tl}_{t}(\mathbf{A})_{s}$. Then
\begin{enumerate}
\item $T[L]$ denotes the subset of $A$ defined as follows: $T[L]_{s} = T[L_{t}]$ and $T[L]_{u} = \varnothing$, if $u\neq s$. Therefore, $T[L] = \delta^{s,T[L_{t}]}$.
\item $T^{-1}[L]$ denotes the subset of $A$ defined as follows: $T^{-1}[L]_{t} = T^{-1}[L_{s}]$ and $T^{-1}[L]_{u} = \varnothing$, if $u\neq t$. Therefore, $T^{-1}[L] = \delta^{t,T^{-1}[L_{s}]}$.
\item $T[X]$ denotes $T[\delta^{t,X}]$.
\item $T^{-1}[Y]$ denotes $T^{-1}[\delta^{s,Y}]$.
\end{enumerate}
\end{definition}

We next provide, by using the notions of elementary translation and of translation, two characterizations of the congruences on a $\Sigma$-algebra which will be applied afterwards, in Section 5, to prove the existence of the congruence cogenerated  by an $S$-sorted subset of the underlying $S$-sorted set of a $\Sigma$-algebra. This shows, in particular, the significance of the notions of elementary translation and of translation.

\begin{proposition}\label{CharacCong}
Let $\mathbf{A}$ be a $\Sigma$-algebra and $\Phi$ an $S$-sorted equivalence on $A$. Then the following conditions are equivalent:
\begin{enumerate}
\item $\Phi$ is a congruence on $\mathbf{A}$.
\item $\Phi$ is a closed under the elementary translations on
$\mathbf{A}$, i.e., for every every $t$, $s\in S$, every $x$, $y\in A_{t}$, and every $T\in \mathrm{Etl}_{t}(\mathbf{A})_{s}$, if $(x,y)\in \Phi_{t}$, then $(T(x),T(y))\in \Phi_{s}$.
\item $\Phi$ is a closed under the translations on $\mathbf{A}$, i.e., for every every $t$, $s\in S$, every $x$, $y\in A_{t}$, and every $T\in \mathrm{Tl}_{t}(\mathbf{A})_{s}$, if $(x,y)\in \Phi_{t}$, then $(T(x),T(y))\in \Phi_{s}$.
\end{enumerate}
\end{proposition}

\begin{proof}
Let us first prove that (1) and (2) are equivalent.

Let us suppose that $\Phi$ is a congruence on $\mathbf{A}$. We want to show that $\Phi$ is closed under the elementary translations on $\mathbf{A}$. Let $t$ and $s$ be elements of $S$ and $T$ a $t$-elementary translation of sort $s$ for $\mathbf{A}$. Then $T\colon A_{t}\mor A_{s}$ and there is a word $w\in \fmon{S}-\{\lambda\}$, an $i\in \bb{w}$, a $\sigma\in
\Sigma_{w,s}$, a family $(a_{j})_{j\in i}\in\prod_{j\in i}A_{w_{j}}$, and a family $(a_{k})_{k\in \bb{w}-(i+1)}
\in\prod_{k\in \bb{w}-(i+1)}A_{w_{k}}$ such that $w_{i} = t$ and, for every $z\in A_{t}$, $T(z) =
F_{\sigma}(a_{0},\ldots,a_{i-1},z,a_{i+1},\ldots,a_{\bb{w}-1})$. Let $x$ and $y$ be elements of $A_{t}$ such that $(x,y)\in \Phi_{t}$. Since, for every $j\in i$, $(a_{j},a_{j})\in \Phi_{w_{j}}$, for every $k\in \bb{w}-(i+1)$, $(a_{k},a_{k})\in \Phi_{w_{k}}$, and, in addition, $(x,y)\in \Phi_{t} = \Phi_{w_{i}}$, then $(T(x),T(y))\in \Phi_{s}$.

Reciprocally, let us suppose that, for every $t$, $s\in S$, every $x$, $y\in A_{t}$, and every $T\in \mathrm{Etl}_{t}(\mathbf{A})_{s}$, if $(x,y)\in \Phi_{t}$, then $(T(x),T(y))\in \Phi_{s}$. We want to show that  $\Phi$ is a congruence on $\mathbf{A}$. Let $(w,u)\in (S^{\star}-\{\lambda\})\times S$, $\sigma\colon w\mor u$,
and $a = (a_{i})_{i\in\bb{w}}$, $b = (b_{i})_{i\in\bb{w}}\in A_{w}$ such that, for every $i\in \bb{w}$ we have that $(a_{i}, b_{i})\in \Phi_{w_{i}}$. We now define, for every $i\in\bb{w}$, $T_{i}$, the $w_{i}$-elementary translations of sort $u$ for $\mathbf{A}$, as the mapping from $A_{w_{i}}$ to $A_{u}$ which sends $x\in A_{w_{i}}$ to $F_{\sigma}(b_{0},\ldots,b_{i-1},x,a_{i+1},\ldots,a_{\bb{w}-1})$ in $A_{u}$. Then $F_{\sigma}(a_{0},\ldots,a_{\bb{w}-1}) = T_{0}(a_{0})$ and $(T_{0}(a_{0}),T_{0}(b_{0}))\in \Phi_{w_{0}}$. But $T_{0}(b_{0}) = T_{1}(a_{1})$ and $(T_{1}(a_{1}),T_{1}(b_{1}))\in \Phi_{w_{1}}$. By proceeding in the same way we, finally, come to $T_{\bb{w}-2}(b_{\bb{w}-2}) = T_{\bb{w}-1}(a_{\bb{w}-1})$, $(T_{\bb{w}-1}(a_{\bb{w}-1}), T_{\bb{w}-1}(b_{\bb{w}-1}))\in \Phi_{w_{\bb{w}-1}}$, and $T_{\bb{w}-1}(b_{\bb{w}-1}) = F_{\sigma}(b_{0},\ldots,b_{\bb{w}-1})$. Therefore $(F_{\sigma}(a), F_{\sigma}(b))\in \Phi_{u}$.

We shall now proceed to verify that (2) and (3) are equivalent.

Since every elementary translations on $\mathbf{A}$ is a translation on $\mathbf{A}$, it is obvious that if $\Phi$ is closed under the translations on $\mathbf{A}$, then $\Phi$ is closed under the elementary translations on $\mathbf{A}$.

Reciprocally, let us suppose that $\Phi$ is closed under the elementary translations on $\mathbf{A}$. We want to show that $\Phi$ is closed under the translations on $\mathbf{A}$. Let $t$ and $s$ be elements of $S$, $x$, $y$ elements of $A_{t}$, $T\in \mathrm{Tl}_{t}(\mathbf{A})_{s}$, and let us suppose that $(x,y)\in \Phi_{t}$. Then there is an $n\in \mathbb{N}-1$, a word $(s_{j})_{j\in n+1}\in S^{n+1}$, and a family $(T_{j})_{j\in n}$ such that $s_{0} = t$, $s_{n} = s$, $T_{0}\in \mathrm{Etl}_{t}(\mathbf{A})_{s_{1}}$, $T_{1}\in \mathrm{Etl}_{s_{1}}(\mathbf{A})_{s_{2}}$, \ldots, $T_{n-1}\in \mathrm{Etl}_{s_{n-1}}(\mathbf{A})_{s}$ and $T = T_{n-1}\comp\cdots\comp T_{0}$. Then, from $(x,y)\in \Phi_{t} = \Phi_{s_{0}}$, we infer that $(T_{0}(x),T_{0}(y))\in \Phi_{s_{1}}$. By proceeding in the same way we, finally, come to $(T_{n-1}(\ldots(T_{0}(x))\ldots),T_{n-1}(\ldots(T_{0}(y))\ldots))\in \Phi_{s} = \Phi_{s_{n}}$, i.e., to $(T(x),T(y))\in \Phi_{s}$.
\end{proof}

We next investigate the relationships between the translations and the homomorphisms between $\Sigma$-algebras.

\begin{proposition}\label{TlandHom}
Let $f\colon \mathbf{A}\mor \mathbf{B}$ be a homomorphism. Then, for every $t$, $s\in S$ and every $T\in \mathrm{Tl}_{t}(\mathbf{A})_{s}$, there exists a $T^{f}\in\mathrm{Tl}_{t}(\mathbf{B})_{s}$ such that $f_{s}\circ T = T^{f}\circ f_{t}$. Moreover, if $f$ is an epimorphism, then, for every $t$, $s\in S$ and every $U\in\mathrm{Tl}_{t}(\mathbf{B})_{s}$, there exists a $T\in \mathrm{Tl}_{t}(\mathbf{A})_{s}$ such that $T^{f} = U$.
\end{proposition}

\begin{proof}
If $T$ is a $t$-elementary translation of sort $s$ for $\mathbf{A}$, for some $t$ and $s$ in $S$, then there is a word $w\in \fmon{S}-\{\lambda\}$, an $i\in \bb{w}$, a $\sigma\in \Sigma_{w,s}$, a family $(a_{j})_{j\in i}\in\prod_{j\in i}A_{w_{j}}$, and a family $(a_{k})_{k\in \bb{w}-(i+1)}\in\prod_{k\in \bb{w}-(i+1)}A_{w_{k}}$ such that $w_{i} = t$ and, for every $x\in A_{t}$, $T(x) = F_{\sigma}(a_{0},\ldots,a_{i-1},x,a_{i+1},\ldots,a_{\bb{w}-1})$. Then it suffices to take as $T^{f}$ precisely the mapping from $B_{t}$ to $B_{s}$ defined, for every $y\in B_{t}$, as follows:
$$
T^{f}(y) = F_{\sigma}(f_{w_{0}}(a_{0}),\ldots,f_{w_{i-1}}(a_{i-1}),y,f_{w_{i+1}}(a_{i+1}),\ldots,f_{w_{\bb{w}-1}}(a_{\bb{w}-1})).
$$
If $T\in\mathrm{Tl}_{t}(\mathbf{A})_{s}$, for some $t$ and $s$ in $S$, and $T$ is not a $t$-elementary translation of sort $s$ for $\mathbf{A}$, then there is an $n\in \mathbb{N}-1$, a word $(s_{j})_{j\in n+1}\in S^{n+1}$, and a family $(T_{j})_{j\in n}$ such that $s_{0} = t$, $s_{n} = s$, $T_{0}\in \mathrm{Etl}_{t}(\mathbf{A})_{s_{1}}$, $T_{1}\in \mathrm{Etl}_{s_{1}}(\mathbf{A})_{s_{2}}$, \ldots, $T_{n-1}\in \mathrm{Etl}_{s_{n-1}}(\mathbf{A})_{s}$ and $T = T_{n-1}\comp\cdots\comp T_{0}$. Then it suffices to take as $T^{f}$ precisely the mapping from $B_{t}$ to $B_{s}$ defined as
$T^{f} = T_{n-1}^{f}\comp\cdots\comp T_{0}^{f}$. Let us notice that if $T$ is $\mathrm{id}_{{A}_{t}}$, for some $t\in S$, then it suffices to take as $T^{f}$ precisely $\mathrm{id}_{{B}_{t}}$.

We next prove that if $f$ is an epimorphism, then, for every $t$, $s\in S$ and every $U\in\mathrm{Tl}_{t}(\mathbf{B})_{s}$, there exists a $T\in \mathrm{Tl}_{t}(\mathbf{A})_{s}$ such that $T^{f} = U$. If $U$ is a $t$-elementary translation of sort $s$ for $\mathbf{B}$, for some $t$ and $s$ in $S$, then and there is a word $w\in \fmon{S}-\{\lambda\}$, an $i\in \bb{w}$, a $\sigma\in\Sigma_{w,s}$, a family $(b_{j})_{j\in i}\in\prod_{j\in i}B_{w_{j}}$, and a family $(b_{k})_{k\in \bb{w}-(i+1)}\in\prod_{k\in \bb{w}-(i+1)}B_{w_{k}}$ such that $w_{i} = t$ and, for every $y\in B_{t}$, $U(y) = F_{\sigma}(b_{0},\ldots,b_{i-1},y,b_{i+1},\ldots,b_{\bb{w}-1})$. Then, since $f$ is an epimorphism, there exists a family $(a_{j})_{j\in i}\in\prod_{j\in i}A_{w_{j}}$ and a family $(a_{k})_{k\in \bb{w}-(i+1)}\in\prod_{k\in \bb{w}-(i+1)}A_{w_{k}}$ such that, for every $j\in i$, $f_{w_{j}}(a_{j}) = b_{j}$, and, for every $k\in \bb{w}-(i+1)$, $f_{w_{k}}(a_{k}) = b_{k}$. Then, after fixing $(a_{j})_{j\in i}$ and $(a_{k})_{k\in \bb{w}-(i+1)}$, it suffices to take as $T$ precisely the mapping from $A_{t}$ to $A_{s}$ defined, for every $x\in A_{t}$, as follows:
$$
T(x) = F_{\sigma}(a_{0},\ldots,a_{i-1},x,a_{i+1},\ldots,a_{\bb{w}-1}).
$$

If $U\in\mathrm{Tl}_{t}(\mathbf{B})_{s}$, for some $t$ and $s$ in $S$, and $U$ is not a $t$-elementary translation of sort $s$ for $\mathbf{B}$, then there is an $n\in \mathbb{N}-1$, a word $(s_{j})_{j\in n+1}\in S^{n+1}$, and a family $(U_{j})_{j\in n}$ such that $s_{0} = t$, $s_{n} = s$, $U_{0}\in \mathrm{Etl}_{t}(\mathbf{A})_{s_{1}}$, $U_{1}\in \mathrm{Etl}_{s_{1}}(\mathbf{A})_{s_{2}}$, \ldots, $U_{n-1}\in \mathrm{Etl}_{s_{n-1}}(\mathbf{A})_{s}$ and $U = U_{n-1}\comp\cdots\comp U_{0}$. Then, after choosing, for every $i\in n$, a $T_{i}$ such that $T_{i}^{f} = U_{i}$, it suffices to take as $T$ precisely the mapping from $A_{t}$ to $A_{s}$ defined as $T = T_{n-1}\comp\cdots\comp T_{0}$.
\end{proof}

\section{Congruence cogenerated  by an $S$-sorted subset of the underlying $S$-sorted set of a $\Sigma$-algebra.}\hfill

In this section, for a $\Sigma$-algebra $\mathbf{A}$, we define a mapping $\Omega^{\mathbf{A}}$ from $\mathrm{Sub}(A)$ to $\mathrm{Cgr}(\mathbf{A})$ which assigns to a subset $L$ of $A$ the so-called congruence cogenerated by $L$, and investigate its properties. In particular, we provide a description of the equivalence classes of $A/\Omega^{\mathbf{A}}(L)$, which will be used in the final section of this article.

Let $\mathbf{A}$ be a $\Sigma$-algebra and $L\subseteq A$. Then $L$ has associated, among others, the $S$-sorted equivalence
$\mathrm{Ker}(\mathrm{ch}^{L})$ on $A$, determined by the character, $\mathrm{ch}^{L}$, of the $S$-sorted subset $L$ of the  underlying $S$-sorted set $A$ of the $\Sigma$-algebra $\mathbf{A}$. Recall that $\mathrm{ch}^{L}$ is the $S$-sorted mapping from  $A$ to $(2)_{s\in S}$ whose $s$th coordinate, for $s\in S$, is
$\mathrm{ch}^{L}_{s}$, the characteristic mapping of $L_{s}$. So, for every $s\in S$, we have that:
$$
  \mathrm{Ker}(\mathrm{ch}^{L})_{s} = \{\,(x,y)\in A^{2}_{s}\mid x\in
  L_{s}\longleftrightarrow y\in L_{s}\,\}.
$$

In what follows we will prove that there exists an $S$-congruence $\Omega^{\mathbf{A}}(L)$ on $\mathbf{A}$, the $S$-congruence cogenerated by the $S$-sorted equivalence $\mathrm{Ker}(\mathrm{ch}^{L})$, which saturates $L$, i.e., which is such that  $\Omega^{\mathbf{A}}(L)\subseteq \mathrm{Ker}(\mathrm{ch}^{L})$, and that it is, in addition, the largest $S$-congruence on $\mathbf{A}$ which has such a property.

In the theory of formal languages a congruence of the type $\Omega^{\mathbf{A}}(L)$ is called the syntactic congruence determined by $L$, and they were defined by Sch\"{u}tzenberger (in~\cite{Sch55} on p.~10) for monoids (he speaks of: ``demi-groupes contenant un \'{e}l\'{e}ment neutre''). Let us add that in~\cite{sl74} on pp.~32--33, S{\l}omi\'{n}ski proved, among other results, that, for a single-sorted algebra $\mathbf{A}$ and for an equivalence relation $\Phi$ on $A$, there exists the greatest congruence on $\mathbf{A}$ contained in $\Phi$ (this is also valid for the many-sorted case).

\begin{definition}
Let $\mathbf{A}$ be a $\Sigma$-algebra and $L\subseteq A$. Then we denote by $\Omega^{\mathbf{A}}(L)$ the binary relation on $A$ defined, for every $t\in S$, as follows:
$$
\Omega^{\mathbf{A}}(L)_{t} = \biggl\{ (x,y)\in A^{2}_{t}\biggm|
\begin{gathered}
\forall\, s\in S\,\,\forall\, T\in \mathrm{Tl}_{t}(\mathbf{A})_{s}\,
\\[-3pt]
(T(x)\in L_{s}\leftrightarrow T(y)\in L_{s})
\end{gathered}
\biggr\}.
$$
\end{definition}

\begin{proposition}\label{CharacCogenCong}
Let $\mathbf{A}$ be a $\Sigma$-algebra and $L\subseteq A$. Then
\begin{enumerate}
\item $\Omega^{\mathbf{A}}(L)$ is a congruence on $\mathbf{A}$.

\item $\Omega^{\mathbf{A}}(L)\subseteq \mathrm{Ker}(\mathrm{ch}^{L})$.

\item For every congruence $\Phi$ on $\mathbf{A}$, if $\Phi\subseteq \mathrm{Ker}(\mathrm{ch}^{L})$, then $\Phi\subseteq \Omega^{\mathbf{A}}(L)$.
\end{enumerate}
In other words, $\Omega^{\mathbf{A}}(L)$ is the greatest congruence on $\mathbf{A}$ which saturates $L$.
\end{proposition}

\begin{proof}
To prove (1) it suffices to take into account Proposition~\ref{CharacCong}. To prove (2), given $t\in S$ and $(x,y)\in \Omega^{\mathbf{A}}(L)_{t}$, it suffices to consider $\mathrm{id}_{A_{t}}\in \mathrm{Tl}_{t}(\mathbf{A})_{t}$, to conclude that $x\in L_{t}$ if, and only if, $y\in L_{t}$, i.e., that $(x,y)\in \mathrm{Ker}(\mathrm{ch}^{L})_{t}$. We now proceed to prove (3).
Let $\Phi$ be a congruence on $\mathbf{A}$ such that $\Phi\subseteq \mathrm{Ker}(\mathrm{ch}^{L})$, i.e., such that, for every $s\in S$ and every $x$, $y\in A_{s}$, if $(x,y)\in \Phi_{s}$, then $x\in L_{s}$ if, and only if, $y\in L_{s}$. We want to show that, for every $t\in S$, $\Phi_{t}\subseteq \Omega^{\mathbf{A}}(L)_{t}$. Let $t$ be an element of $S$ and $(x,y)\in \Phi_{t}$. Then, since $\Phi$ is a congruence on $\mathbf{A}$, for every $s\in S$ and every $T\in \mathrm{Tl}_{t}(\mathbf{A})_{s}$, we have that $(T(x),T(y))\in \Phi_{s}$. Hence, by the hypothesis on $\Phi$, $T(x)\in L_{s}$ if, and only if, $T(y)\in L_{s}$. Therefore $\Phi\subseteq \Omega^{\mathbf{A}}(L)$.
\end{proof}

\begin{definition}
Let $\mathbf{A}$ be a $\Sigma$-algebra and $L\subseteq A$. Then we call $\Omega^{\mathbf{A}}(L)$ the congruence on $\mathbf{A}$ \emph{cogenerated} by $L$ (or the \emph{syntactic} congruence on $\mathbf{A}$ determined by $L$).
\end{definition}

The following picture illustrates the position of the congruence $\Omega^{\mathbf{A}}(L)$ in the lattice $\mathbf{Cgr}(\mathbf{A})$.
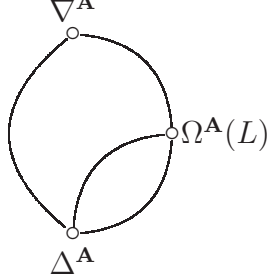
$$
\xymatrix@R=2.5pc@C=2.5pc{
{} &
*[o]{\circ} \save[]+<0pt,10pt>*{\nabla^{\mathbf{A}}}\restore
\ar@{-}@/^1.0pc/[rd]
\ar@{-}@/_2pc/[dd]
&
{}  \\
{} & {} &
*[o]{\circ} \save[]+<20pt,0pt>*{\Omega^{\mathbf{A}}(L)}\restore
\ar@{-}@/_1pc/[ld]
\ar@{-}@/^1pc/[ld]
\\
{} &
*[o]{\circ} \save[]+<0pt,-10pt>*{\Delta^{\mathbf{A}}}\restore
& {}
}
$$

\begin{remark}
Let $L$ be a subset of a semigroup (or monoid). Then the syntactic (or principal) congruence of $L$ (also called the two-sided principal congruence of $L$) falls under the notion of congruence cogenerated by $L$.
\end{remark}

\begin{proposition}\label{CharacSatCCog}
Let $\mathbf{A}$ be a $\Sigma$-algebra, $L$ a subset of $A$, and $\Phi\in\mathrm{Cgr}(\mathbf{A})$. Then $L\in \Phi\!-\!\mathrm{Sat}(A)$, i.e., $L = [L]^{\Phi}$, if, and only if,  $\Phi\subseteq\Omega^{\mathbf{A}}(L)$.
\end{proposition}

\begin{proof}
Let us suppose that $L = [L]^{\Phi}$. Then, since $\Omega^{\mathbf{A}}(L)$ is the greatest congruence on $\mathbf{A}$ such that $L = [L]^{\Omega^{\mathbf{A}}(L)}$, we have that $\Phi\subseteq\Omega^{\mathbf{A}}(L)$.

Reciprocally, let us suppose that $\Phi\subseteq\Omega^{\mathbf{A}}(L)$. Let $t$ be a sort in $S$ and $y\in [L]^{\Phi}_{t} = \bigcup_{x\in L_{t}}[x]_{\Phi_{t}}$. Then there exists an $x\in L_{t}$ such that $y\in [x]_{\Phi_{t}}$. Hence $(x,y)\in \Phi_{t}$. Therefore $(x,y)\in \Omega^{\mathbf{A}}(L)_{t}$. Hence, for every $s\in S$ and every $T\in \mathrm{Tl}_{t}(\mathbf{A})_{s}$,
$T(x)\in L_{s}$ if, and only if, $T(y)\in L_{s}$. Thus, for $s = t$ and $T = \mathrm{id}_{A_{t}}$, $x\in L_{t}$ if, and only if, $y\in L_{t}$. Consequently $y\in L_{t}$. This proves that $[L]^{\Phi}_{t}\subseteq L_{t}$. So $L\in \Phi\!-\!\mathrm{Sat}(A)$.
\end{proof}

\begin{proposition}\label{RepCongInterCCogKroneckerDelta}
Let $\mathbf{A}$ be a $\Sigma$-algebra and $\Phi\in\mathrm{Cgr}(\mathbf{A})$. Then
$$
\textstyle
\Phi = \bigcap\{\Omega^{\mathbf{A}}(\delta^{s,[a]_{\Phi_{s}}})\mid s\in S \And a\in A_{s}\}.
$$
\end{proposition}

\begin{proof}
It is straightforward to verify that, for every $s\in S$ and every $a\in A_{s}$, $\delta^{s,[a]_{\Phi_{s}}}$ is $\Phi$-saturated. Hence $\Phi \subseteq \bigcap\{\Omega^{\mathbf{A}}(\delta^{s,[a]_{\Phi_{s}}})\mid s\in S \And a\in A_{s}\}$.

Reciprocally, let $s$ be an element of $S$ and $a$, $b\in A_{s}$. If $(a,b)\nin \Phi_{s}$, then $(a,b)\nin \mathrm{Ker}(\mathrm{ch}^{\delta^{s,[a]_{\Phi_{s}}}})_{s}$. Hence $(a,b)\nin \Omega^{\mathbf{A}}(\delta^{s,[a]_{\Phi_{s}}})_{s}$. Therefore we have that $\bigcap\{\Omega^{\mathbf{A}}(\delta^{s,[a]_{\Phi_{s}}})\mid s\in S \! \And\! a\in A_{s}\}\subseteq \Phi$.
\end{proof}

\begin{remark}
Let $\mathbf{A}$ be a $\Sigma$-algebra. Then
$$
\textstyle
\Delta^{\mathbf{A}} = \bigcap\{\Omega^{\mathbf{A}}(\delta^{s,a})\mid s\in S\! \And\! a\in A_{s}\}.
$$
\end{remark}

\begin{proposition}\label{Compl}
Let $\mathbf{A}$ be a $\Sigma$-algebra and $L$ a subset of $A$. Then $\Omega^{\mathbf{A}}(L) = \Omega^{\mathbf{A}}(\complement_{A}L)$.
\end{proposition}

\begin{proposition}\label{InterCCog and CCogInter}
Let $\mathbf{A}$ be a $\Sigma$-algebra, $J$ a nonempty set in $\ensuremath{\boldsymbol{\mathcal{U}}}$, and $(L^{j})_{j\in J}$ a $J$-indexed family of subsets of $A$. Then $\bigcap_{j\in J}\Omega^{\mathbf{A}}(L^{j})\subseteq \Omega^{\mathbf{A}}(\bigcap_{j\in J}L^{j})$.
\end{proposition}

\begin{proof}
To prove that $\bigcap_{j\in J}\Omega^{\mathbf{A}}(L^{j}) \subseteq \Omega^{\mathbf{A}}(\bigcap_{j\in J}L^{j})$ it suffices to prove that $\bigcap_{j\in J}\Omega^{\mathbf{A}}(L^{j}) \subseteq \mathrm{Ker}(\mathrm{ch}^{\bigcap_{j\in J}L^{j}})$. But $\bigcap_{j\in J}\Omega^{\mathbf{A}}(L^{j}) \subseteq \bigcap_{j\in J}\mathrm{Ker}(\mathrm{ch}^{L^{j}})$ and $\bigcap_{j\in J}\mathrm{Ker}(\mathrm{ch}^{L^{j}}) \subseteq \mathrm{Ker}(\mathrm{ch}^{\bigcap_{j\in J}L^{j}})$, thus $\bigcap_{j\in J}\Omega^{\mathbf{A}}(L^{j}) \subseteq \mathrm{Ker}(\mathrm{ch}^{\bigcap_{j\in J}L^{j}})$.
\end{proof}

\begin{proposition}\label{TAntiTrans}
Let $\mathbf{A}$ be a $\Sigma$-algebra, $L$ a subset of $A$, $t$, $s\in S$, and $T\in\mathrm{Tl}_{t}(\mathbf{A})_{s}$. Then $\Omega^{\mathbf{A}}(L)\subseteq \Omega^{\mathbf{A}}(T^{-1}[L])$.
\end{proposition}

\begin{proposition}\label{TAntiHom}
Let $f$ be a homomorphism from $\mathbf{A}$ to $\mathbf{B}$ and $M$ a subset of $B$. Then
$
(f\times f)^{-1}[\Omega^{\mathbf{B}}(M)]\subseteq \Omega^{\mathbf{A}}(f^{-1}[M]).
$
Moreover, if $f$ is an epimorphism, then
$
(f\times f)^{-1}[\Omega^{\mathbf{B}}(M)] = \Omega^{\mathbf{A}}(f^{-1}[M]).
$
\end{proposition}

\begin{proposition}
Let $f$ be a homomorphism from $\mathbf{A}$ to $\mathbf{B}$, $M$ a subset of $B$, and $t\in S$. Then
$$
\textstyle
(f_{t}\times f_{t})^{-1}[\Omega^{\mathbf{B}}(M)_{t}]\subseteq \bigcap\{\Omega^{\mathbf{A}}(f^{-1}[U^{-1}[M]])_{t}\mid U\in \mathrm{Tl}_{t}(\mathbf{B})_{s}\! \And\! s\in S\}.
$$
Moreover, if $f$ is an epimorphism, then
$$
\textstyle
(f_{t}\times f_{t})^{-1}[\Omega^{\mathbf{B}}(M)_{t}] = \bigcap\{\Omega^{\mathbf{A}}(f^{-1}[U^{-1}[M]])_{t}\mid U\in \mathrm{Tl}_{t}(\mathbf{B})_{s}\! \And\! s\in S\}.
$$
\end{proposition}

\begin{remark}
For every $\Sigma$-algebra $\mathbf{A}$, $\Omega^{\mathbf{A}}$ can be considered as the component at $\mathbf{A}$ of a natural transformation $\Omega$ between two contravariant functors from a suitable category of $\Sigma$-algebras to the category $\mathbf{Set}$. In fact, let us consider the category $\mathbf{Alg}(\Sigma)_{\mathrm{epi}}$, with objects the $\Sigma$-algebras and morphisms the epimorphisms between $\Sigma$-algebras. Then we have, on the one hand, the functor $\mathrm{P}^{-}$ from $\mathbf{Alg}(\Sigma)^{\mathrm{op}}_{\mathrm{epi}}$ to $\mathbf{Set}$ which assigns to a $\Sigma$-algebra $\mathbf{A}$ the set $\mathrm{Sub}(A)$, and to an epimorphism $f\colon \mathbf{A}\mor \mathbf{B}$ the mapping $f^{-1}[\cdot]$ from $\mathrm{Sub}(B)$ to  $\mathrm{Sub}(A)$, and, on the other hand, the functor $\mathrm{Cgr}$ from $\mathbf{Alg}(\Sigma)^{\mathrm{op}}_{\mathrm{epi}}$ to $\mathbf{Set}$ which assigns to a $\Sigma$-algebra $\mathbf{A}$ the set $\mathrm{Cgr}(\mathbf{A})$, and to an epimorphism $f\colon \mathbf{A}\mor \mathbf{B}$ the mapping $(f\times f)^{-1}[\cdot]$ from $\mathrm{Cgr}(\mathbf{B})$ to $\mathrm{Cgr}(\mathbf{A})$. Then the mapping $\Omega$ from $\mathrm{Alg}(\Sigma)$ to $\mathrm{Mor}(\mathbf{Set})$ which sends a $\Sigma$-algebra $\mathbf{A}$ to the mapping $\Omega^{\mathbf{A}}$ from $\mathrm{Sub}(A)$ to $\mathrm{Cgr}(\mathbf{A})$ which assigns to a subset $L$ of $A$ precisely $\Omega^{\mathbf{A}}(L)$ is a natural transformation from $\mathrm{P}^{-}$ to $\mathrm{Cgr}$, because, for every epimorphism $f\colon \mathbf{A}\mor \mathbf{B}$, the following diagram
$$\xymatrix{
\mathrm{Sub}(A) \ar[r]^-{\Omega^{\mathbf{A}}} & \mathrm{Cgr}(\mathbf{A})\\
\mathrm{Sub}(B)\ar[u]^-{f^{-1}[\cdot]}\ar[r]_-{\Omega^{\mathbf{B}}} & \mathrm{Cgr}(\mathbf{B})\ar[u]_-{(f\times f)^{-1}[\cdot]}
   }
$$
commutes, i.e., for every $M\subseteq B$, $(f\times f)^{-1}[\Omega^{\mathbf{B}}(M)] = \Omega^{\mathbf{A}}(f^{-1}[M])$.
\end{remark}

\begin{remark}
Let $\mathbf{A}$ be a $\Sigma$-algebra. It is clear that given $L$, $L'\subseteq A$ such that $L\subseteq L'$, if $\Omega^{\mathbf{A}}(L)\subseteq \mathrm{Ker}(\mathrm{ch}^{L'})$, then $\Omega^{\mathbf{A}}(L)\subseteq \Omega^{\mathbf{A}}(L')$. However, in general, the mapping $\Omega^{\mathbf{A}}$ is not necessarily isotone, i.e., order-preserving.
\end{remark}

We next state, for a $\Sigma$-algebra $\mathbf{A}$, a necessary and sufficient condition for $\Omega^{\mathbf{A}}$ to be isotone.

\begin{proposition}
Let $\mathbf{A}$ be a $\Sigma$-algebra. Then the mapping $\Omega^{\mathbf{A}}$ is isotone if, and only if, for every nonempty set $J$ in $\ensuremath{\boldsymbol{\mathcal{U}}}$ and every $J$-indexed family of subsets $(L^{j})_{j\in J}$ of $A$, $\Omega^{\mathbf{A}}(\bigcap_{j\in J}L^{j}) = \bigcap_{j\in J}\Omega^{\mathbf{A}}(L^{j})$.
\end{proposition}

\begin{proof}
Let us suppose that for every nonempty set $J$ and every family of subsets $(L^{j})_{j\in J}$ of $A$, $\Omega^{\mathbf{A}}(\bigcap_{j\in J}L^{j}) = \bigcap_{j\in J}\Omega^{\mathbf{A}}(L^{j})$. Then given $L$, $L'\subseteq A$ such that $L\subseteq L'$, we have that $L\cap L' = L$. Therefore $\Omega^{\mathbf{A}}(L\cap L') = \Omega^{\mathbf{A}}(L)$. But $\Omega^{\mathbf{A}}(L\cap L') = \Omega^{\mathbf{A}}(L)\cap \Omega^{\mathbf{A}}(L')$. Hence $\Omega^{\mathbf{A}}(L)\subseteq \Omega^{\mathbf{A}}(L')$.

Reciprocally, let us suppose that $\Omega^{\mathbf{A}}$ is isotone. Then given a nonempty set $J$ and a family of subsets $(L^{j})_{j\in J}$ of $A$, since, for every $j\in J$, $\bigcap_{j\in J}L^{j}\subseteq L^{j}$, we have that, for every $j\in J$, $\Omega^{\mathbf{A}}(\bigcap_{j\in J}L^{j})\subseteq \Omega^{\mathbf{A}}(L^{j})$. Therefore $\Omega^{\mathbf{A}}(\bigcap_{j\in J}L^{j}) \subseteq \bigcap_{j\in J}\Omega^{\mathbf{A}}(L^{j})$. The inclusion $\bigcap_{j\in J}\Omega^{\mathbf{A}}(L^{j}) \subseteq \Omega^{\mathbf{A}}(\bigcap_{j\in J}L^{j})$ holds by Proposition~\ref{InterCCog and CCogInter}.
\end{proof}

We next provide, for a $\Sigma$-algebra $\mathbf{A}$ and a subset $L$ of $A$, a description of the equivalence classes of $A/\Omega^{\mathbf{A}}(L)$, the underlying $S$-sorted set of $\mathbf{A}/\Omega^{\mathbf{A}}(L)$. This description will be used afterwards, in the final section, to prove that two definitions of the concept of formation of regular languages with respect to $\Sigma$ are equivalent.

\begin{proposition}\label{DesClasCog}
Let $\mathbf{A}$ be a $\Sigma$-algebra, $L\subseteq A$, $t\in S$, and $a\in A_{t}$. If we denote by $\mathcal{X}_{L,t,a}$ the subset of $\mathrm{Sub}(A_{t})$ defined as follows:
$$
\mathcal{X}_{L,t,a} = \biggl\{ X\in \mathrm{Sub}(A_{t})\biggm|
\begin{gathered}
\exists\, s\in S\,\exists\, T\in \mathrm{Tl}_{t}(\mathbf{A})_{s}
\\[-3pt]
(X = T^{-1}[L_{s}] \And T(a)\in L_{s})
\end{gathered}
\biggr\},
$$
and by $\overline{\mathcal{X}}_{L,t,a}$ the subset of $\mathrm{Sub}(A_{t})$ defined as follows:
$$
\overline{\mathcal{X}}_{L,t,a} = \biggl\{ X\in \mathrm{Sub}(A_{t})\biggm|
\begin{gathered}
\exists\, s\in S\,\exists\, T\in \mathrm{Tl}_{t}(\mathbf{A})_{s}
\\[-3pt]
(X = T^{-1}[L_{s}] \And T(a)\nin L_{s})
\end{gathered}
\biggr\},
$$
then $[a]_{\Omega^{\mathbf{A}}(L)_{t}} = \bigcap\mathcal{X}_{L,t,a}-\bigcup\overline{\mathcal{X}}_{L,t,a}$.
\end{proposition}

\begin{proof}
We first prove that $[a]_{\Omega^{\mathbf{A}}(L)_{t}}\subseteq\bigcap\mathcal{X}_{L,t,a}-\bigcup\overline{\mathcal{X}}_{L,t,a}$.
Let $b$ be an element of $A_{t}$ such that $b\in [a]_{\Omega^{\mathbf{A}}(L)_{t}}$. Then $(a,b)\in \Omega^{\mathbf{A}}(L)_{t}$. Hence, for every $s\in S$ and every $T\in \mathrm{Tl}_{t}(\mathbf{A})_{s}$, we have that $T(a)\in L_{s}$ if, and only if, $T(b)\in L_{s}$. We want to show that $b\in \bigcap\mathcal{X}_{L,t,a}$ and $b\nin \bigcup\overline{\mathcal{X}}_{L,t,a}$. Let $s$ be an element of $S$ and $T$ and element of $\mathrm{Tl}_{t}(\mathbf{A})_{s}$. We want to verify that $b$ belongs to $ T^{-1}[L_{s}]$ when $T(a)\in L_{s}$. But, if $T(a)\in L_{s}$, then, by hypothesis, we have, in particular, that $T(b)\in L_{s}$, hence $b\in T^{-1}[L_{s}]$. Let $s$ be an element of $S$ and $T$ and element of $\mathrm{Tl}_{t}(\mathbf{A})_{s}$. We want to verify that $b$ does not belongs to $ T^{-1}[L_{s}]$ when $T(a)\nin L_{s}$. But, if $T(a)\nin L_{s}$, then, by hypothesis, we have, in particular, that $T(b)\nin L_{s}$, hence $b\nin T^{-1}[L_{s}]$.

We next prove that $\bigcap\mathcal{X}_{L,t,a}-\bigcup\overline{\mathcal{X}}_{L,t,a}\subseteq [a]_{\Omega^{\mathbf{A}}(L)_{t}}$.
Let $b$ be an element of $A_{t}$ such that $b\in \bigcap\mathcal{X}_{L,t,a}-\bigcup\overline{\mathcal{X}}_{L,t,a}$. Then $b\in \bigcap\mathcal{X}_{L,t,a}$ \emph{and} $b\nin \bigcup\overline{\mathcal{X}}_{L,t,a}$. Let $s$ be an element of $S$, $T$ and element of $\mathrm{Tl}_{t}(\mathbf{A})_{s}$, and let us suppose that $T(a)\in L_{s}$. Then $b\in T^{-1}[L_{s}]$. Hence $T(b)\in L_{s}$. Let $s$ be an element of $S$, $T$ and element of $\mathrm{Tl}_{t}(\mathbf{A})_{s}$, and let us suppose that $T(a)\nin L_{s}$. Then $b\nin T^{-1}[L_{s}]$. Hence $T(b)\nin L_{s}$. Thus we have that if $T(a)\in L_{s}$, then $T(b)\in L_{s}$ \emph{and} if $T(a)\nin L_{s}$, then $T(b)\nin L_{s}$. Therefore $T(a)\in L_{s}$ if, and only if, $T(b)\in L_{s}$.
\end{proof}

\section{$\Sigma$-finite index congruence formation, $\Sigma$-regular language formation, and an Eilenberg type theorem for them.}

In this section, we define, for an $S$-sorted signature $\Sigma$ and under a suitable condition on the free $\Sigma$-algebras, the concepts of formation of finite index congruences with respect to $\Sigma$, of formation of finite $\Sigma$-algebras, of formation of regular languages with respect to $\Sigma$, and of $\mathrm{BPS}$-formation of regular languages  with respect to $\Sigma$, which is a generalization to the many-sorted case of that proposed in~\cite{bps14} on p.~1748, and of which we prove that is equivalent to that of formation of regular languages with respect to $\Sigma$. Moreover, we investigate the properties of the aforementioned formations and prove that the algebraic lattices of all $\Sigma$-finite index congruence formations, of all $\Sigma$-finite algebra formations, and of all $\Sigma$-regular language formations are isomorphic.

In the remainder of this section, following a strongly rooted tradition in the fields of formal languages and automata, we agree to call languages the subsets of $\mathrm{T}_{\Sigma}(A)$, the underlying $S$-sorted set of a free $\Sigma$-algebra $\mathbf{T}_{\Sigma}(A)$ on an $S$-sorted set $A$.

\begin{proposition}\label{Cong2LangBasic}
Let $\mathfrak{F}$ be a formation of congruences with respect to $\Sigma$, then the function $\mathcal{L}_{\mathfrak{F}}$ from   $\ensuremath{\boldsymbol{\mathcal{U}}}^{S}$ which assigns to $A\in \ensuremath{\boldsymbol{\mathcal{U}}}^{S}$ the subset
\begin{align}
\mathcal{L}_{\mathfrak{F}}(A) &= \{L\in \mathrm{Sub}(\mathrm{T}_{\Sigma}(A))\mid \exists\, \Phi\in\mathfrak{F}(A)\,(L = [L]^{\Phi})\}\notag \\
 &= \{L\in \mathrm{Sub}(\mathrm{T}_{\Sigma}(A))\mid \Omega^{\mathbf{T}_{\Sigma}(A)}(L)\in \mathfrak{F}(A)\},\notag
\end{align}
of $\mathrm{Sub}(\mathrm{T}_{\Sigma}(A))$ has the following properties:
\begin{enumerate}
\item For every $A\in \ensuremath{\boldsymbol{\mathcal{U}}}^{S}$,  $\nabla^{\mathbf{T}_{\Sigma}(A)}\!-\!\mathrm{Sat}(\mathrm{T}_{\Sigma}(A))\subseteq \mathcal{L}_{\mathfrak{F}}(A)$. In particular, $\varnothing^{S}$ and $\mathrm{T}_{\Sigma}(A)$ are languages in $\mathcal{L}_{\mathfrak{F}}(A)$.

\item For every $A\in \ensuremath{\boldsymbol{\mathcal{U}}}^{S}$ and every languages $L$ and $L'$ in $\mathcal{L}_{\mathfrak{F}}(A)$,
     $$
     (\Omega^{\mathbf{T}_{\Sigma}(A)}(L)\cap \Omega^{\mathbf{T}_{\Sigma}(A)}(L'))\!-\!\mathrm{Sat}(\mathrm{T}_{\Sigma}(A))\subseteq \mathcal{L}_{\mathfrak{F}}(A).
     $$

\item For every $B\in \ensuremath{\boldsymbol{\mathcal{U}}}^{S}$, every language $M\in \mathcal{L}_{\mathfrak{F}}(B)$, and every homomorphism $f\colon \mathbf{T}_{\Sigma}(A)\mor \mathbf{T}_{\Sigma}(B)$, if \[\mathrm{pr}^{\Omega^{\mathbf{T}_{\Sigma}(B)}(M)}\circ f\colon \mathbf{T}_{\Sigma}(A)\mor \mathbf{T}_{\Sigma}(B)/\Omega^{\mathbf{T}_{\Sigma}(B)}(M)\] is an epimorphism, then $\mathrm{Ker}(\mathrm{pr}^{\Omega^{\mathbf{T}_{\Sigma}(B)}(M)}\circ f)\!-\!\mathrm{Sat}(\mathrm{T}_{\Sigma}(A))\subseteq \mathcal{L}_{\mathfrak{F}}(A)$.
\end{enumerate}
\end{proposition}

\begin{proof}
That $\nabla^{\mathbf{T}_{\Sigma}(A)}\!-\!\mathrm{Sat}(\mathrm{T}_{\Sigma}(A))\subseteq \mathcal{L}_{\mathfrak{F}}(A)$ follows from the fact that the congruence $\nabla^{\mathbf{T}_{\Sigma}(A)}$ belongs to $\mathfrak{F}(A)$.

Let $L$ and $L'$  be languages in  $\mathcal{L}_{\mathfrak{F}}(A)$. Then $\Omega^{\mathbf{T}_{\Sigma}(A)}(L)$ and  $\Omega^{\mathbf{T}_{\Sigma}(A)}(L')$ are congruences in $\mathfrak{F}(A)$. Thus $\Omega^{\mathbf{T}_{\Sigma}(A)}(L)\cap\Omega^{\mathbf{T}_{\Sigma}(A)}(L')$ is a congruence in $\mathfrak{F}(A)$. Therefore
$(\Omega^{\mathbf{T}_{\Sigma}(A)}(L)\cap \Omega^{\mathbf{T}_{\Sigma}(A)}(L'))\!-\!\mathrm{Sat}(\mathrm{T}_{\Sigma}(A))\subseteq \mathcal{L}_{\mathfrak{F}}(A)$.

For every $B\in \boldsymbol{\mathcal{U}}^{S}$, every language $M\in \mathcal{L}_{\mathfrak{F}}(B)$, and every homomorphism $f\colon \mathbf{T}_{\Sigma}(A)\mor \mathbf{T}_{\Sigma}(B)$, if \[\mathrm{pr}^{\Omega^{\mathbf{T}_{\Sigma}(B)}(M)}\circ f\colon \mathbf{T}_{\Sigma}(A)\mor \mathbf{T}_{\Sigma}(B)/\Omega^{\mathbf{T}_{\Sigma}(B)}(M)\] is an epimorphism, then $\mathrm{Ker}(\mathrm{pr}^{\Omega^{\mathbf{T}_{\Sigma}(B)}(M)}\circ f)$ is a congruence in $\mathfrak{F}(A)$. Therefore $\mathrm{Ker}(\mathrm{pr}^{\Omega^{\mathbf{T}_{\Sigma}(B)}(M)}\circ f)\!-\!\mathrm{Sat}(\mathrm{T}_{\Sigma}(A))\subseteq \mathcal{L}_{\mathfrak{F}}(A)$.
\end{proof}

\begin{remark}
If $\mathfrak{F}$ is a formation of congruences with respect to $\Sigma$, then, for every $A\in \ensuremath{\boldsymbol{\mathcal{U}}}^{S}$, $\mathcal{L}_{\mathfrak{F}}(A) = \bigcup_{\Phi\in \mathfrak{F}(A)}\Phi\!-\!\mathrm{Sat}(\mathrm{T}_{\Sigma}(A)) = \bigcup_{\Phi\in \mathfrak{F}(A)}\mathrm{Fix}([\cdot]^{\Phi})$.
\end{remark}

\begin{remark}
Given an $S$-sorted set $A$, an $L\in \mathcal{L}_{\mathfrak{F}}(A)$, and a $\Psi\in \mathrm{Cgr}(\mathbf{T}_{\Sigma}(A))$, if $\Omega^{\mathbf{T}_{\Sigma}(A)}(L)\subseteq \Psi$, then $[L]^{\Psi}\in \mathcal{L}_{\mathfrak{F}}(A)$. In fact, from $L\in \mathcal{L}_{\mathfrak{F}}(A)$ we infer that $L = [L]^{\Phi}$, for some $\Phi\in \mathfrak{F}(A)$. But $\Phi\subseteq \Omega^{\mathbf{T}_{\Sigma}(A)}(L)$, because $\Omega^{\mathbf{T}_{\Sigma}(A)}(L)$ is the greatest congruence on $\mathbf{T}_{\Sigma}(A)$ which saturates $L$. Hence $\Phi\subseteq \Psi$. Therefore, by Proposition~\ref{PropIncSat}, $[[L]^{\Psi}]^{\Phi} = [L]^{\Psi}$. So $[L]^{\Psi}\in \mathcal{L}_{\mathfrak{F}}(A)$.
\end{remark}

\begin{corollary}
For every $L\in \mathcal{L}_{\mathfrak{F}}(A)$, $\Omega^{\mathbf{T}_{\Sigma}(A)}(L)\!-\!\mathrm{Sat}(\mathrm{T}_{\Sigma}(A))\subseteq \mathcal{L}_{\mathfrak{F}}(A)$, entails that, for every $s$, $t\in S$, every $T\in \mathrm{Tl}_{t}(\mathbf{T}_{\Sigma}(A))_{s}$, and every $L\in \mathcal{L}_{\mathfrak{F}}(A)$, $T^{-1}[L]\in\mathcal{L}_{\mathfrak{F}}(A)$, i.e., $\mathcal{L}_{\mathfrak{F}}$ is closed under the inverse image of translations.
\end{corollary}

\begin{proof}
It follows from Proposition~\ref{TAntiTrans}.
\end{proof}

\begin{corollary}\label{CorolariAtoms}
For every pair of congruences $\Phi$ and $\Psi$ on $\mathbf{T}_{\Sigma}(A)$, every sort $s$ in $S$, and every term $P\in\mathrm{T}_{\Sigma}(A)_s$. If $\delta^{s,[P]_{\Phi_s}}$ and $\delta^{s,[P]_ {\Psi_s}}$ are languages in $\mathcal{L}_{\mathfrak{F}}(A)$, then so is $\delta^{s,[P]_{(\Phi\cap\Psi)_s}}$.
\end{corollary}

\begin{proof}
Let $s$ be a sort in $S$ and $P\in\mathbf{T}_{\Sigma}(A)_s$. If $\delta^{s,[P]_{\Phi_s}}$ and $\delta^{s,[P]_ {\Psi_s}}$ are languages in $\mathcal{L}_{\mathfrak{F}}(A)$, then  $\Omega^{\mathbf{T}_{\Sigma}(A)}(\delta^{s,[P]_{\Phi_s}})\in \mathfrak{F}(A)$ and
$\Omega^{\mathbf{T}_{\Sigma}(A)}(\delta^{s,[P]_{\Psi_s}})\in \mathfrak{F}(A)$. Hence $\Omega^{\mathbf{T}_{\Sigma}(A)}(\delta^{s,[P]_{\Phi_s}})\cap \Omega^{\mathbf{T}_{\Sigma}(A)}(\delta^{s,[P]_{\Psi_s}})\in \mathfrak{F}(A)$. But we have that $$\Omega^{\mathbf{T}_{\Sigma}(A)}(\delta^{s,[P]_{\Phi_s}})\cap \Omega^{\mathbf{T}_{\Sigma}(A)}(\delta^{s,[P]_{\Psi_s}})\subseteq \Omega^{\mathbf{T}_{\Sigma}(A)}(\delta^{s,[P]_{\Phi_s}}\cap \delta^{s,[P]_{\Psi_s}}).$$ Hence, $\Omega^{\mathbf{T}_{\Sigma}(A)}(\delta^{s,[P]_{\Phi_s}}\cap \delta^{s,[P]_{\Psi_s}})\in \mathfrak{F}(A)$. On the other hand, we have that $$\delta^{s,[P]_{\Phi_s}}\cap \delta^{s,[P]_{\Psi_s}}=\delta^{s,[P]_{\Phi_s\cap\Psi_s}}=\delta^{s,[P]_{(\Phi\cap\Psi)_s}}.$$ Thus $\Omega^{\mathbf{T}_{\Sigma}(A)}(\delta^{s,[P]_{\Phi_s}}\cap \delta^{s,[P]_{\Psi_s}}) = \Omega^{\mathbf{T}_{\Sigma}(A)}(\delta^{s,[P]_{(\Phi\cap\Psi)_s}})$. Therefore, $\delta^{s,[P]_{(\Phi\cap\Psi)_s}}\in \mathfrak{F}(A)$.
\end{proof}

\begin{corollary}\label{CorolariAlgebraBooleana}
If $L, L'\in \mathcal{L}_{\mathfrak{F}}(A)$, then $L\cup L'$, $L\cap L'$ and $\complement_{\mathbf{T}_{\Sigma}(A)}L$ are in $\mathcal{L}_{\mathfrak{F}}(A)$ i.e., $\mathcal{L}_{\mathfrak{F}}(A)$ is a Boolean subalgebra of $\mathbf{Sub}(\mathrm{T}_{\Sigma}(A))$, the Boolean algebra of all subsets of the underlying $S$-sorted set of the free $\Sigma$-algebra $\mathbf{T}_{\Sigma}(A)$ on $A$.
\end{corollary}

\begin{proof}
Let $L$ and $L'$ be languages in $\mathcal{L}_{\mathfrak{F}}(A)$, then $L$ is $\Phi$-saturated and $L'$ is $\Psi$-saturated for some congruences $\Phi$ and $\Psi$ in $\mathfrak{F}(A)$. Since $\mathfrak{F}$ is a formation of congruences $\Phi\cap\Psi\in \mathfrak{F}(A)$. We conclude, using Corollary~\ref{IncSat}, that $L$ and $L'$ are  $\Phi\cap\Psi$-saturated and, by Proposition~\ref{CABA Saturades}, that $L\cup L'$, $L\cap L'$ and $\complement_{\mathbf{T}_{\Sigma}(A)}L$ are in $\mathcal{L}_{\mathfrak{F}}(A)$.
\end{proof}

\begin{corollary}\label{CorolariAntiimatge}
Let $A$ and $B$ be two $S$-sorted sets, $f$ an homomorphism from $\mathbf{T}_{\Sigma}(A)$ to $\mathbf{T}_{\Sigma}(B)$, and $M\in\mathcal{L}_{\mathfrak{F}}(B)$, then $f^{-1}[M]\in\mathcal{L}_{\mathfrak{F}}(A)$.
\end{corollary}
\begin{proof} It follows from Proposition~\ref{TAntiHom}.
\end{proof}

\begin{definition}
Let $\mathbf{A}$ be a $\Sigma$-algebra and $\Phi\in \mathrm{Cgr}(\mathbf{A})$. We say that $\Phi$ is of \emph{finite index} if the $S$-sorted set $A/\Phi$ is finite. We denote by $\mathrm{Cgr}_{\mathrm{fi}}(\mathbf{A})$ the set of all congruences on $\mathbf{A}$ of finite index.
\end{definition}

\begin{remark}
Let $\mathbf{A}$ be a $\Sigma$-algebra. Then $\mathrm{Cgr}_{\mathrm{fi}}(\mathbf{A})\neq\varnothing$ if, and only if, $\mathrm{supp}_{S}(\mathbf{A})$ is finite. Therefore, if $\mathrm{card}(S)<\aleph_{0}$, then, for every every $\Sigma$-algebra $\mathbf{A}$, $\mathrm{Cgr}_{\mathrm{fi}}(\mathbf{A})\neq\varnothing$. Let us notice that the category $\mathbf{Sgr}\!-\!\mathbf{Act}(\mathbf{Set})$, of left actions of semigroups on sets, which has as objects ordered quadruples $(S,A,\cdot,\lambda)$ where $S$ and $A$ are sets, $\cdot$ a binary operation on $S$, and $\lambda$ a left action of the semigroup $\mathbf{S} = (S,\cdot)$ on the set $A$, thus $\lambda\colon S\times A\mor A$ such that, for every $x$, $y\in S$ and every $a\in A$, $\lambda(x\cdot y,a) = \lambda(x,\lambda(y,x))$, and, as morphisms from $(S,A,\cdot,\lambda)$ to $(S',A',\cdot',\lambda')$ the ordered pairs $(f,g)$ where $f$ is a homomorphism from $\mathbf{S}$ to $\mathbf{S}'$ and $g$ a mapping from $A$ to $A'$ such that, for every $x\in S$ and every $a\in A$, $g(\lambda(x,a)) = \lambda'(f(x),g(a))$, satisfies the above condition. Some further examples are the following: the category $\mathbf{Mon}\!-\!\mathbf{Act}(\mathbf{Set})$, of left actions of monoids on sets, the category $\mathbf{Grp}\!-\!\mathbf{Act}(\mathbf{Set})$, of left actions of groups on sets, and the category $\mathbf{Mod} = \mathbf{CRng}\!-\!\mathbf{Act(AbGrp)}$ of (left) actions of commutative rings on abelian groups, all of which are defined in the same way as was defined the category $\mathbf{Sgr}\!-\!\mathbf{Act}(\mathbf{Set})$.
\end{remark}

\begin{proposition}
Let $\mathbf{A}$ be a $\Sigma$-algebra such that $\mathrm{supp}_{S}(\mathbf{A})$ is finite. Then $\mathrm{Cgr}_{\mathrm{fi}}(\mathbf{A})$ is a filter of the algebraic lattice $\mathbf{Cgr}(\mathbf{A})$.
\end{proposition}

\begin{proof}
It is easy to verify the following properties. (1) That $\nabla^{\mathbf{A}}\in \mathrm{Cgr}_{\mathrm{fi}}(\mathbf{A})$. (2) That,  for every $\Phi$ and $\Psi\in \mathrm{Cgr}_{\mathrm{fi}}(\mathbf{A})$, $\Phi\cap \Psi\in\mathrm{Cgr}_{\mathrm{fi}}(\mathbf{A})$. And (3) that, for every $\Phi\in \mathrm{Cgr}_{\mathrm{fi}}(\mathbf{A})$ and every $\Psi\in \mathrm{Cgr}(\mathbf{A})$, if $\Phi\subseteq \Psi$, then $\Psi\in \mathrm{Cgr}_{\mathrm{fi}}(\mathbf{A})$.
\end{proof}

In what follows, we assume that, for every $A\in \boldsymbol{\mathcal{U}}^{S}$, $\mathrm{supp}_{S}(\mathbf{T}_{\Sigma}(A))$ is finite.

\begin{remark}
If $S$ is finite, then, obviously, for every $A\in \boldsymbol{\mathcal{U}}^{S}$, $\mathrm{supp}_{S}(\mathbf{T}_{\Sigma}(A))$ is finite. If $S$ is infinite, then there exists an $A\in \boldsymbol{\mathcal{U}}^{S}$ such that $\mathrm{supp}_{S}(\mathbf{T}_{\Sigma}(A))$ is infinite, e.g., for $A = 1 = (1)_{s\in S}$, we have that $\mathrm{supp}_{S}(\mathbf{T}_{\Sigma}(1)) = S$, thus $\mathrm{supp}_{S}(\mathbf{T}_{\Sigma}(1))$ is infinite. Hence, if, for every $A\in \boldsymbol{\mathcal{U}}^{S}$, $\mathrm{supp}_{S}(\mathbf{T}_{\Sigma}(A))$ is finite, then $S$ is finite. Therefore, the following conditions are equivalent: (1) for every $A\in \boldsymbol{\mathcal{U}}^{S}$, $\mathrm{supp}_{S}(\mathbf{T}_{\Sigma}(A))$ is finite and (2) $S$ is  finite.
\end{remark}

\begin{definition}
Let $\mathfrak{F}$ be a formation of congruences with respect to $\Sigma$ as in Definition~\ref{DefFormCgr}. We say that $\mathfrak{F}$ is a \emph{formation of finite index congruences with respect to} $\Sigma$ if, for every $A\in \boldsymbol{\mathcal{U}}^{S}$, $\mathfrak{F}(A)\subseteq \mathrm{Cgr}_{\mathrm{fi}}(\mathbf{T}_{\Sigma}(A))$. We denote by $\mathrm{Form}_{\mathrm{Cgr}_{\mathrm{fi}}}(\Sigma)$ the set of all formations of finite index congruences with respect to $\Sigma$. Let us notice that $\mathrm{Form}_{\mathrm{Cgr}_{\mathrm{fi}}}(\Sigma)\subseteq \prod_{A\in \boldsymbol{\mathcal{U}}^{S}}\Downarrow\! \mathrm{Cgr}_{\mathrm{fi}}(\mathbf{T}_{\Sigma}(A))$, where $\Downarrow\! \mathrm{Cgr}_{\mathrm{fi}}(\mathbf{T}_{\Sigma}(A))$ is the subset of $\mathrm{Filt}(\mathbf{Cgr}(\mathbf{T}_{\Sigma}(A)))$ consisting of all filters of $\mathbf{Cgr}(\mathbf{T}_{\Sigma}(A))$ which are included in $\mathrm{Cgr}_{\mathrm{fi}}(\mathbf{T}_{\Sigma}(A))$. Therefore, a formation of finite index congruences with respect to $\Sigma$ is a choice function for $(\Downarrow\!\mathrm{Cgr}_{\mathrm{fi}}(\mathbf{T}_{\Sigma}(A)))_{A\in \boldsymbol{\mathcal{U}}^{S}}$.
\end{definition}

\begin{remark}
If $\mathfrak{F}$ is a formation of finite index congruences with respect to $\Sigma$, then, in particular, for every $A\in \boldsymbol{\mathcal{U}}^{S}$, $\mathfrak{F}(A)$ is a filter (which in its turn is included in the filter $\mathrm{Cgr}_{\mathrm{fi}}(\mathbf{T}_{\Sigma}(A))$).
\end{remark}

Since two formations of finite index congruences $\mathfrak{F}$ and $\mathfrak{G}$ with respect to $\Sigma$ can be compared in a natural way, e.g., by stating that $\mathfrak{F}\leq \mathfrak{G}$ if, and only if, for every $A\in \boldsymbol{\mathcal{U}}^{S}$, $\mathfrak{F}(A)\subseteq \mathfrak{G}(A)$, we next proceed to investigate the properties of $\mathbf{Form}_{\mathrm{Cgr}_{\mathrm{fi}}}(\Sigma) = (\mathrm{Form}_{\mathrm{Cgr}_{\mathrm{fi}}}(\Sigma),\leq)$.

\begin{proposition}
$\mathbf{Form}_{\mathrm{Cgr}_{\mathrm{fi}}}(\Sigma)$ is a complete lattice.
\end{proposition}

\begin{proof}
It is obvious that $\mathbf{Form}_{\mathrm{Cgr}_{\mathrm{fi}}}(\Sigma)$ is an ordered set. On the other hand, if we take as choice function for the family $(\mathrm{Cgr}_{\mathrm{fi}}(\mathbf{T}_{\Sigma}(A)))_{A\in \boldsymbol{\mathcal{U}}^{S}}$ the function  $\mathfrak{F}$ which associates with each $A$ in $\boldsymbol{\mathcal{U}}^{S}$ precisely
$$
\mathfrak{F}(A) = \mathrm{Cgr}_{\mathrm{fi}}(\mathbf{T}_{\Sigma}(A)),
$$
then $\mathfrak{F}$ is a formation of finite index congruences with respect to $\Sigma$ and, actually, the greatest one. The condition on the supports of the free $\Sigma$-algebras guarantees the existence of finite index congruences. Let us, finally, prove that, for every nonempty set $J$ in $\ensuremath{\boldsymbol{\mathcal{U}}}$ and every family $(\mathfrak{F}_{j})_{j\in J}$ in $\mathrm{Form}_{\mathrm{Cgr}_{\mathrm{fi}}}(\Sigma)$, there exists $\bigwedge_{j\in J}\mathfrak{F}_{j}$, the greatest lower bound of $(\mathfrak{F}_{j})_{i\in J}$ in $\mathbf{Form}_{\mathrm{Cgr}_{\mathrm{fi}}}(\Sigma)$. Let $\bigwedge_{j\in J}\mathfrak{F}_{j}$ be the function defined, for every $A\in \boldsymbol{\mathcal{U}}^{S}$, as $(\bigwedge_{j\in J}\mathfrak{F}_{j})(A) = \bigcap_{j\in J}\mathfrak{F}_{j}(A)$. Thus defined $\bigwedge_{j\in J}\mathfrak{F}_{j}\in \mathrm{Form}_{\mathrm{Cgr}_{\mathrm{fi}}}(\Sigma)$. In fact, for every $j\in J$, the set $\mathfrak{F}_{j}(A)$ contains only finite index congruences. Therefore the same happens with $(\bigwedge_{j\in J}\mathfrak{F}_{j})(A)$. Moreover, for every $j\in J$, we have that $\bigwedge_{j\in J}\mathfrak{F}_{j}\leq \mathfrak{F}_{j}$ and, for every formation of finite index congruences with respect to $\Sigma$, $\mathfrak{F}$, if, for every $j\in J$, we have that $\mathfrak{F}\leq \mathfrak{F}_{j}$, then $\mathfrak{F}\leq \bigwedge_{j\in J}\mathfrak{F}_{j}$.
From this we can assert that the ordered set $\mathbf{Form}_{\mathrm{Cgr}_{\mathrm{fi}}}(\Sigma)$ is a complete lattice.

For every set $J$ in $\ensuremath{\boldsymbol{\mathcal{U}}}$ and every family $(\mathfrak{F}_{j})_{j\in J}$ in $\mathrm{Form}_{\mathrm{Cgr}_{\mathrm{fi}}}(\Sigma)$, $\bigvee_{j\in J}\mathfrak{F}_{j}$, the least upper bound of $(\mathfrak{F}_{j})_{i\in J}$ in $\mathbf{Form}_{\mathrm{Cgr}_{\mathrm{fi}}}(\Sigma)$ is obtained in the standard way.
\end{proof}
%
%

\begin{remark}
Later on, after having defined and investigated the notion of formation of finite $\Sigma$-algebras, we will improve the above lattice-theoretic results about $\mathrm{Form}_{\mathrm{Cgr}_{\mathrm{fi}}}(\Sigma)$ by proving that it is, in fact, an algebraic lattice.
\end{remark}


We next define the notion of formation of finite $\Sigma$-algebras (recall that we are assuming that $S$ is finite or, what is equivalent, that, for every $A\in \boldsymbol{\mathcal{U}}^{S}$, $\mathrm{supp}_{S}(\mathbf{T}_{\Sigma}(A))$ is finite).

\begin{definition}
We denote by $\mathrm{Alg}_{\mathrm{f}}(\Sigma)$ the set of all finite $\Sigma$-algebras.
\end{definition}

\begin{definition}
Let $\mathcal{F}$ be a formation of $\Sigma$-algebras as in Definition~\ref{DefFormAlg}. We say that $\mathcal{F}$ is a \emph{formation of finite} $\Sigma$-\emph{algebras} if $\mathcal{F}\subseteq \mathrm{Alg}_{\mathrm{f}}(\Sigma)$. We denote by $\mathrm{Form}_{\mathrm{Alg}_{\mathrm{f}}}(\Sigma)$ the set of all formations of finite $\Sigma$-algebras.
\end{definition}

Since $\mathrm{Form}_{\mathrm{Alg}_{\mathrm{f}}}(\Sigma)\subseteq\mathrm{Alg}_{\mathrm{f}}(\Sigma)$, two formations $\mathcal{F}$ and $\mathcal{G}$ of finite $\Sigma$-algebras can be compared in a natural way by stating that $\mathcal{F}\leq \mathcal{G}$ if, and only if, $\mathcal{F}\subseteq \mathcal{G}$. Therefore $\mathbf{Form}_{\mathrm{Alg}_{\mathrm{f}}}(\Sigma) = (\mathrm{Form}_{\mathrm{Alg}_{\mathrm{f}}}(\Sigma),\leq)$ is an ordered set.

We next proceed to investigate the properties of $\mathrm{Form}_{\mathrm{Alg}_{\mathrm{f}}}(\Sigma)$.

\begin{proposition}
The subset $\mathrm{Form}_{\mathrm{Alg}_{\mathrm{f}}}(\Sigma)$ of $\mathrm{Sub}(\mathrm{Alg}_{\mathrm{f}}(\Sigma))$ is an algebraic closure system.
\end{proposition}

\begin{proof}
The proof is similar to the proof of Proposition~\ref{FormAlgAlgLat}.
\end{proof}

\begin{corollary}
$\mathbf{Form}_{\mathrm{Alg}_{\mathrm{f}}}(\Sigma)$ is an algebraic lattice.
\end{corollary}

\begin{proposition}
The complete lattices $\mathbf{Form}_{\mathrm{Alg}_{\mathrm{f}}}(\Sigma)$ and $\mathbf{Form}_{\mathrm{Cgr}_{\mathrm{fi}}}(\Sigma)$ are isomorphic.
\end{proposition}

\begin{proof}
Consider the bijections $\theta_{\Sigma}$ and $\theta_{\Sigma}^{-1}$ defined in Proposition~\ref{FormAlgFormCgrIso}. Let $\mathcal{F}$ be a formation of finite $\Sigma$-algebras. Then $\mathfrak{F}_{\mathcal{F}}$, the value of $\theta_{\Sigma}$ at $\mathcal{F}$, which is the function from $\boldsymbol{\mathcal{U}}^{S}$ which assigns to $A\in \boldsymbol{\mathcal{U}}^{S}$ the subset
$$
\mathfrak{F}_{\mathcal{F}}(A) = \{\Phi\in \mathrm{Cgr}(\mathbf{T}_{\Sigma}(A))\mid \mathbf{T}_{\Sigma}(A)/\Phi\in \mathcal{F}\}
$$
of $\mathrm{Cgr}(\mathbf{T}_{\Sigma}(A))$ is a formation of congruences with respect to $\Sigma$.
Moreover, for every $A\in \boldsymbol{\mathcal{U}}^{S}$, $\mathfrak{F}_{\mathcal{F}}(A)$ only contains finite index congruences. Therefore $\mathfrak{F}_{\mathcal{F}}(A)\subseteq \mathrm{Cgr}_{\mathrm{fi}}(\mathbf{T}_{\Sigma}(A))$.

Reciprocally, let $\mathfrak{F}$ be a formation of finite index congruences with respect to $\Sigma$. Then $\mathcal{F}_{\mathfrak{F}}$, the value of $\theta^{-1}_{\Sigma}$ at $\mathfrak{F}$, which is
$$
\mathcal{F}_{\mathfrak{F}} = \biggl\{ \mathbf{C}\in \mathrm{Alg}(\Sigma)\biggm|
\begin{gathered}
\exists\, A\in \boldsymbol{\mathcal{U}}^{S}\,\, \exists\, \Phi\in \mathfrak{F}(A)
\\[-3pt]
(\mathbf{C}\cong \mathbf{T}_{\Sigma}(A)/\Phi)
\end{gathered}
\biggr\},
$$
is a formation of $\Sigma$-algebras. Moreover, $\mathcal{F}_{\mathfrak{F}}$ contains only finite $\Sigma$-algebras. Therefore $\mathcal{F}_{\mathfrak{F}}\subseteq \mathrm{Alg}_{\mathrm{f}}(\Sigma)$.

From the above it follows that the bi-restriction of $\theta_{\Sigma}$ to $\mathbf{Form}_{\mathrm{Alg}_{\mathrm{f}}}(\Sigma)$ and $\mathbf{Form}_{\mathrm{Cgr}_{\mathrm{fi}}}(\Sigma)$ is an isomorphism between the complete lattices  $\mathbf{Form}_{\mathrm{Alg}_{\mathrm{f}}}(\Sigma)$ and $\mathbf{Form}_{\mathrm{Cgr}_{\mathrm{fi}}}(\Sigma)$.
\end{proof}

\begin{corollary}\label{FormCgrfiAlg}
$\mathbf{Form}_{\mathrm{Cgr}_{\mathrm{fi}}}(\Sigma)$ is an algebraic lattice.
\end{corollary}


\begin{remark}
Let $J$ be a nonempty set in $\ensuremath{\boldsymbol{\mathcal{U}}}$ and $(\mathfrak{F}_{j})_{j\in J}$ an upward directed family in $\mathrm{Form}_{\mathrm{Cgr}_{\mathrm{fi}}}(\Sigma)$. Then the function $\mathfrak{F}$ defined, for every $A\in \boldsymbol{\mathcal{U}}^{S}$, as $\mathfrak{F}(A) = \bigcup_{j\in J}\mathfrak{F}_{j}(A)$ is the least upper bound of $(\mathfrak{F}_{j})_{i\in J}$ in $\mathbf{Form}_{\mathrm{Cgr}_{\mathrm{fi}}}(\Sigma)$. Moreover, since $\mathbf{Form}_{\mathrm{Cgr}_{\mathrm{fi}}}(\Sigma)$ is an algebraic lattice, it is meet-continuous, i.e., for every $\mathfrak{F}$ in $\mathrm{Form}_{\mathrm{Cgr}_{\mathrm{fi}}}(\Sigma)$, every nonempty set $J$ in $\ensuremath{\boldsymbol{\mathcal{U}}}$, and every upward directed family $(\mathfrak{F}_{j})_{j\in J}$ in $\mathrm{Form}_{\mathrm{Cgr}_{\mathrm{fi}}}(\Sigma)$ we have that
$$
\textstyle
\mathfrak{F}\wedge\bigvee_{j\in J}\mathfrak{F}_{j} = \bigvee_{j\in J} (\mathfrak{F}\wedge \mathfrak{F}_{j}).
$$
\end{remark}


We next define the concepts of regular language and of formation of regular languages with respect to $\Sigma$.

\begin{definition} Let $\mathbf{A}$ be an $\Sigma$-algebra such that $\mathrm{supp}_{S}(\mathbf{A})$ is finite and $L\subseteq A$. We say that $L$ is a \emph{regular language} if $\Omega^{\mathbf{A}}(L)\in \mathrm{Cgr}_{\mathrm{fi}}(\mathbf{A})$. We denote by $\mathrm{Lang_{r}}(\mathbf{A})$ the set of all $L\subseteq A$ such that $L$ is regular.
\end{definition}

\begin{remark}
Let $\mathbf{A}$ be an $\Sigma$-algebra such that $\mathrm{supp}_{S}(\mathbf{A})$ is finite and $L\subseteq A$. Then $L$ can be regular and not finite.
\end{remark}

\begin{definition}\label{Def1FRL}
A \emph{formation of regular languages with respect to} $\Sigma$ is a function $\mathcal{L}$ from $\boldsymbol{\mathcal{U}}^{S}$  such that, for every $A\in \boldsymbol{\mathcal{U}}^{S}$, $\mathcal{L}(A)\subseteq \mathrm{Lang_{r}}(\mathbf{T}_{\Sigma}(A))$, and the following conditions are satisfied:
\begin{enumerate}

\item For every $A\in \boldsymbol{\mathcal{U}}^{S}$,   $\nabla^{\mathbf{T}_{\Sigma}(A)}\!-\!\mathrm{Sat}(\mathrm{T}_{\Sigma}(A))\subseteq \mathcal{L}(A)$.

\item For every $A\in \boldsymbol{\mathcal{U}}^{S}$ and every $L$ and $L'\in \mathcal{L}(A)$,
      $$
      (\Omega^{\mathbf{T}_{\Sigma}(A)}(L)\cap \Omega^{\mathbf{T}_{\Sigma}(A)}(L'))\!-\!\mathrm{Sat}(\mathrm{T}_{\Sigma}(A))\subseteq \mathcal{L}(A).
      $$

\item For every $B\in \boldsymbol{\mathcal{U}}^{S}$, every $M\in \mathcal{L}(B)$, and every homomorphism $f\colon \mathbf{T}_{\Sigma}(A)\mor \mathbf{T}_{\Sigma}(B)$, if \[\mathrm{pr}^{\Omega^{\mathbf{T}_{\Sigma}(B)}(M)}\circ f\colon \mathbf{T}_{\Sigma}(A)\mor \mathbf{T}_{\Sigma}(B)/\Omega^{\mathbf{T}_{\Sigma}(B)}(M)\] is an epimorphism, then $\mathrm{Ker}(\mathrm{pr}^{\Omega^{\mathbf{T}_{\Sigma}(B)}(M)}\circ f)\!-\!\mathrm{Sat}(\mathrm{T}_{\Sigma}(A))\subseteq \mathcal{L}(A)$.
\end{enumerate}

We denote by $\mathrm{Form}_{\mathrm{Lang_{r}}}(\Sigma)$ the set of all formations of regular languages with respect to $\Sigma$. Notice that $\mathrm{Form}_{\mathrm{Lang_{r}}}(\Sigma)\subseteq \prod_{A\in \boldsymbol{\mathcal{U}}^{S}} \mathrm{Sub}(\mathrm{Lang_{r}}(\mathbf{T}_{\Sigma}(A)))$. Therefore a formation of regular languages with respect to $\Sigma$ is a special type of choice function for $(\mathrm{Sub}(\mathrm{Lang_{r}}(\mathbf{T}_{\Sigma}(A))))_{A\in \boldsymbol{\mathcal{U}}^{S}}$.
\end{definition}

Two formations of regular languages with respect to $\Sigma$, $\mathcal{L}$ and $\mathcal{M}$, can be compared in a natural way, e.g., by stating that $\mathcal{L}\leq \mathcal{M}$ if, and only if, for every $A\in \boldsymbol{\mathcal{U}}^{S}$, $\mathcal{L}(A)\subseteq \mathcal{M}(A)$. We denote by $\mathbf{Form}_{\mathrm{Lang_{r}}}(\Sigma) = (\mathrm{Form}_{\mathrm{Lang_{r}}}(\Sigma),\leq)$ the corresponding ordered set.

Before proving that there exists an isomorphism between the complete lattice of all formations of finite index congruence with respect to $\Sigma$ and the complete lattice of all formations of regular languages with respect to $\Sigma$, which, ultimately, will be an isomorphism between algebraic lattices, we provide next an alternative but, as we will prove below, equivalent definition of formation of regular languages with respect to $\Sigma$ by means of, among others, translations and Boolean operations. Let us point out that this alternative definition is a generalization to the many-sorted case of that proposed in~\cite{bps14} on p.~1748. For this reason, we call the just mentioned formation of regular languages with respect to $\Sigma$ a $\mathrm{BPS}$-formation of regular languages with respect to $\Sigma$.

\begin{definition}\label{Def2FRL}
A \emph{BPS-formation of regular languages with respect to} $\Sigma$ is a function $\mathcal{L}$ from $\boldsymbol{\mathcal{U}}^{S}$ such that, for every $A\in \boldsymbol{\mathcal{U}}^{S}$, $\mathcal{L}(A)\subseteq \mathrm{Lang_{r}}(\mathbf{T}_{\Sigma}(A))$, and the following conditions are satisfied:
\begin{enumerate}

\item[(BPS 1)] For every $A\in \boldsymbol{\mathcal{U}}^{S}$,   $\nabla^{\mathbf{T}_{\Sigma}(A)}\!-\!\mathrm{Sat}(\mathrm{T}_{\Sigma}(A))\subseteq \mathcal{L}(A)$.

\item[(BPS 2)] For every $A\in \boldsymbol{\mathcal{U}}^{S}$, every $L\in \mathcal{L}(A)$, every $s$, $t\in S$, and every $T\in \mathrm{Tl}_{t}(\mathbf{T}_{\Sigma}(A))_{s}$, the language $T^{-1}[L]\in\mathcal{L}(A)$, i.e., $\mathcal{L}(A)$ is closed under the inverse image of translations.

\item[(BPS 3)] For every $A\in \boldsymbol{\mathcal{U}}^{S}$ and every $L, L'\in \mathcal{L}(A)$, $L\cup L'$, $L\cap L'$ and $\complement_{\mathbf{T}_{\Sigma}(A)}L$ are in $\mathcal{L}(A)$ i.e., $\mathcal{L}(A)$ is a Boolean subalgebra of $\mathbf{Sub}(\mathrm{T}_{\Sigma}(A))$.

\item[(BPS 4)] For every $A, B\in \boldsymbol{\mathcal{U}}^{S}$, every $M\in \mathcal{L}(B)$, and every homomorphism $f\colon \mathbf{T}_{\Sigma}(A)\mor \mathbf{T}_{\Sigma}(B)$, if \[\mathrm{pr}^{\Omega^{\mathbf{T}_{\Sigma}(B)}(M)}\circ f\colon \mathbf{T}_{\Sigma}(A)\mor \mathbf{T}_{\Sigma}(B)/\Omega^{\mathbf{T}_{\Sigma}(B)}(M)\] is an epimorphism, then $f^{-1}[M]\in\mathcal{L}(A)$.
\end{enumerate}

If we denote by $\mathrm{Form}^{\mathrm{BPS}}_{\mathrm{Lang_{r}}}(\Sigma)$ the set of all BPS-formations of regular languages with respect to $\Sigma$, then $\mathrm{Form}^{\mathrm{BPS}}_{\mathrm{Lang_{r}}}(\Sigma)\subseteq \prod_{A\in \boldsymbol{\mathcal{U}}^{S}} \mathrm{Sub}(\mathrm{Lang_{r}}(\mathbf{T}_{\Sigma}(A)))$. Therefore a BPS-formation of regular languages with respect to $\Sigma$ is a special type of choice function for the family $(\mathrm{Sub}(\mathrm{Lang_{r}}(\mathbf{T}_{\Sigma}(A))))_{A\in \boldsymbol{\mathcal{U}}^{S}}$.
\end{definition}

\begin{proposition} Definitions \ref{Def1FRL} and \ref{Def2FRL} are equivalent.
\end{proposition}
\begin{proof}
Let $\mathcal{L}$ be a formation of regular languages in the sense of Definition~\ref{Def1FRL}. Let us check that it fulfils all the conditions set out in Definition~\ref{Def2FRL}. Let $A$ be an $S$-sorted set. Then, obviously, $(\mathrm{BPS}\, 1)$ is fulfilled. Let us verify $(\mathrm{BPS}\, 2)$. Let $A$ be an $S$-sorted set, $L\in \mathcal{L}(A)$, $s$, $t\in S$, and $T\in \mathrm{Tl}_{t}(\mathbf{T}_{\Sigma}(A))_{s}$. By Corollary~\ref{TAntiTrans}, $\Omega^{\mathbf{T}_{\Sigma}(A)}(L)\subseteq \Omega^{\mathbf{T}_{\Sigma}(A)}(T^{-1}[L])$. Hence, by Corollary~\ref{IncSat}, $T^{-1}[L]$ is $\Omega^{\mathbf{T}_{\Sigma}(A)}(L)$-saturated. But, by definition, all the $\Omega^{\mathbf{T}_{\Sigma}(A)}(L)$-saturated languages belong to $\mathcal{L}(A)$. Therefore $T^{-1}[L]\in\mathcal{L}(A)$. Now let us verify $(\mathrm{BPS}\, 3)$. Let $A$ be an $S$-sorted set and $L, L'\in \mathcal{L}(A)$. Let us notice that $\Omega^{\mathbf{T}_{\Sigma}(A)}(L)\cap\Omega^{\mathbf{T}_{\Sigma}(A)}(L')$ is a congruence included in both $\Omega^{\mathbf{T}_{\Sigma}(A)}(L)$ and $\Omega^{\mathbf{T}_{\Sigma}(A)}(L')$. Hence, by Corollary~\ref{IncSat}, $L$ and $L'$ are saturated languages for the congruence $\Omega^{\mathbf{T}_{\Sigma}(A)}(L)\cap\Omega^{\mathbf{T}_{\Sigma}(A)}(L')$. By Proposition~\ref{CABA Saturades} $L\cup L'$, $L\cap L'$, and $\complement_{\mathbf{T}_{\Sigma}(A)}L$ are also $\Omega^{\mathbf{T}_{\Sigma}(A)}(L)\cap\Omega^{\mathbf{T}_{\Sigma}(A)}(L')$-saturated and, thus, languages in $\mathcal{L}(A)$. Finally, let us verify $(\mathrm{BPS}\, 4)$. Let $A$ and $B$  be $S$-sorted sets, $M\in \mathcal{L}(B)$, and $f\colon \mathbf{T}_{\Sigma}(A)\mor \mathbf{T}_{\Sigma}(B)$ a homomorphism such that $\mathrm{pr}^{\Omega^{\mathbf{T}_{\Sigma}(B)}(M)}\circ f$ is an epimorphism. Then, by definition, all languages saturated for the congruence $\mathrm{Ker}(\mathrm{pr}^{\Omega^{\mathbf{T}_{\Sigma}(B)}(M)}\circ f)$ are in $\mathcal{L}(A)$. By Proposition~\ref{TAntiHom}, $\mathrm{Ker}(\mathrm{pr}^{\Omega^{\mathbf{T}_{\Sigma}(B)}(M)}\circ f)\subseteq \Omega^{\mathbf{A}}(f^{-1}[M])$. Thus, by Corollary~\ref{IncSat}, $f^{-1}[M]$ is saturated for the congruence $\mathrm{Ker}(\mathrm{pr}^{\Omega^{\mathbf{T}_{\Sigma}(B)}(M)}\circ f)$ and we conclude that it belongs to $\mathcal{L}(A)$.

Reciprocally, let $\mathcal{L}$ be a $\mathrm{BPS}-$formation of regular languages. Let us check that it fulfils all the conditions set out in Definition~\ref{Def1FRL}. Let $A$ be an $S$-sorted set. Then, obviously, $(1)$ is fulfilled. Let $A$ be an $S$-sorted set and $L$, $L'\in \mathcal{L}(A)$. In the sequel, to simplify notation, $\Phi$ stands for the congruence $\Omega^{\mathbf{T}_{\Sigma}(A)}(L)\cap \Omega^{\mathbf{T}_{\Sigma}(A)}(L')$. We want to show that all the saturated languages for the congruence $\Phi$ are in $\mathcal{L}(A)$. Since $L$ and $L'$ are regular languages, the congruences $\Omega^{\mathbf{T}_{\Sigma}(A)}(L)$ and $\Omega^{\mathbf{T}_{\Sigma}(A)}(L')$ have finite index. Hence $\Phi$ has also finite index. In particular, $\mathrm{supp}_{S}({\mathbf{T}}_{\Sigma}(A)/\Phi)$ is finite and, for every $t\in \mathrm{supp}_{S}({\mathbf{T}}_{\Sigma}(A)/\Phi)$, ${\mathrm{T}}_{\Sigma}(A)_{t}/\Phi_{t}$ is finite. Let $K$ be a saturated language for the congruence $\Phi$. Then $K$ will also have finite support. 
To verify this last assertion it suffices to show that, for every $t\in \mathrm{supp}_{S}(K)$, the language $\delta^{t,K_t}$ is a language in $\mathcal{L}(A)$. Let $t$ be an element of $\mathrm{supp}_{S}(K)$. Since $K$ is a saturated language for $\Phi$, it follows that $K_{t} = \bigcup_{P\in K_{t}}[P]_{\Phi_{t}}$. On the other hand, since $\Phi$ has finite index, there exists a finite number of equivalence classes with respect to $\Phi_{t}$. Thus, the above union is finite and so it suffices to show that, for every $t\in \mathrm{supp}_{S}(K)$ and every $P\in K_t$, the language $\delta^{t,[P]_{\Phi_{t}}}$ belongs to $\mathcal{L}(A)$. But, since $[P]_{\Phi_{t}}=[P]_{\Omega^{\mathbf{T}_{\Sigma}(A)}(L)_t}\cap [P]_{\Omega^{\mathbf{T}_{\Sigma}(A)}(L')_t}$, we only need to show that both $\delta^{t,[P]_{\Omega^{\mathbf{T}_{\Sigma}(A)}(L)_t}}$ and $\delta^{t,[P]_{\Omega^{\mathbf{T}_{\Sigma}(A)}(L')_t}}$ are languages in $\mathcal{L}(A)$. However, by Proposition~\ref{DesClasCog}, $[P]_{\Omega^{\mathbf{T}_{\Sigma}(A)}(L)_t}$ has the following representation
\begin{equation}\label{Eq1}
\textstyle
[P]_{\Omega^{\mathbf{A}}(L)_{t}} = \bigcap\mathcal{X}_{L,t,P}-\bigcup\overline{\mathcal{X}}_{L,t,P},
\end{equation}
where  $\mathcal{X}_{L,t,P}$ denotes the subset of $\mathrm{Sub}({T}_{\Sigma}(A)_{t})$ defined as follows:
$$
\mathcal{X}_{L,t,P} = \biggl\{ X\in \mathrm{Sub}({T}_{\Sigma}(A)_{t})\biggm|
\begin{gathered}
\exists\, s\in S\,\exists\, T\in \mathrm{Tl}_{t}(\mathbf{T}_{\Sigma}(A))_{s}
\\[-3pt]
(X = T^{-1}[L_{s}] \And T(P)\in L_{s})
\end{gathered}
\biggr\},
$$
and by $\overline{\mathcal{X}}_{L,t,P}$ the subset of $\mathrm{Sub}(A_{t})$ defined as follows:
$$
\overline{\mathcal{X}}_{L,t,P} = \biggl\{ X\in \mathrm{Sub}({T}_{\Sigma}(A)_{t})\biggm|
\begin{gathered}
\exists\, s\in S\,\exists\, T\in \mathrm{Tl}_{t}(\mathbf{T}_{\Sigma}(A))_{s}
\\[-3pt]
(X = T^{-1}[L_{s}] \And T(P)\nin L_{s})
\end{gathered}
\biggr\}.
$$

Let us consider a sort $t \in S$ and an $X\in \mathcal{X}_{L,t,P}$. Then, by definition, there exists a translation $T\in \mathrm{Tl}_{t}(\mathbf{A})_{s}$ such that $X = T^{-1}[L_{s}]$. But, by Definition~\ref{DAntiTrans}, we have that $T^{-1}[L_{s}] = T^{-1}[\delta^{s,L_{s}}]$, hence $X = T^{-1}[\delta^{s,L_{s}}]$. On the other hand, $\delta^{s,L_{s}}$ can be represented as $L\cap\delta^{s,{T}_{\Sigma}(A)_{s}}$. But, by Proposition~\ref{NablaSat}, $\delta^{s,{T}_{\Sigma}(A)_{s}}$ is a $\nabla^{\mathbf{T}_{\Sigma}(A)}$-saturated language and, thus, is a language in $\mathcal{L}(A)$. Therefore, since, by hypothesis, $L$ is a language in $\mathcal{L}(A)$ and $\mathcal{L}(A)$ is closed under finite intersections, $\delta^{s,L_{s}}$ is a language in $\mathcal{L}(A)$, and, consequently, the language $\delta^{t,X} = \delta^{t,T^{-1}[\delta^{s,L_{s}}]}$ belongs to $\mathcal{L}(A)$, since $\mathcal{L}(A)$ is closed under inverse images of translations.
On the other hand, by Proposition~\ref{TAntiTrans}, for every $T\in \mathrm{Tl}_{t}(\mathbf{T}_{\Sigma}(A))_{s}$ we have that $\Omega^{\mathbf{T}_{\Sigma}(A)}(\delta^{s,L_{s}})\subseteq \Omega^{\mathbf{T}_{\Sigma}(A)}(T^{-1}[\delta^{s,L_{s}}])$, and, in addition to this, we have that $\Omega^{{\mathbf{T}}_{\Sigma}(A)}(L)\subseteq \Omega^{\mathbf{T}_{\Sigma}(A)}(\delta^{s,L_{s}})$. Hence, for every $t\in S$ and every $X\in \mathcal{X}_{L,t,P}$, the language $\delta^{t,X}$ is $\Omega^{{\mathbf{T}}_{\Sigma}(A)}(L)$-saturated. Since $\Omega^{{\mathbf{T}}_{\Sigma}(A)}(L)$ has finite index, there exists only a finite number of $\Omega^{{\mathbf{T}}_{\Sigma}(A)}(L)$-saturated languages. Therefore, the families $\mathcal{X}_{L,t,P}$ and  $\overline{\mathcal{X}}_{L,t,P}$ are finite. Hence, from equation \ref{Eq1}, we have that $\delta^{t,[P]_{\Omega^{\mathbf{T}_{\Sigma}(A)}(L)_t}}$ can be represented by using Boolean operations involving only a finite number of languages in $\mathcal{L}(A)$, and, thus, $\delta^{t,[P]_{\Omega^{\mathbf{T}_{\Sigma}(A)}(L)_t}}$ is a language in $\mathcal{L}(A)$. An analogous argument can be used to conclude that $\delta^{t,[P]_{\Omega^{\mathbf{T}_{\Sigma}(A)}(L')_t}}$ also belongs to $\mathcal{L}(A)$. Therefore all the saturated languages for the congruence $\Phi = \Omega^{\mathbf{T}_{\Sigma}(A)}(L)\cap \Omega^{\mathbf{T}_{\Sigma}(A)}(L')$ are in $\mathcal{L}(A)$. Finally, let us verify $(4)$. Let $A$ and $B$ be $S$-sorted sets, $M\in \mathcal{L}(B)$, and $f\colon \mathbf{T}_{\Sigma}(A)\mor \mathbf{T}_{\Sigma}(B)$ a homomorphism such that $\mathrm{pr}^{\Omega^{\mathbf{T}_{\Sigma}(B)}(M)}\circ f$ is an epimorphism. In the sequel, to simplify notation, $\Psi$ stands for the congruence $\mathrm{Ker}(\mathrm{pr}^{\Omega^{\mathbf{T}_{\Sigma}(B)}(M)}\circ f)$. We need to show that all languages saturated for the congruence $\Psi$  belong to $\mathcal{L}(A)$. Let us notice that the congruence $\Psi$ has also finite index and we can proceed as in the second case, that is, we only need to prove that, for every $s\in S$ and every $P\in {T}_{\Sigma}(A)_{s}$, the language $\delta^{s, [P]_{\Psi_{s}}}$ belongs to $\mathcal{L}(A)$. The statement will follow since the remaining $\Psi$-saturated languages are finite unions of these atomic languages.

Consider the language $K$ in $\mathrm{Sub}({T}_{\Sigma}(B))$ defined by
$$
K=[f[\delta^{s,[P]_{\Psi_s}}]]^{\Omega^{\mathbf{T}_{\Sigma}(B)}(M)},
$$
that is, $K$ is the $\Omega^{\mathbf{T}_{\Sigma}(B)}(M)$-saturation of the language $f[\delta^{s,[P]_{\Psi_s}}]$. Let us notice  that in the second case we have proved that, if $M$ is a language in $\mathcal{L}(B)$, then all the $\Omega^{\mathbf{T}_{\Sigma}(B)}(M)$-saturated languages are in $\mathcal{L}(B)$, thus $K$ is a language in $\mathcal{L}(B)$. It follows from Proposition~\ref{CharacSatCCog} that $\Omega^{\mathbf{T}_{\Sigma}(B)}(M)\subseteq \Omega^{\mathbf{T}_{\Sigma}(B)}(K)$ and we conclude that
$$
\mathrm{pr}^{\Omega^{\mathbf{T}_{\Sigma}(B)}(K)}\circ f\colon \mathbf{T}_{\Sigma}(A)\mor \mathbf{T}_{\Sigma}(B)/\Omega^{\mathbf{T}_{\Sigma}(B)}(K)
$$
is an epimorphism. Then $f^{-1}[K]\in\mathcal{L}(A)$. We claim that the languages $\delta^{s, [P]_{\Psi_{s}}}$ and $f^{-1}[K]$ are equal.

In fact, let $t$ be a sort in $S$. If $t\neq s$, then, on the one hand, $\delta^{s, [P]_{\Psi_{s}}}_{} t= \varnothing$. On the other hand, we have that
$$
K_t=[f_t[\delta^{s,[P]_{\Psi_s}}_{t}]]^{\Omega^{\mathbf{T}_{\Sigma}(B)}(M)_{t}} = [\varnothing]^{\Omega^{\mathbf{T}_{\Sigma}(B)}(M)_{t}} = \varnothing.
$$
It follows that $(f^{-1}[K])_{t} = f^{-1}_{t}[K_{t}] = f^{-1}_{t}[\varnothing]=\varnothing$.

Now, for the case $t=s$, we have, on the one hand, that $\delta^{s, [P]_{\Psi_{s}}}_{s} = [P]_{\Psi_{s}}$. On the other hand, $(f^{-1}[K])_{s} = f^{-1}_s[K_s]$, where
$$
K_{s} = [f_{s}[\delta^{s,[P]_{\Psi_s}}_{s}]]^{\Omega^{\mathbf{T}_{\Sigma}(B)}(M)_{s}} = [f_{s}[[P]_{\Psi_{s}}]]^{\Omega^{\mathbf{T}_{\Sigma}(B)}(M)_{s}}.
$$
Let $R$ be an element of $[P]_{\Psi_{s}}$, then $(f_{s}(R),f_{s}(P))\in \Omega^{\mathbf{T}_{\Sigma}(B)}(M)_{s}$. Since $P\in [P]_{\Psi_{s}}$, then $f_{s}(P)\in f_{s}[[P]_{\Psi_{s}}]$ and we conclude that $f_{s}(R)\in K_{s}$. It follows that $R$ is a term in $(f^{-1}[K])_{s}$. For the converse, let $R$ be a term in $(f^{-1}[K])_{s}$, then $f_{s}(R)$ is a term in $K_{s}$, that is there exists a term $Q\in f_{s}[[P]_{\Psi_{s}}]$ such that $(Q,f_{s}(R))\in \Omega^{\mathbf{T}_{\Sigma}(B)}(M)_{s}$. Since $Q\in f_{s}[[P]_{\Psi_{s}}]$, there exists some term $P'\in [P]_{\Psi_{s}}$ such that $Q = f_{s}(P')$. That is, $(f_{s}(P'),f_{s}(R))\in \Omega^{\mathbf{T}_{\Sigma}(B)}(M)_{s}$. We conclude that $(P',R)\in \Psi_{s}$. But since $P'\in [P]_{\Psi_{s}}$, we can assert  that $R\in [P]_{\Psi_{s}}$. Hereby completing our proof.
\end{proof}


The following result about the fact that $\mathbf{Form}_{\mathrm{Lang_{r}}}(\Sigma)$ is a complete lattice may be proved in much the same way as those previously stated for other types of many-sorted formations. So the details are left to the reader.

\begin{proposition}
$\mathbf{Form}_{\mathrm{Lang_{r}}}(\Sigma)$ is a complete lattice.
\end{proposition}


\begin{proposition}\label{Cong2LangEnFinit}
Let $\mathfrak{F}$ be a formation of finite index congruences with respect to $\Sigma$, then the function $\mathcal{L}_{\mathfrak{F}}$ from   $\boldsymbol{\mathcal{U}}^{S}$ which assigns to $A\in \boldsymbol{\mathcal{U}}^{S}$ the subset
\begin{align} \mathcal{L}_{\mathfrak{F}}(A) &= \{L\in \mathrm{Sub}(\mathrm{T}_{\Sigma}(A))\mid \exists\, \Phi\in\mathfrak{F}(A)\,(L = [L]^{\Phi})\}\notag \\ &= \{L\in \mathrm{Sub}(\mathrm{T}_{\Sigma}(A))\mid \Omega^{\mathbf{T}_{\Sigma}(A)}(L)\in \mathfrak{F}(A)\}.\notag
\end{align}
of $\mathrm{Sub}(\mathrm{T}_{\Sigma}(A))$ is a formation of regular languages with respect to $\Sigma$.
\end{proposition}
\begin{proof}
Let $A$ be an $S$-sorted set and $L$ a language in $ \mathcal{L}_{\mathfrak{F}}(A)$. Then $\Omega^{\mathbf{T}_{\Sigma}(A)}(L)$ is a congruence in $\mathfrak{F}(A)$ and, therefore, it is a finite index congruence. It follows that all the languages in $ \mathcal{L}_{\mathfrak{F}}(A)$ are regular. The remaining properties follow from Proposition~\ref{Cong2LangBasic}.
\end{proof}

\begin{proposition}\label{Lang2CongEnFinit}
Let $\mathcal{L}$ be a formation of regular languages with respect to $\Sigma$, then the function $\mathfrak{F}_{\mathcal{L}}$ from $\boldsymbol{\mathcal{U}}^{S}$ which assigns to $A\in \boldsymbol{\mathcal{U}}^{S}$ the subset
$$
\mathfrak{F}_{\mathcal{L}}(A) = \{\Phi\in\mathrm{Cgr}_{\mathrm{fi}}(\mathbf{T}_{\Sigma}(A))\mid  \Phi\!-\!\mathrm{Sat}(\mathrm{T}_{\Sigma}(A))\subseteq \mathcal{L}(A)\}
$$
of $\mathrm{Cgr}_{\mathrm{fi}}(\mathbf{T}_{\Sigma}(A))$ is a formation of finite index congruences with respect to $\Sigma$.
\end{proposition}

\begin{proof}
Let $A$ be an $S$-sorted set. Then $\nabla^{\mathbf{T}_{\Sigma}(A)}\in \mathfrak{F}_{\mathcal{L}}(A)$ since we have that $\nabla^{\mathbf{T}_{\Sigma}(A)}\!-\!\mathrm{Sat}(\mathrm{T}_{\Sigma}(A))\subseteq \mathcal{L}(A)$ and $\nabla^{\mathbf{T}_{\Sigma}(A)}\in \mathrm{Cgr}_{\mathrm{fi}}(\mathbf{T}_{\Sigma}(A))$ (recall that $\mathrm{T}_{\Sigma}(A)$ is a regular language, which implies that $\mathrm{supp}_{S}(\mathbf{T}_{\Sigma}(A))$ is finite).

Let $\Phi$ be an element of $\mathfrak{F}_{\mathcal{L}}(A)$ and $\Psi$ a congruence on $\mathbf{T}_{\Sigma}(A)$ such that  $\Phi\subseteq \Psi$. Then, by Corollary~\ref{IncSat}, we have that $\Psi\!-\!\mathrm{Sat}(\mathrm{T}_{\Sigma}(A))$ is included in $\Phi\!-\!\mathrm{Sat}(\mathrm{T}_{\Sigma}(A))$. Moreover, $\Psi$ has finite index, since $\mathrm{T}_{\Sigma}(A)/\Psi$ is a quotient of $\mathrm{T}_{\Sigma}(A)/\Phi$. Hence, $\Psi\in \mathfrak{F}_{\mathcal{L}}(A)$.

Let $\Phi$ and $\Psi$ be congruences in $\mathfrak{F}_{\mathcal{L}}(A)$. Since $\mathrm{T}_{\Sigma}(A)/(\Psi\cap\Phi)$ can be subdirectly embedded in the product $\mathrm{T}_{\Sigma}(A)/\Psi\times \mathrm{T}_{\Sigma}(A)/\Phi$ we have that $\Phi\cap \Psi$ has finite index. Moreover, since $\Phi$ and $\Psi$ are congruences in $\mathfrak{F}_{\mathcal{L}}(A)$, we have that  $\Psi\!-\!\mathrm{Sat}(\mathbf{T}_{\Sigma}(A))$ and $\Phi\!-\!\mathrm{Sat}(\mathbf{T}_{\Sigma}(A))$ are included in $\mathcal{L}(A)$. Then, for every sort $s$ in $S$ and every term $P\in\mathrm{T}_{\Sigma}(A)_s$, the languages $\delta^{s,[P]_{\Phi_s}}$ and $\delta^{s,[P]_{\Psi_s}}$ belong to $\mathcal{L}(A)$, hence, from Corollary~\ref{CorolariAtoms}, the language $\delta^{s, [P]_{(\Psi\cap\Phi)_s}}$ belongs to $\mathcal{L}(A)$. On the other hand, since $\Phi\cap \Psi$ has finite index, any $\Phi\cap \Psi$-saturated language can be represented as a finite union of such Kronecker's deltas. Hence, from  Corollary~\ref{CorolariAlgebraBooleana}, we conclude that $L$ is a language in $\mathcal{L}(A)$. Therefore all $\Phi\cap\Psi$-saturated languages belong to $\mathcal{L}(A)$.

Finally, let $B$ be an $S$-sorted set, $\Theta\in \mathfrak{F}_{\mathcal{L}}(B) $, and $f$ an homomorphism $f\colon \mathbf{T}_{\Sigma}(A)\mor \mathbf{T}_{\Sigma}(B)$ such that
$$
\mathrm{pr}^{\Theta}\circ f\colon \mathbf{T}_{\Sigma}(A)\mor \mathbf{T}_{\Sigma}(B)/\Theta
$$
is an epimorphism. Since $\mathbf{T}_{\Sigma}(A)/\mathrm{Ker}(\mathrm{pr}^{\Theta}\circ f)$ is isomorphic to $\mathbf{T}_{\Sigma}(B)/\Theta$, we have that $\mathrm{Ker}(\mathrm{pr}^{\Theta}\circ f)$ has finite index. We next prove that if $L$ is an element of $\mathrm{Ker}(\mathrm{pr}^{\Theta}\circ f)\!-\!\mathrm{Sat}(\mathrm{T}_{\Sigma}(A))$, then $L$ belongs to $\mathcal{L}(A)$. In fact, since $\Theta\in\mathfrak{F}_{\mathcal{L}}(B)$, we have that $[f[L]]^{\Theta}$ is a language in $\mathcal{L}(B)$. On the other hand, since $\mathrm{pr}^{\Theta}\circ f$ is surjective, $\Theta\subseteq \Omega^{\mathbf{T}_{\Sigma}(B)}([f[L]]^{\Theta})$, and the triangle in the following diagram commutes
$$\xymatrix@C=40pt{
\mathbf{T}_{\Sigma}(A) \ar[r]^-{f} & \mathbf{T}_{\Sigma}(B) \ar[d]_-{\mathrm{pr}^{\Theta}} \ar[rrd]^*[l]{\mathrm{pr}^{\Omega^{\mathbf{T}_{\Sigma}(B)}([f[L]]^{\Theta})}} & {} & {}\\
{} & \mathbf{T}_{\Sigma}(B)/\Theta \ar[rr]_-{\mathrm{p}^{\Theta,\Omega^{\mathbf{T}_{\Sigma}(B)}([f[L]]^{\Theta})}} & {} & \mathbf{T}_{\Sigma}(B)/\Omega^{\mathbf{T}_{\Sigma}(B)}([f[L]]^{\Theta})
}
$$
we have that $\mathrm{pr}^{\Omega^{\mathbf{T}_{\Sigma}(B)}([f[L]]^{\Theta})}\circ f$ is surjective. Thus, since $\mathcal{L}$ is a formation of regular languages with respect to $\Sigma$, every   $\mathrm{Ker}(\mathrm{pr}^{\Omega^{\mathbf{T}_{\Sigma}(B)}([f[L]]^{\Theta})}\circ f)$-saturated language belong to $\mathcal{L}(A)$. We claim that $L$ is $\mathrm{Ker}(\mathrm{pr}^{\Omega^{\mathbf{T}_{\Sigma}(B)}([f[L]]^{\Theta})}\circ f)$-saturated. Let $s$ be a sort in $S$, and $P$, $Q\in\mathrm{T}_{\Sigma}(A)_{s}$ such that $P\in L_{s}$ and $(P,Q)\in \mathrm{Ker}(\mathrm{pr}^{\Omega^{\mathbf{T}_{\Sigma}(B)}([f[L]]^{\Theta})}\circ f)$. Then $(f_s(P),f_s(Q))\in \Omega^{\mathbf{T}_{\Sigma}(B)}([f[L]]^{\Theta})_{s}\subseteq\mathrm{Ker}(\mathrm{ch}^{[f[L]]^{\Theta}})_{s}$. However, as we know that $P\in L_{s}$, we have that $f_{s}(P)\in f_{s}[L_{s}]=f[L]_{s}\subseteq [f[L]]^{\Theta}_{s}$. Hence, $f_{s}(Q)\in [f[L]]^{\Theta}_{s}$. Consequently, there exists a term $R\in f[L]_{s}$ such that $(f_s(Q),R)\in \Theta_{s}$. Now, from $R\in f[L]_{s} = f_{s}[L_{s}]$ we infer that there exist a term $R'\in L_{s}$ such that $f_{s}(R')=R$. Hence, $(f_{s}(Q),f_{s}(R'))\in \Theta_{s}$ and $(Q,R')\in \mathrm{Ker}(\mathrm{pr}^{\Theta}\circ f)_s$. But $L$ is $\mathrm{Ker}(\mathrm{pr}^{\Theta}\circ f)$-saturated, therefore $Q\in L_{s}$.
\end{proof}

Finally, we prove that there exists an isomorphism between the complete lattices $\mathbf{Form}_{\mathrm{Cgr}_{\mathrm{fi}}}(\Sigma)$ and $\mathbf{Form}_{\mathrm{Lang_{r}}}(\Sigma)$, from which it follows that $\mathbf{Form}_{\mathrm{Lang_{r}}}(\Sigma)$ is also an algebraic lattice.

\begin{proposition}
The complete lattices $\mathbf{Form}_{\mathrm{Cgr}_{\mathrm{fi}}}(\Sigma)$ and $\mathbf{Form}_{\mathrm{Lang_{r}}}(\Sigma)$ are isomorphic.
\end{proposition}

\begin{proof}
Let us first prove that, for every $\mathfrak{F}\in \mathrm{Form}_{\mathrm{Cgr}_{\mathrm{fi}}}(\Sigma)$, $\mathfrak{F} = \mathfrak{F}_{\mathcal{L}_{\mathfrak{F}}}$. By definition, $\mathcal{L}_{\mathfrak{F}}$ is such that, for every  $A\in \boldsymbol{\mathcal{U}}^{S}$,  $\mathcal{L}_{\mathfrak{F}}(A)$ is
$$
\mathcal{L}_{\mathfrak{F}}(A) = \{L\in \mathrm{Sub}(\mathrm{T}_{\Sigma}(A))\mid \exists\, \Phi\in\mathfrak{F}(A)\,(L = [L]^{\Phi})\}
$$
On the other hand, by definition, we have that, for every  $A\in \boldsymbol{\mathcal{U}}^{S}$,  $\mathfrak{F}_{\mathcal{L}_{\mathfrak{F}}}(A)$ is
$$
 \mathfrak{F}_{\mathcal{L}_{\mathfrak{F}}}(A) = \{\Phi\in\mathrm{Cgr}_{\mathrm{fi}}(\mathbf{T}_{\Sigma}(A))\mid\Phi\!-\!\mathrm{Sat}(\mathrm{T}_{\Sigma}(A))\subseteq \mathcal{L}_{\mathfrak{F}}(A)\}.
$$

Let us prove that $\mathfrak{F} \leq \mathfrak{F}_{\mathcal{L}_{\mathfrak{F}}}$. Let $A$ be an element of  $\boldsymbol{\mathcal{U}}^{S}$ and let $\Phi$ be a congruence in $\mathfrak{F}(A)$. Then $\Phi$ has finite index and all $\Phi$-saturated languages belong to $\mathcal{L}_{\mathfrak{F}}(A)$. Therefore $\Phi$ belongs to $\mathfrak{F}_{\mathcal{L}_{\mathfrak{F}}}(A)$.

Let us now prove that $\mathfrak{F}_{\mathcal{L}_{\mathfrak{F}}} \leq \mathfrak{F}$. Let $\Phi$ be a congruence in $\mathfrak{F}_{\mathcal{L}_{\mathfrak{F}}}(A)$. Then $\Phi$ has finite index and all $\Phi$-saturated languages belong to $\mathcal{L}_{\mathfrak{F}}(A)$. Since, for every $s\in S$ and every term $P\in\mathrm{T}_{\Sigma}(A)_s$, the language $\delta^{s,[P]_{\Phi_s}}$ is $\Phi$-saturated, there are congruences $\Psi^{(s,[P]_{\Phi_s})}$ in $\mathfrak{F}(A)$ for which $\delta^{s,[P]_{\Phi_s}}$ is $\Psi^{(s,[P]_{\Phi_s})}$-saturated. From this it follows that $\Psi^{(s,[P]_{\Phi_s})}\subseteq \Omega^{\mathbf{T}_{\Sigma}(A)}(\delta^{s,[P]_{\Phi_s}})$. Hence, by Proposition~\ref{RepCongInterCCogKroneckerDelta}, we have that
$$
\textstyle
\Phi = \bigcap\{\Omega^{\mathbf{T}_{\Sigma}(A)}(\delta^{s,[P]_{\Phi_{s}}})\mid s\in S\! \And\! P\in\mathrm{T}_{\Sigma}(A)_s\}.
$$
But since $\Phi$ has finite index, the last intersection is finite. Consequently $\Phi$ is a congruence in $\mathfrak{F}(A)$.

We next prove that, for every $\mathcal{L}\in \mathrm{Form}_{\mathrm{Lang_{r}}}(\Sigma)$, $\mathcal{L} = \mathcal{L}_{\mathfrak{F}_{\mathcal{L}}}$. By definition, $\mathfrak{F}_{\mathcal{L}}$ is such that, for every  $A\in \boldsymbol{\mathcal{U}}^{S}$,  $\mathfrak{F}_{\mathcal{L}}(A)$ is
$$
\mathfrak{F}_{\mathcal{L}}(A)=\{\Phi\in\mathrm{Cgr}_{\mathrm{fi}}(\mathbf{T}_{\Sigma}(A))\mid  \Phi\!-\!\mathrm{Sat}(\mathbf{T}_{\Sigma}(A))\subseteq \mathcal{L}(A)\}.
$$
On the other hand, by definition, for every  $A\in \boldsymbol{\mathcal{U}}^{S}$, we have that
$$
\mathcal{L}_{\mathfrak{F}_{\mathcal{L}}}(A) = \{L\in \mathrm{Sub}(\mathrm{T}_{\Sigma}(A))\mid \exists\, \Phi\in\mathfrak{F}_{\mathcal{L}}(A)\,(L = [L]^{\Phi})\}
$$
Let us prove that $\mathcal{L}\leq \mathcal{L}_{\mathfrak{F}_{\mathcal{L}}}$. Let $L$ be a language in $\mathcal{L}(A)$. Then $L$ is regular, hence $\Omega^{\mathbf{T}_{\Sigma}(A)}(L)$ has finite index and $\Omega^{\mathbf{T}_{\Sigma}(A)}(L)\!-\!\mathrm{Sat}(\mathrm{T}_{\Sigma}(A))\subseteq \mathcal{L}(A)$. It follows that $\Omega^{\mathbf{T}_{\Sigma}(A)}(L)$ is a congruence in $\mathfrak{F}_{\mathcal{L}}(A)$. Finally, $L\in \mathcal{L}_{\mathfrak{F}_{\mathcal{L}}}(A)$ since it is an $\Omega^{\mathbf{T}_{\Sigma}(A)}(L)$-saturated language.

Let us now prove that $\mathcal{L}_{\mathfrak{F}_{\mathcal{L}}}\leq \mathcal{L}$. Let $L$ be a language in $\mathcal{L}_{\mathfrak{F}_{\mathcal{L}}}(A)$. Then there exists a congruence $\Phi$ in $\mathfrak{F}_{\mathcal{L}}(A)$ such that  $L$ is $\Phi$-saturated. Since $\Phi$ is a congruence in $\mathfrak{F}_{\mathcal{L}}(A)$, all $\Phi$-saturated languages belong to $\mathcal{L}(A)$. Thus, $L\in\mathcal{L}(A)$.

If we denote by $\vartheta_{\Sigma}$ the bijection from $\mathrm{Form}_{\mathrm{Cgr}_{\mathrm{fi}}}(\Sigma)$ to $\mathrm{Form}_{\mathrm{Lang_{r}}}(\Sigma)$, then it is straightforward to prove that, for every $\mathfrak{F}$, $\mathfrak{G}\in \mathrm{Form}_{\mathrm{Cgr}_{\mathrm{fi}}}(\Sigma)$, $\mathfrak{F}\leq \mathfrak{G}$ if, and only if,  $\vartheta_{\Sigma}(\mathfrak{F})\leq \vartheta_{\Sigma}(\mathfrak{G})$ (or, what is equivalent, that the bijection $\vartheta_{\Sigma}$ is such that both $\vartheta_{\Sigma}$ and $\vartheta_{\Sigma}^{-1}$ are order-preserving). Therefore the complete lattices $\mathbf{Form}_{\mathrm{Cgr}_{\mathrm{fi}}}(\Sigma)$ and $\mathbf{Form}_{\mathrm{Lang_{r}}}(\Sigma)$ are isomorphic.
\end{proof}

From the just stated proposition together with Corollary~\ref{FormCgrfiAlg}, we obtain immediately the following corollary.

\begin{corollary}
$\mathbf{Form}_{\mathrm{Lang_{r}}}(\Sigma)$ is an algebraic lattice.
\end{corollary}


\begin{remark}
Let $J$ be a nonempty set in $\ensuremath{\boldsymbol{\mathcal{U}}}$ and $(\mathcal{L}_{j})_{j\in J}$ an upward directed family in $\mathrm{Form}_{\mathrm{Lang_{r}}}(\Sigma)$. Then the function $\mathcal{L}$ defined, for every $A\in \boldsymbol{\mathcal{U}}^{S}$, as $\mathcal{L}(A) = \bigcup_{j\in J}\mathcal{L}_{j}(A)$ is the least upper bound of $(\mathcal{L}_{j})_{i\in J}$ in $\mathbf{Form}_{\mathrm{Lang_{r}}}(\Sigma)$. Moreover, since $\mathbf{Form}_{\mathrm{Lang_{r}}}(\Sigma)$ is an algebraic lattice, it is meet-continuous, i.e., for every $\mathcal{L}$ in $\mathrm{Form}_{\mathrm{Lang_{r}}}(\Sigma)$, every nonempty set $J$ in $\ensuremath{\boldsymbol{\mathcal{U}}}$, and every upward directed family $(\mathcal{L}_{j})_{j\in J}$ in $\mathrm{Form}_{\mathrm{Lang_{r}}}(\Sigma)$ we have that
$$
\textstyle
\mathcal{L}\wedge\bigvee_{j\in J}\mathcal{L}_{j} = \bigvee_{j\in J} (\mathcal{L}\wedge \mathcal{L}_{j}).
$$
\end{remark}


\textbf{Acknowledgement.}
We would like to thank our dear friend Jos\'{e} Garc\'{\i}a Roca---example of intelligence, goodness, and integrity---for correcting our English---needless to say, any remaining errors are our sole responsibility. But, above all, for his tireless encouragement and invaluable continued moral support.


\end{document}